%% file: main.tex
\documentclass[format=acmsmall,review=false]{acmart}
\usepackage{acm-ec-23}
\usepackage{etoc}
\usepackage{booktabs} 
\usepackage[ruled]{algorithm2e} 

\SetAlFnt{\small}
\SetAlCapFnt{\small}
\SetAlCapNameFnt{\small}
\SetAlCapHSkip{0pt}
\IncMargin{-\parindent}

\input{ec-23-headers}
\setcitestyle{authoryear}

\makeatletter
\def\@ACM@checkaffil{
    \if@ACM@instpresent\else
    \ClassWarningNoLine{\@classname}{No institution present for an affiliation}%
    \fi
    \if@ACM@citypresent\else
    \ClassWarningNoLine{\@classname}{No city present for an affiliation}%
    \fi
    \if@ACM@countrypresent\else
        \ClassWarningNoLine{\@classname}{No country present for an affiliation}%
    \fi
}
\makeatother

\title[The Computational Complexity of Kakutani Fixed Points]{The Computational Complexity  of Multi-player Concave Games and Kakutani Fixed Points}

\author{Christos Papadimitriou }
\affiliation{%
  \institution{Columbia University}
  \state{New York}
  \city{New York}
  \country{USA}
}
\email{christos@cs.columbia.edu}

\author{Emmanouil-Vasileios Vlatakis-Gkaragkounis}
\affiliation{%
  \institution{University California, Berkeley}
  \state{California}
  \city{Berkeley}
  \country{USA}
}
\email{emvlatakis@berkeley.edu}

\author{Manolis Zampetakis}
\affiliation{%
  \institution{University California, Berkeley}
    \state{California}
  \city{Berkeley}
  \country{USA}
}
\email{mzampet@berkeley.edu}

\begin{abstract}
  Kakutani's Fixed Point theorem is a fundamental theorem in topology with numerous applications in game theory and economics. Computational formulations of Kakutani exist only in special cases and are too restrictive to be useful in reductions. In this paper, we provide a general computational formulation of Kakutani's Fixed Point Theorem and we prove that it is PPAD-complete. As an application of our theorem we are able to characterize the computational complexity of the following fundamental problems: 
  \begin{enumerate}
    \item \textbf{Concave Games.} Introduced by the celebrated works of Debreu and Rosen in the 1950s and 60s, concave $n$-person games have found many important applications in Economics and Game Theory. We characterize the computational complexity of finding an equilibrium in such games. We show that a general formulation of this problem belongs to PPAD, and that finding an equilibrium is PPAD-hard even for a rather restricted games of this kind: strongly-concave utilities that can be expressed as multivariate polynomials of a constant degree with axis aligned box constraints.
    \item \textbf{Walrasian Equilibrium.} Using Kakutani's fixed point Arrow and Debreu we resolve an open problem related to Walras's theorem on the existence of price equilibria in general economies. There are many results about the PPAD-hardness of Walrasian equilibria, but the inclusion in PPAD is only known for piecewise linear utilities. We show that the problem with general convex utilities is in PPAD.
  \end{enumerate}
Along the way we provide a Lipschitz continuous version of Berge's maximum theorem that may be of independent interest.
\end{abstract}

\begin{document}

\addtocontents{toc}{\protect\setcounter{tocdepth}{0}}
\begin{titlepage}

\maketitle

\end{titlepage}

\section{Introduction} \label{sec:intro}
During the 1950s and 1960s, game theory and mathematical economics grew monumentally and hand-in-hand. The spark for this was young John F.~Nash who during his last year in grad school did three vastly consequential things \cite{nash1950,nash1951non}: (a) he defined the Nash equilibrium, a solution concept that would dominate game theory for the next half century; (b) he proved that it is  universal; and (c) he introduced fixed point theorems to the arsenal of mathematical economics --- the same year he also discovered Nash bargaining.  In retrospect,  it may have been very fortunate that Nash at first used Kakutani's fixed point theorem to prove universality, before realizing that the simpler and older theorem by Brouwer suffices for his purposes. His work, and the use of these mathematical tools, inspired Arrow and Debreu to finally prove Walras's hypothesis on prices 
and proceed to the articulation of the fundamental theorems of welfare economics and the quest for price adjustments that lead to equilibrium prices \cite{arrowdebreu1954}. This, in turn, enabled game theorists and mathematicians to circle back and  generalize Nash's theorem by defining very general classes of games guaranteed to have equilibria \cite{debreu1952,rosen1965existence}. 
\smallskip

Half a century after Nash's theorem and the ensuing equilibrium theorems, computer scientists started to think computationally about these two important areas, and special complexity classes (TFNP and PPAD among others) had to be defined to accommodate them and articulate their computational narrative.  Algorithmic game theory would eventually classify the complexity of these important concepts in economics as PPAD-complete, see e.g., \cite{daskalakis2009complexity,ChenPY17}. With the exception of bimatrix games \cite{chen2009settling}, an arbitrary small approximation parameter must be supplied to adapt to the numerical complexity of these problems; without such maneuver, the problems are complete for FIXP, a complexity class which seems to hover beyond NP \cite{etessami2010complexity}. A little earlier, Geanakoplos \cite{geanakoplos2003nash} had given a proof of the existence of such general equilibria using solely Brouwer's fixed point theorem.
\smallskip

Surprisingly, the prelude and the final act of the drama outlined in the opening paragraph --- that is to say, Kakutani's theorem and the sweepingly general games defined by Debreu, Rosen, and Fan --- have {\em not} been treated adequately by this computational theory, {\em and fixing this is our goal in this paper}.

At a first glance, Geanakoplos's proof of the existence of price equilibria directly from Brouwer seems very useful in this direction, since it is well known that finding approximate Brouwer fixed points is in PPAD. Nevertheless, a closer look at Geanakoplos's proof reveals that his construction of the Brouwer function involves an exact optimization oracle which is not computationally efficient. Hence, understanding the computational complexity of such general equilibria remains very much an open problem --- and that is one of the questions we resolve in this paper.
\smallskip

The convex games of Debreu and Rosen are so general that their PPAD-hardness (and FIXP-hardness for exact solution) has never been in doubt, and so one must focus on the two remaining problems: (a) are their approximate versions in PPAD? and (b) how much can one simplify these problems and retain full PPAD-hardness?  Here we resolve both legs of this problem. For (a), even though the inclusion of convex games to FIXP was settled in \cite{filos2021fixp}, to show inclusion of the approximate version in PPAD we must use our result on Kakutani.  As for (b), we show PPAD-completeness of \textit{strongly}-convex games when the utility is a low-degree polynomial. 
\smallskip

Turning now to our result on Kakutani, in the early days of TFNP a problem named {\sc Kakutani} {\em was} defined and sketched to be in PPAD \cite{Papadimitriou1994}; however, this result concerned a simplified version, in which the value of the set-valued map is a convex polytope explicitly given by its vertices, and this version is useless when dealing with general convex games.  A proper definition of Kakutani is needed, and for this one must resort to computational convex geometry \cite{grotschel2012geometric}. 

In Section 2 we develop the required machinery --- which turns out to be rather extensive --- first to define {\sc Kakutani} appropriately, and second to prove that it is in PPAD (and, of course, PPAD-complete).

\subsection{Our Contribution}

\noindent To summarize, our contributions about {\sc Kakutani} and its applications are:
    
    
    
    
    

\begin{enumerate}
\setlength{\parskip}{.2ex}
    \item \textbf{Complexity of Finding Kakutani Fixed Points.} Formulating a general version of the approximate \hyperref[def:Kakutani]{\textsc{Kakutani}} problem, and proving that it is in PPAD, a result that is likely to enable more proofs that other general fixpoint problems lie in PPAD;
    
    \begin{inftheorem}[Theorem \ref{thm:Kakutani:member}]
    Finding Kakutani fixed points is PPAD-complete.
    \end{inftheorem}
    
    \item \textbf{Complexity of Equilibria in Concave Games} Classifying the complexity of the classical general multi-player convex games using our results for the approximate {\sc Kakutani} problem;
    
    \begin{inftheorem}[Theorem \ref{main:thm:concave:complete} (Membership) : Theorem \ref{thm:concave:inclusion}]
      Finding equilibria in general concave games is in PPAD.
    \end{inftheorem}
     Identifying concrete concave utility functions, presented as simply bounded polynomials, for which the corresponding multiplayer game is hard.  This may render this problem an attractive starting point for further reductions.
    
    \begin{inftheorem}[Theorem \ref{main:thm:concave:complete} (Hardness) : Theorem \ref{thm:strong-concave-ppad-complete}]
      Finding equilibria even in strongly concave games with utility functions represented as sums of monomials of bounded degree is PPAD-hard.
    \end{inftheorem}

    \item \textbf{Complexity of Competitive Equilibria in Walrasian Economy - Inclusion.}
    Classifying the complexity of the standard Walrasian Equilibrium in supply-demand \cite{arrowdebreu1954} markets using our results on \hyperref[def:Kakutani]{\textsc{Kakutani}} and a novel robust version of Berge's Maximum Theorem.

     \begin{inftheorem}[Theorem \ref{thm:warlasian:inclusion}]
      Finding equilibria in exchange economy based on Walras Model is in PPAD.
    \end{inftheorem}
    
\end{enumerate}

\noindent\textbf{Organization}. In Section \ref{sec:Kakutani} we present our computational formulation of Kakutani's theorem and we characterize its computational complexity. In Section \ref{sec:concave} we characterize the complexity of computing equilibria in concave games. In Section \ref{sec:Walras} we conclude our work with the complexity of market equilibria in Walras Economy. 

\subsection{Our Techniques}

The object of study in Kakutani's fixed point theorem is
a point-to-set map. That is, a function $F$ that takes as input
a vector $\vecx \in \R^d$ and outputs a convex set 
$S \subseteq \R^d$. 
The main reason for the absence of a general formulation of a
computational problem for Kakutani is that all the simple ways
to explicitly and succinctly represent a convex set, such as the
convex hull of a set of points or a convex polytope defined from
linear inequalities, end up being very restrictive and useless
to capture the important actual applications of Kakutani. For 
this reason in this paper we use a succinct but 
\textit{implicit} way to represent the convex set $S$ via a
polynomial sized circuit that computes weak separation oracles
for $S$. This formulation is general enough to capture virtually
all the game theoretic applications of Kakutani's fixed point
theorem. Nevertheless, this formulation introduces many
technical difficulties that  arise from handling the errors of
projections to convex sets and optimization of convex functions
using these weak separation oracles. 

  To begin with, we observe that a black box application of the 
results of \cite{grotschel1981ellipsoid, grotschel2012geometric}
for the ellipsoid method does not suffice to analyze the  
complexity of this general computational formulation of 
Kakutani's theorem. For this reason, we provide a novel analysis
of the  errors of the projection and optimization algorithms for
convex sets and convex functions and carefully show that we can 
tolerate all the approximation errors that lead to an
approximate Kakutani fixed point. 

Next, to apply our computational formulation of 
Kakutani's theorem, e.g., in concave games, we need to define
and prove a stronger Lipschitz version of Berge's Maximum
Theorem (Theorem \ref{thm:our-berge-for-optimization}). We 
believe that this theorem is of independent interest and will 
have applications to other computational problems. 

  To show the inclusion of concave games and Walrasian equilibirum to PPAD, we 
show that, rather surprisingly, the smoothness, i.e, Lipschitzness of the gradient, of
the objective functions is not necessary. To prove this we show 
that after an $\ell_2$-regularization step of the objective functions, similar to the regularization used in \cite{geanakoplos2003nash},
we can apply the aforementioned Lipschitz version of Berge's 
Maximum Theorem to show the reduction to our computational formulation Kakutani's theorem, which establishes inclusion to PPAD.

  For our hardness result, we first prove that finding
Brouwer fixed points with constant approximation error when the
Brouwer map is represented as a constant degree polynomial is
still PPAD-hard, and then apply this PPAD-hardness to construct the instance of the concave games problem that we need.

\subsection{A Meta-Approach For Future Problems}
We're finishing our summary with a plan for a new approach to make our method more modular. This plan helps use our techniques to solve future problems more easily and adaptably.

The meta-approach we propose is based on a series of steps, some of which are dependent on the specific problem in question. The following outline provides a more detailed overview:
\begin{enumerate}
\item \textbf{Problem-dependent Step:} Our first step involves the construction of a convex, compact argmax set-valued map derived from a concave function. The essence of this step is to ensure the map's non-emptiness, which is crucial as it implies the existence of the equilibrium you're seeking. This equilibrium is then verified through the application of our Kakutani’s oracle. The concept of utilizing a set-valued map is to cater to a broad array of problem sets, each presenting its unique equilibrium conditions.

\item \textbf{Problem-dependent Step:} The second step is centered on the proof of Lipschitz continuity for your constrained, parametrized set-valued map. This proof is pivotal for maintaining the stability of our solutions and ensuring their reliability. In the context of Concave Games, the Lipschitzness was proven by establishing bounds on the discrepancy between the positive and negative dilation of strategy sets. This helped us maintain a balanced state of equilibrium. On the other hand, for Walsarian Markets, we employed Hoffman Error bounds to prove Lipschitzness. 

\item \textbf{General Step:} Lastly, the Lipschitzness of the $\epsilon$-argmax$(\theta)$ operator is proven using  our Robust Berge Theorem and an $O(\epsilon)$-regularizer. This proof essentially shows that the operator is well-behaved and can manage small changes in the input without causing large fluctuations in the output. The existence of an approximate equilibrium is thus implied. This step transcends the specifics of any single problem and applies to all, providing a solid foundation for the applicability of our approach.
\end{enumerate}
In conclusion, the meta-approach provides a robust and flexible roadmap for applying our techniques to future problems. Each step is designed to ensure both adaptability and rigor, enabling the method to address an expansive range of scenarios effectively.



\section{Preliminaries}
\noindent \textbf{Notation.} For any compact and convex $K \subseteq \R^d$ and 
$B \in \R_+$, we define $L_{\infty}(K, B)$ to be the set of all continuous 
functions $f : K \to \R$ such that 
$\max_{\vec{x} \in K} \abs{f(\vec{x})} \le B$. 
Additionally, let $vol(K)$ represents the Lebesgue volume measure of the set $K$.
When $K = [0, 1]^d$, we use
$L_{\infty}(B)$ instead of $L_{\infty}([0, 1]^d, B)$ for ease of notation. For
$p > 0$, we define 
$\diam_{p}(K) = \max_{\vec{x}, \vec{y} \in K} \norm{\vec{x} - \vec{y}}_{p}$,
where $\norm{\cdot}_p$ is the usual $\ell_p$-norm of vectors. For an alphabet
set $\Sigma$, the set $\Sigma^*$, called the Kleene star of $\Sigma$, is equal
to $\cup_{i = 0}^{\infty} \Sigma^i$. For any string $\vecq \in \Sigma^*$ we use
$\abs{\vecq}$ to denote the length of $\vecq$. We use the symbol $\log(\cdot)$
for base $2$ logarithms and $\ln(\cdot)$ for the natural logarithm. We use 
$[n] \triangleq \{1, \ldots, n\}$, $\nm{n} \triangleq \{0, \dots, n - 1\}$, and
$[n]_0 \triangleq \{0, \dots, n\}$.  
We next define the complexity classes $\FNP$ and $\PPAD$, as well as the notion of reductions that we use in this paper.

\begin{definition}[Search Problems - ${\FNP}$] \label{def:FNP}
    A binary relation $\calQ \subseteq \set{0, 1}^* \times \set{0, 1}^*$ is in
  the class $\FNP$ if (i) for every $\vecx, \vecy \in \set{0, 1}^*$ such that
  $(\vecx, \vecy) \in \calQ$, it holds that 
  $\abs{\vecy} \le \poly(\abs{\vecx})$; and (ii) there exists an algorithm that
  verifies whether $(\vecx, \vecy) \in \calQ$ in time 
  $\poly(\abs{\vecx}, \abs{\vecy})$. The \textit{search problem} associated with
  a binary relation $\calQ$ takes some $\vecx$ as input and requests as output
  some $\vecy$ such that $(\vecx, \vecy) \in \calQ$ or outputting $\bot$ if no
  such $\vecy$ exists. 
  The \textit{decision problem} associated with $\calQ$
  takes some $\vecx$ as input and requests as output the bit $1$, if there
  exists some $\vecy$ such that $(\vecx, \vecy) \in \calQ$, and the bit $0$,
  otherwise. The class $\NP$ is defined as the set of decision problems
  associated with relations $\calQ \in \FNP$.
\end{definition}

\begin{definition}[Polynomial-Time Reductions]
    A search problem $P_1$ is \textit{polynomial-time reducible} to a search
  problem $P_2$ if there exist polynomial-time computable functions 
  $f : \set{0, 1}^* \to \set{0, 1}^*$ and 
  $g : \set{0, 1}^* \times \set{0, 1}^* \times \set{0, 1}^* \to \set{0, 1}^*$
  with the following properties: (i) if $\vecx$ is an input to $P_1$, then
  $f(\vecx)$ is an input to $P_2$; and (ii) if $\vecy$ is a solution to $P_2$ on
  input $f(\vecx)$, then $g(\vecx, f(\vecx), \vecy)$ is a solution to $P_1$ on
  input $\vecx$.
\end{definition}

{\small
\begin{nproblem}[\textsc{End-of-a-Line}]
  \textsc{Input:} Binary circuits $\calC_S$ (for successor) and $\calC_P$ (for
  predecessor) with $n$ inputs and $n$ outputs.
  \smallskip

  \noindent \textsc{Output:} One of the following:
  \begin{enumerate}
    \item[0.] $\vec{0}$ if either both $\calC_P(\calC_S(\vec{0}))$ and
              $\calC_S(\calC_P(\vec{0}))$ are equal to $\vec{0}$, or if they are
              both different than $\vec{0}$, where $\vec{0}$ is the all-$0$
              string.
    \item[1.] a binary string $\vecx \in \set{0, 1}^n$ such that 
              $\vecx \neq \vec{0}$ and $\calC_P(\calC_S(\vecx)) \neq \vecx$ or
              $\calC_S(\calC_P(\vecx)) \neq \vecx$.
  \end{enumerate}
\end{nproblem}}
\medskip  
\noindent Finally, $\PPAD$ is the set of problems in $\FNP$ which can be polynomial-time reduced to 
\textsc{End-of-a-Line}.

\section{Computational Complexity of Kakutani Fixed Points}
\label{sec:Kakutani}

In this section, we define a computational version of Kakutani's Fixed Point
Theorem that is more general than the one in \cite{Papadimitriou1994},
and therefore more useful for showing the inclusion in $\PPAD$ of equilibrium problems, as we will see in Section \ref{sec:concave} and Section \ref{sec:Walras}.

Kakutani's Theorem generalizes Brouwer's theorem \cite{brouwer1911abbildung,Knaster1929} 
to set-valued functions (also known as \emph{correspondences}). In the following Section \ref{sec:ContinuitySetValued} we discuss conceptions of continuity for set-valued maps, while in Section \ref{sec:RepresentationSetValued} we present a computationally efficient way to represent such maps using tools from convex optimization. This leads to our computational version Kakutani's fixed point theorem and the proof that it is $\PPAD$-complete in \ref{sec:Kakutani:PPAD}. Finally, in Section \ref{sec:Berge} we present a robust version of the celebrated Berge's Maximum Theorem which is an important tool as well in showing the $\PPAD$ inclusions in Section \ref{sec:concave} and Section \ref{sec:Walras}.

\subsection{Topological Kakutani's Fixed Point Theorem \& Continuity in Set-Valued Maps } \label{sec:ContinuitySetValued}

Let's define first formally the notion of a \emph{set-valued map} (also known as \emph{multivalued function} or the single-worded \emph{correspondence}) together with the notions of continuity in set-valued maps.

\begin{definition}[Correspondence] \label{def:correspondence}
    Let $\calX$ and $\calY$ be topological spaces. A \emph{correspondence} or a \emph{set-valued map} $\Phi$ from $\calX$ to $\calY$  is a map that assigns to each element $\vecx\in\calX$ a (possibly empty) subset $\Phi(\vecx)\subset \calY$. To distinguish a correspondence notionally from a single-valued function, we adopt the notation of $\Phi:\calX\rightrightarrows\calY$ instead of  the set-valued version of $\Phi:\calX\to \wp(\calY)$.
\end{definition}

The dual concepts of \emph{upper semi-continuity} and \emph{lower semi-continuity} are the analogue of continuity in the domain of correspondences. For single-valued function these notions are both equivalent with continuity but for set-valued maps they are not equivalent anymore. A correspondence that has {both} properties is said to be \emph{continuous}.

\begin{definition}[Semi-continuity]\label{def:semi-continuity}
Following the presentation of Chap.~VI in \cite{berge1997topological}, we have that:
\begin{enumerate}[topsep=0pt,left=0ex,itemsep=0pt,parsep=.0ex]
\item A correspondence $F:\calX\rightrightarrows\calY$ is called \emph{upper semi-continuous} (u.s.c.) at a point $\vecx_o\in\calX$ if and only if for any open subset $\calV$ of $\calY$ with $F(\vecx_o)\subseteq \calV$ there is a neighborhood of $\vecx_o$, denoted as $\calU(\vecx_o)$, such that $F(\vecx)\subseteq \calV$ for all $\vecx\in\calU(\vecx_o)$.
\item A correspondence $F:\calX\rightrightarrows\calY$ is called \emph{lower semi-continuous} (l.s.c.) at a point $\vecx_o\in\calX$ if and only if for  any open subset $\calV$ of $\calY$ with $F(\vecx_o)\cap \calV\neq \emptyset$ there is a neighborhood of $\vecx_o$, denoted as $\calU(\vecx_o)$, such that $F(\vecx)\cap \calV\neq \emptyset$ for all $\vecx\in\calU(\vecx_o)$.
\end{enumerate}
\end{definition}

\noindent We can now state Kakutani's  theorem:
\begin{theorem}[Shizuo Kakutani's Fixed Point Theorem \cite{kakutani1941generalization}]\label{thm:Kakutani'sTheorem}
Let $\calX\subset \R^d$ be compact and convex. If $F:\calX\rightrightarrows \calX$ is an upper-semi continuous correspondence that has nonempty, convex, compact values then $F$ has a fixed point, i.e $\vecx^\star\in F(\vecx^\star)$.
\end{theorem}

As is common in the computational complexity of finding fixed points, we shall be seeking {\em approximate} Kakutani fixed points --- since the exact solutions to fixed point problems is usually $\FIXP$-hard \cite{etessami2010complexity} which lies above $\NP$. This requires us to define certain notions of distance (stepping away from the pure topological nature of the theorem):

\begin{definition}\label{def:Metrics} 
\begin{enumerate}[topsep=0pt,left=0ex,itemsep=0pt,parsep=.0ex]
\item Let $\metric(\vecx, \vecz)$ be the metric between any points in $\vecx,\vecz\in \R^d$.
   \item 
    Let $\calS$ be a convex, non-empty and compact set.  
    We define the \emph{set-point distance} of a point $\vecx$ from a set $\calS$ to be 
$\dist
(\vecx, S)
:= \displaystyle\inf_{\vecz \in S}
\metric(\vecx, \vecz)$. 
\item The projection map of a point $\vecx$ to the set $\calS$ is $\Pi_\calS(\vecx)=
\displaystyle\arg\inf_{\vecz \in S} \metric(\vecx, \vecz)
$\footnote{Notice that by convexity of set $\calS$ and for any norm, the aforementioned map is well-defined and corresponds to a single point in $\calS$.
}. 
\item The diameter of a set $\calS\subseteq \calM$ is $\diam(\calS)=\sup\{\metric(\vecx,\vecy) : \vecx,\vecy\in \calS\}$
\item The \emph{closed $\eps$-parallel body of $\calS$} to be $\bar{\class{B}}(\calS,\eps):=\bigcup_{\vecx\in \calS}\{\vecz\in\calM:\metric(\vecx,\vecz)\leq\eps\}$, namely the union of closed $\eps$-balls centered in each element of a set $\calS$.\footnote{For brevity,  we will also write directly $\bar{\class{B}}(\vecx,\eps)$ instead of $\bar{\class{B}}(\{\vecx\},\eps)$ in the singleton set case.}
\item Let $\calX$ and $\calY$ be two non-empty sets. We define their \emph{Hausdorff distance} 
$\hausdorff(\calX,\calY) =\\
\max\left\{\,\displaystyle\sup_{\vecx \in \calX} \dist(\vecx,\calY),\, \displaystyle\sup_{\vecy \in \calY} \dist(\calX,\vecy) \,\right\} \!= \inf\{\,{\eps\ge 0}: \calX\subseteq\bar{\class{B}}(\calY,\eps)\land\, \calY\subseteq\bar{\class{B}}(\calX,\eps) \,\} \!$
\item Finally, we define the \emph{inner  
closed $\eps$-parallel body} of $S$, 
$\bar{\class{B}}(\calS,-\eps)=\{{\vecx\in \calS}:\bar{\class{B}}(\vecx,\eps)\subseteq \calS\}$. 
The elements of $\bar{\class{B}}(\calS,-\eps)$ can be viewed as the points ``deep inside of $\calS$'', while $\bar{\class{B}}(\calS,\eps)$ as the points that are ``almost inside of $\calS$''.
\end{enumerate}
\end{definition}
For convex bodies, the following properties can be easily shown:
    {\small
\begin{equation*}
\bar{\class{B}}(\bar{\class{B}}(\calS,\eps),-\eps)=\calS, 
\bar{\class{B}}(\bar{\class{B}}(\calS,-\eps),\eps)\subseteq\calS \ \& \
 \bar{\class{B}}(\bar{\class{B}}(\calS,-\eps_1),-\eps_2)=\bar{\class{B}}(\calS,-\eps_1-\eps_2), \bar{\class{B}}(\bar{\class{B}}(\calS,\eps_1),\eps_2)=\bar{\class{B}}(\calS,\eps_1+\eps_2)
\end{equation*}}
 \begin{fact} If $\calA, \calB $ are bounded, convex  sets and have non-empty interior, then $\hausdorff(\calA, \calB)~=~
 \hausdorff(
 \partial 
 \calA, 
 \partial
 \calB)$
 where $\partial 
 \calA, 
 \partial
 \calB$ are the boundaries of $\calA, \calB$ respectively.
 \end{fact}
 \newcommand{\boundellipse}[3]{(#1) ellipse (#2 and #3)}

\begin{minipage}{0.3\textwidth}
\hspace{2em}
\begin{tikzpicture}[bullet/.style={circle,fill,inner sep=1.5pt,node contents={}},scale=0.25]
\draw (0,0) node[bullet,label=below right:$\footnotesize{\small\Pi_S(\vecx)}$,alias=PC]  
    -- (15:3) coordinate[midway] (aux)
            node[bullet,label=above right:$\footnotesize \vecx$] ;
\draw (PC)          to[out=105,in=0] ++ 
          (-0.75,1) to[out=180,in=105] ++ 
          (-1.5,-3) node[above right]{$\footnotesize S$} 
                    to[out=-75,in=180] ++ 
            (1.25,-1) to[out=0,in=-75] cycle;
\draw (1,1) 
node[above]{\footnotesize $\dist(\vecz, S) = \displaystyle\inf_{\vecz \in S} \metric(\vecx,\vecz)$};
\end{tikzpicture}
\end{minipage}
\hspace{-2em}
 \begin{minipage}{0.3\textwidth}
\begin{tikzpicture}[scale=0.9]
    \draw[color=black] (0,0) ellipse (2cm and 1cm);
    \draw[color=black,dashed] (0,0) ellipse (2.25cm and 1.25cm);
    \filldraw[color=blue,opacity=0.2] (0,0) ellipse (2cm and 1cm);
    \filldraw[dashed,color=red,opacity=0.2] (0,0) ellipse (1.75cm and 0.75cm);
    \draw[black,thick] (-2,0) ellipse (0.25cm and 0.25cm);
    \filldraw[color=white,opacity=0.7] (-2,0) ellipse (0.25cm and 0.25cm);
    \draw (-2,0)--(-2.25,0);
    \draw[black,thick] (0,1) ellipse (0.25cm and 0.25cm);
    \filldraw[color=white,opacity=0.7] (0,1) ellipse (0.25cm and 0.25cm);
    \draw (0,1)--(0,1.25);
    \draw[black,thick] (0,-0.75) ellipse (0.25cm and 0.25cm);
    \filldraw[color=white,opacity=0.7] (0,-0.75) ellipse (0.25cm and 0.25cm);
    \draw (0,-1)--(0,-0.75);
    \foreach \x in {(-2,0),(-2.25,0),(0,1),(0,1.25),(0,-1),(0,-0.75)}{
    \fill \x circle[radius=0.5pt];}  
    \draw (1.75,0.25) node {\footnotesize $\calS$};
    \draw (-0.1,0.9) node[above right]{\footnotesize$\eps$};
    \draw (-.1,-2) node[above right]{\footnotesize$\bar{\class{B}}({\calS},\eps)$};
    \draw (0,-0.5) node[above right]{\footnotesize$\bar{\class{B}}({\calS},-\eps)$};
    \draw (-2.25,-0.1) node[above right]{\footnotesize$\eps$};
    \draw (-.1,-1) node[above right]{\footnotesize$\eps$};
\end{tikzpicture}
 \end{minipage}
 \hspace{1em}
\begin{minipage}{0.27\textwidth}
\begin{tikzpicture}[bullet/.style={circle,fill,inner sep=1.5pt,node contents={}},scale=0.3]
\draw [black,fill=cyan!50,opacity=0.2]\boundellipse{4,4}{7}{3};
\draw [black,fill=red!20,opacity=0.4] \boundellipse{1,4}{4}{2};
\draw [black,fill=blue!30,opacity=0.6]\boundellipse{4,4}{3}{1};
\draw (5,4) -- (7,4) ;
\draw (-3,4) -- (1,4);
\draw (4,2) node[above right]{\footnotesize$\calX$};
\draw (0,2) node[above left]{\footnotesize$\calY$};
\foreach \x in {(5,4),(7,4),(-3,4),(1,4)}{
    \fill \x circle[radius=2pt];
}
\draw (-1,4)
node[above]{\scriptsize$\displaystyle\sup_{\vecy\in \calY}\displaystyle\inf_{\vecx\in \calX} \metric(\vecx,\vecy)$};
\draw (7,4)
node[above]{\scriptsize$\displaystyle\sup_{\vecx\in \calX}\displaystyle\inf_{\vecy\in \calY} \metric(\vecx,\vecy)$};
\draw (6,1.5) node[above right]{\footnotesize$\bar{\class{B}}(\calX,\eps)$};
\end{tikzpicture}

\end{minipage}
\bigskip

\noindent Given these definitions of the  distance metrics we can now define semi-continuity (see Definition \ref{def:semi-continuity}) as follows. 

\begin{definition} 
\begin{enumerate} 
\item A correspondence $F:\calX\rightrightarrows\calY$ is called \emph{(Hausdorff)  upper semi-continuous} (H-u.s.c.) at a point $\vecx_o\in\calX$ if and only if for every $\eps>0$ there is a neighborhood $\calU$ of $\vecx_o$ such that  $F(\vecx)\subseteq~\bar{\class{B}}(F(\vecx_o),\eps)$ for all $\vecx\in\calU(\vecx_o)$.
\item A correspondence $F:\calX\rightrightarrows\calY$ is called \emph{(Hausdorff)  lower semi-continuous} (H-l.s.c.) at a point $\vecx_o\in\calX$ if and only if for every $\eps>0$ there is a neighborhood $\calU$ of $\vecx_o$ such that  $F(\vecx_o)\subseteq~\bar{\class{B}}(F(\vecx),\eps)$ for all $\vecx\in\calU(\vecx_o)$.
\end{enumerate}
\end{definition}

Next, we define the set-valued analogue of Lipschitz continuity.
\begin{definition}
A correspondence  $F : \calX \rightrightarrows \calY$ is called (globally) $L$-\emph{Lipschitz continuous with respect Hausdorff metric} or simply \emph{Hausdorff Lipschitz continuous} if there exists a real constant $L \ge 0$ such that, for all $\vecx_1$ and $\vecx_2$ in $\calX$, $\displaystyle \hausdorff(F(\vecx_{1}),F(\vecx_{2}))\leq L \metric(\vecx_{1},\vecx_{2})$
\end{definition}

\begin{remark}
Indeed if $F(\vecx)$ is compact for every $\vecx$, and globally Hausdorff Lipschitz, then both upper and lower semi-continuity trivially hold. To see that,
for an arbitrary $\eps>0$ and any $\vecx_o\in \calX$, then for the neighborhood $\calU(\vecx_o)=\{\vecx\in\calX:\metric(\vecx,\vecx_o)<\eps/L\}$, it holds that $\hausdorff(F(\vecx),F(\vecx_o))\le \eps$. Thus by definition of Hausdorff metric we get that $F(\vecx)\subseteq~\bar{\class{B}}(F(\vecx_o),\eps)$ and 
$F(\vecx_o)\subseteq~\bar{\class{B}}(F(\vecx),\eps)$.
\end{remark}

Throughout the paper, we will assume the metric space $(\R^d,\ell_2)$ and our correspondences would be restricted without loss of generality in the \textsc{HyperCube}, the compact box $[0,1]^d$. Additionally, we will focus on set-valued functions $F(\vecx)$ which are $L-$Hausdorff Lipschitz and whose output for every $\vecx$ is a closed and compact convex subset of $[0,1]^d$ {\em that contains a ball of radius $\eta$} for some fixed $\eta>0$. 
We will call these correspondences/set-valued maps $(\eta,\sqrt{d},L)$-\emph{well-conditioned}. 

\subsection{Representing Set-valued Maps and {\sc Kakutani}'s Computational Complexity}
\label{sec:RepresentationSetValued}
To transcend the polytope-based sketched formulation of the problem of finding Kakutani fixpoints in \cite{Papadimitriou94} (which is, to our knowledge, the only extant computation formulation of this problem), we will describe the outputs of $F$ via an oracle. As a warm-up, we start with the definition of computing Kakutani's fixpoints if we had in our disposal a perfect precision projection oracle for a $L-$Hausdorff Lipschitz correspondence. Later we will relax this requirement by using an oracle that gives an approximate projection, or even a weak separation oracle with a small margin of error.
\medskip
\newcommand{\SKakutani}{\textsc{Kakutani with Projections}}
\label{def:Kakutani}
{\small\begin{nproblem}[\SKakutani]
  \textsc{Input:} A projection circuit or Turing Machine $\calC_{\projectionto{F}{\cdot}}$ that computes the projection of a point to an $L-$Hausdorff Lipschitz set-valued map 
  $F : [0, 1]^d \rightrightarrows [0, 1]^d$ and an accuracy parameter $\alpha$.
  \smallskip
  
  \noindent \textsc{Output:} One of the following:
\begin{enumerate}
    \item[0.] \textit{(Violation of $L$-almost Lipschitzness)}\\ Three vectors $\vecx, \vecy, \vecz\in [0, 1]^d$ and a constant $\epsilon>0$ such that $\vecz\in F(\vecx)$ and  \\
    $\|\calC_{\projectionto{F(\vecx)}{\vecz}}-\calC_{\projectionto{F(\vecy)}{\vecz}}\|\ge L\|\vecx-\vecy\|+\eps$.
    \item[1.] vectors $\vecx, \vecz \in [0, 1]^d$ such that $\norm{\vecx - \vecz} \le \alpha$ and $\vecz \in F(\vecx)\Leftrightarrow d(\vecx,F(\vecx))\leq \alpha$. 
  \end{enumerate}
\end{nproblem}}
\begin{lemma} \label{lem:SKakutani:completeness}
  The computational problem $\SKakutani$ is $\PPAD$-complete.
\end{lemma}


We do not present here a proof of the above lemma since the proof follows from the proof of a more robust version, namely Theorem~\ref{thm:Kakutani:member}, that we present next. This more robust version is closer to the equilibrium existence applications of Kakutani and it will help us show the inclusion to PPAD of concave games and Walrasian equilibrium. Indeed, the existence of such strong separation/membership/projection oracles\footnote{The definition of the computation problem of Kakutani with Projections actually requests the inherently weaker promise that every point of the space admits a unique projection on $F(\vecx)$. Of course, the convexity of $F(\vecx)$ is sufficient for the uniqueness of the nearest point. Interestingly, however \cite{johnson1987nonconvex} showed construction of non-convex sets which can admit the unique nearest point property too.} is burdensome for arbitrary convex sets. 
For instance, not all (and even natural examples of) convex sets have polynomial-time oracles, which makes the task of optimizing over them impossible, see \cite[Ch. 4, pg. 67]{vishnoi2021algorithms} and \cite{de2002approximation}. And finally, being polynomial in the bit complexity of the description for a convex set is also often a serious problem, since for certain convex programs, the bit complexity of all close-to-optimal solutions is exponential to their description, \cite[Ch. 4, pg. 67]{vishnoi2021algorithms}. Hence, before discussing the our computational version of Kakutani's fixed point theorem we need to discuss some prerequisites from convex optimization.

\subsubsection{Convex Optimization Prerequisites} \label{sec:optimizationPrerequisites}

In order to describe a computational version of the problem that would apply in generic convex sets, 
we first need to define a simple model of computation. Following the approach of \cite{grotschel1981ellipsoid,padberg1981russian,karp1982linear}, 
for an arbitrary $\vecx\in \calX \subset [0, 1]^{d}$,  the convex and compact set $F(\vecx)$ is represented by a {\em weak separation oracle} which, for any given point, decides whether that point is inside the set or provides an {\em almost separating hyperplane}. We describe syntactically the aforementioned oracles
via a circuit $\calC_{F(\vecx)}$:
{\small  \begin{oracle}[\ProbWSEP]
  \label{d:WSEP}
    \textsc{Input:} A vector $\vecz \in \Q^{d}\cap[0, 1]^{d}$ and a rational number $\delta>0$     \smallskip
    
    \noindent\textsc{Output:} A pair of $(\veca,b)\in\Q^d\times\Q$ such that
\begin{enumerate}
\item[-] the threshold $b\in[0,1]\in\Q$ denotes the participation of $\vecz$ in $F(\vecx)$. More precisely, if $\vecz \in \bar{\class{B}}(F(\vecx),\delta)$ then 
           $b > 1/2$, otherwise $\calC_F$ outputs $b \leq 1/2$,
\item[-] the vector $\veca\in\Q^d$, with $\|\veca\|_\infty=1$ is meaningful only when $b \leq 1/2$ in which case it defines an almost separating hyperplane $\calH(\veca,\vecz):=\{\vecy\in[0,1]^d: \langle \veca, \vecy-\vecz \rangle =0 \}$ 
between the vector 
           $\vecz$ and the set $F(\vecx)$ such that 
           $ \langle \veca, \vecy-\vecz \rangle\leq \delta$ for
           every $\vecy \in \bar{\class{B}}(F(\vecx),-\delta)$.
    \end{enumerate}
  \end{oracle}}
 \bigskip 
 In other words, for a set-valued map $F(\vecx)$, a Weak Separation Oracle (WSO) is a circuit which received as input: $i)$ the point in question “$\vecz$” \& $ii)$ the accuracy of the separation oracle and outputs either an \emph{almost-membership} or \emph{a guarantee of an almost-separation}.
  \begin{remark}An essential requirement that we impose on these oracles is that the returned separating hyperplane should possess a polynomial bit complexity with respect to the relevant parameters. The absence of this  constraint allows designing a malicious oracle that consistently returns separating hyperplanes with exponential bit complexity, rendering algorithms such as the Ellipsoid method ineffective, despite any design optimizations. This requirement is satisfied by the implementation of the oracle utilizing a linear arithmetic circuit of polynomial size.
  \end{remark}
  
\begin{definition}[Strong Separation Oracles]
In the case that $\delta=0$, we call the separation oracle \emph{strong}, denoted by $\mathrm{SO}_\calX$. For instance, for the special restricted cases of convex polytopes or spheres, such ``exact''-precision oracles are available. However, as we explained earlier, the existence of such computationally efficient strong separation oracles could be presumptuous assumption for arbitrary convex sets.
\end{definition}

\paragraph{A polynomial version of constrained convex optimization.} Towards proving \PPAD-membership for computing Kakutani's fixpoints, Generalized Nash Equilibria in concave games and Walrasian Equilibria in markets, in Section~\ref{sec:concave}\ \&~\ref{sec:Walras}, it is worth recalling 
what the syntactic definition for an algorithm to solve efficiently a general convex program. 
We defer this detailed discussion to Appendix~\ref{app:SyntacticRepresentationSets}. 
We can define the computational version of a weak convex program as
{\small  \begin{nproblem}[\ProbWCCO]
  \label{d:WCCO}
    \textsc{Input:} A zeroth and first order oracle for the convex function $f:\R^d\to \R$, 
    a rational number $\delta>0$ and a weak separation oracle $\mathrm{WSO}_{\calX}$ for a non-empty closed convex set $\calX\subseteq\textsc{HyperCube}$.
    \smallskip
    \noindent \textsc{Output:} One of the following:
    \begin{enumerate}[topsep=0pt,left=2ex,itemsep=0pt,parsep=0.0ex]
            \item[0.] \textit{(Violation of non-emptiness)} \\ A failure symbol $\bot$ with a polynomial-sized witness that certifies that $\parallelbody{\calX}{-\delta}=\emptyset$.
    \item[1.] \textit{(Approximate Minimization)} \\
    A vector $\vecz\in \Q^d\cap\parallelbody{\calX}{\delta}$, such that $f(\vecz)\leq \displaystyle\min_{\vecy\in \parallelbody{\calX}{-\delta}}f(\vecy)+\delta$.
    \end{enumerate}
  \end{nproblem}}
For simplicity, we will assume that function is $L$-Lipshitz continuous or equivalently that all of its subgradients are bounded by some constant $L$. Additionally, if the separating oracle for the feasible set $\calX$ is strong then the optimization problem can be formed as:
{\small  \begin{nproblem}[\ProbSCCO]
  \label{d:SCCO}
    \textsc{Input:} A zeroth and first order oracle for the convex function $f:\R^d\to \R$, 
    a rational number $\delta>0$ and a strong separation oracle $\mathrm{SO}_{\calX}$ for a non-empty closed convex set $\calX\subseteq\textsc{HyperCube}$.
    \smallskip
    \noindent \textsc{Output:} One of the following:
    \begin{enumerate}[topsep=0pt,left=2ex,itemsep=0pt,parsep=0.0ex]
            \item[0.] \textit{(Violation of non-emptiness)} \\ A failure symbol $\bot$ with a polynomial-sized witness that certifies that $\parallelbody{\calX}{-\delta}=\emptyset.$
 \item[1.] \textit{(Approximate Minimization)}\\
    A vector $\vecz\in \Q^d\cap\calX$, such that $f(\vecz)\leq \displaystyle\min_{\vecy\in \calX}f(\vecy)+\delta$.
    \end{enumerate}
  \end{nproblem}}
\begin{remark}
It is useful to reiterate the distinction between the guarantees provided by strong and weak separation oracles. When a strong separation oracle for a closed convex set is accessible, the ellipsoid method can yield a point in $\calX$, which is an $\epsilon$-approximate minimizer of the objective function $f$ in $\calK$. In contrast, in the weak separation oracle case, the solution belongs to $\parallelbody{\calX}{\eps}$  and the guarantee is about the minimizer of $f$ at $\parallelbody{\calX}{-\eps}$.
\end{remark}

For the case of linear programming, the seminal work of \cite[Ch. 2, pg.56]{grotschel2012geometric} provide
a detailed analysis of ellipsoid method. 
Inspired by the \citeauthor{shor1977cut}'s subgradient cuts methodology,
we will provide  an all-inclusive proof for both weak and strong separation oracles 
for the generic case of constained convex programming. 
For concision, we defer the proof at the supplement (See Appendix~\ref{app:ellipsoid}).
\begin{theorem}\label{thm:approx-minimization}
There exists an oracle-polynomial time algorithm, denoted by $\widehat{{\mathcal{A}rg\min}}[f,\mathrm{WSO}_{\calX},\delta]$, that solves \ProbWCCO.
Additionally, if the separating oracle for the feasible set $\calX$ is strong $\mathrm{SO}_{\calX}$ then the corresponding output guarantee can be strengthen 
solving the so-called problem of \textsc{Strong Constrained Convex Optimization}.
\end{theorem}
\noindent As a corollary, there is an oracle-polynomial algorithm for the following problem of approximately minimizing $\ell_2^2$:
\begin{corollary}\label{thm:approx-project}
There exists an oracle-polynomial time algorithm, based on the central-cut ellipsoid method and  denoted by $\widehat{\Pi},\widetilde{\Pi}$, that solves the following projection point to set problem:
\end{corollary}
  {
  \begin{table}[H]
\begin{tabular}{cc}
\begin{minipage}{0.45\textwidth}{\small
  \begin{nproblem}[\ProbWAP]\small
  \label{d:WAP}
    \textsc{Input:} 
    A rational number $\epsilon>0$ and a weak separation oracle $\mathrm{WSO}_{\calX}$ for a non-empty closed convex set $\calX\subseteq\textsc{HyperCube}$ and a vector $\vecx$ that belongs to $\Q^d\cap\calX$.
    \smallskip
    
    \noindent \textsc{Output:} One of the following:
    \begin{enumerate}[topsep=0pt,left=2ex,itemsep=0pt,parsep=0.0ex]
            \item[0.] \textit{(Violation of non-emptiness)} \\ A failure symbol $\bot$ followed by a polynomial-sized witness that certifies that $ \parallelbody{\calX}{-\eps}=\emptyset$.
    \item[1.] \textit{(Approximate Projection)} \\
    A vector $\vecz\in \Q^d\cap\parallelbody{\calX}{\eps}$, such that :\\$\|\vecz-\vecx\|_2^2\leq \displaystyle\min_{\vecy\in \parallelbody{\calX}{-\eps}}\|\vecx-\vecy\|_2^2+\eps$.
    \end{enumerate}
  \end{nproblem}}
  \end{minipage}&
  \begin{minipage}{0.45\textwidth}{\small
  \begin{nproblem}[\ProbSAP]\small
  \label{d:SAP}
    \textsc{Input:} 
    A rational number $\epsilon>0$ and a strong separation oracle $\mathrm{SO}_{\calX}$ for a non-empty closed convex set $\calX\subseteq\textsc{HyperCube}$ and a vector $\vecx$ that belongs to $\Q^d\cap\calX$.
    \smallskip
    
    \noindent \textsc{Output:} One of the following:
    \begin{enumerate}[topsep=0pt,left=2ex,itemsep=0pt,parsep=0.0ex]
            \item[0.] \textit{(Violation of non-emptiness)} \\ A failure symbol $\bot$ followed by a polynomial-sized witness that certifies that $\parallelbody{\calX}{-\eps}=\emptyset$.
    \item[1.] \textit{(Approximate Projection)} \\
    A vector $\vecz\in \Q^d\cap\calX$, such that:\\ $\|\vecz-\vecx\|_2^2\leq \displaystyle\min_{\vecy\in \calX}\|\vecx-\vecy\|_2^2+\eps$.
    \end{enumerate}
  \end{nproblem}}
  \end{minipage}
\end{tabular}
  \end{table}}
\noindent Notice that by definition of \ProbWCCO, it is necessary to provide always an oracle for the subgradients of the objective function.
It is worth mentioning that especially for the case of $\ell_2^2$ we can derive a syntactic representation of an exact first-order oracle:
Indeed, let $f(\vecz)=\|\vecz-\vecx\|_2^2/2$ be the squared distance $\ell_2^2$ from the input vector $\vecx\in[0,1]^d$. Then it holds that 
 $\{\vecy\in [0,1]^d: f(\vecy)\leq f(\vecz)\}\subseteq \{ \vecy\in [0,1]^d: \vecw^\top \vecy\leq  \vecw^\top \vecz, \text{ for }\vecw=\nabla f(\vecz)/\|\nabla f(\vecz)\|_{\infty} \}$. Thus, the halfspace $H=\{\vecy\in [0,1]^d: \left(\tfrac{\nabla f(\vecz)}{\|\nabla f(\vecz)\|_{\infty}}\right)^\top (\vecy-\vecz)\leq 0  \}$ separates exactly the level sets of our objective function. It is important to notice that for rational inputs both the function $\ell_2^2$ and its $\ell_{\infty}$-normalized gradient remain rational. Finally, since  $\algoprojectionto{\calX}{\vecx}{}$ queries actually gradients of $\ell_2^2$ only  for the iterative candidates of the central-cut ellipsoid method \,\textendash\ the centroids of the corresponding ellipsoids\,\textendash\, both zeroth \& first order oracles  are by construction rational and hence polynomially exactly computable.
\paragraph{The Disparity of Solution Guarantees in Weak Oracle Model.}
From the statement of the theorem, an obvious disparity arises as an unavoidable curse of weak separation oracle; while the output of the algorithm belongs to $\parallelbody{\calX}{\eps}$, the performance guarantee refers on the deeper set
$\parallelbody{\calX}{-\eps}$. The following theorem aims to resolve this issue for the squared distance $\ell_2^2$: 
\begin{theorem}\label{thm:approximation-projection} Let $F: [0,1]^d\rightrightarrows[0,1]^d$ be an  $(\eta,\sqrt{d},L)$ well-conditioned  correspondence, and two vectors $\vecx,\vecy\in[0,1]^d$. There exists a constant $\hat{c}_{d,\eta}\geq 1$, such that $\|\algoprojectionto{F(\vecx)}{\vecy}{\parameters}-
\projectionto{\parallelbody{F(\vecx)}{\eps}}{\vecy}
\|_2\leq \hat{c}_{d,\eta}\cdot \eps$
\end{theorem}
Notice that in order to to bridge this disparity  we will assume that  the set-valued maps are $(\eta,\sqrt{d},L)$-well conditioned, which means $\forall\vecx\in[0,1]^d\ \exists\veca_0\in F(\vecx):\parallelbody{\veca}{\eta}\subseteq F(\vecx)$. Fortunately, in any of the aforementioned weak-version algorithm, 
(\textsc{Opt.}/ \textsc{Proj.}),  inner radius $\eta$ is polynomially refutable by ellipsoid method. The proof and discussion of why this assumption is tight can be found in Appendix~\ref{app:disparity}.


Notice that the above theorem close the disparity between the minimizers of $\ell_2^2$ at $\parallelbody{\calX}{\epsilon}$ and $\parallelbody{\calX}{-\epsilon}$. For \Kakutani\  problem, we will see that the aforementioned theorem is sufficient.
A by-product of the machinery, that we will develop for the inclusion in \PPAD, would be the generalization of the above theorem for general strongly convex functions. (See Lemma~\ref{lem:max-dilation-lipschitz}).

\subsubsection{Kakutani's Computational Complexity}
\label{sec:Kakutani:PPAD}

Having built the necessary background, we can finally define the computational problem of finding an approximate Kakutani fixpoint using either weak or strong separation oracles\footnote{
Once again it is useful to underline the dissimilarity between the guarantees provided by strong and weak separation oracles. In the case of a strong separation oracle, the output $\epsilon$-approximate Kakutani fixed point is a member of the set $F(\vecx)$, whereas in the case of a weak separation oracle, it is relaxed to reside within the $\epsilon$-neighborhood of $F(\vecx)$. It is also worth mentioning that, in the case of a strong separation oracle and when the function $F(\cdot)$ has explicit polynomial bounded bit representation, such as in the case of a polytope, the non-emptiness refutation can be strengthened to $F(\vecx)\neq\emptyset$.
}
as a total problem in $\FNP$.

  {
  \begin{table}[H]
\begin{tabular}{cc}
\begin{minipage}{0.45\textwidth}
{\small\begin{nproblem}[\ProbWKakutani]\small
  \textsc{Input:} A circuit $\calC_F$ that represents weak separation oracle for an  $(\eta,\sqrt{d},L)$ well-conditioned  correspondence
  $F : [0, 1]^d \rightrightarrows [0, 1]^d$  and an accuracy parameter $\alpha$.
  \smallskip

  \noindent \textsc{Output:} One of the following:
\begin{enumerate}[topsep=0pt,left=2ex,itemsep=0pt,parsep=0.0ex]
    \item[0a.] \textit{(Violation of  $\eta$-non emptiness)}: A vector $\vecx\in[0,1]^d$ such that $vol(F(\vecx))\le vol(\parallelbody{0}{\eta})$.
    \item[0b.] \textit{(Violation of $L$-almost algorithmic Lipschitzness)}\\
    Four vectors $\vecp, \vecq, \vecz,\vecw \in [0, 1]^d$ and a constant $\epsilon>0$ such that $\vecw=\algoprojectionto{F(\vecq)}{\vecq}{\parameters}$ and $\vecz=\algoprojectionto{F(\vecp)}{\vecw}{\parameters}$ but\\
    $\|\vecz-\vecw\|> L\|\vecp-\vecq\|+\hat{\calL}_{d,\eta}\cdot\eps$
    \footnote{  $\hat{\calL}_{d,\eta}=3(1+\hat{c}_{d,\eta})$, where $\hat{c}_{d,\eta}$ constant of Theorem \ref{thm:approximation-projection}.   }.
\item[1.] vectors $\vecx, \vecz \in [0, 1]^d$ such that 
$\norm{\vecx - \vecz} \le \alpha$ and $\vecz \in F(\vecx)\Leftrightarrow d(\vecx,F(\vecx))\leq \alpha$. 
  \end{enumerate}
\end{nproblem}}
  \end{minipage}&
  \begin{minipage}{0.45\textwidth}
 {\small
\begin{nproblem}[\ProbSKakutani]\small
  \textsc{Input:} A circuit $\calC_F$ that represents strong separation oracle for an  $L$-Hausdorff Lipschitz  correspondence
  $F : [0, 1]^d \rightrightarrows [0, 1]^d$  and an accuracy parameter $\alpha$.
  \smallskip

  \noindent \textsc{Output:} One of the following:
\begin{enumerate}[topsep=0pt,left=2ex,itemsep=0pt,parsep=0.0ex]
    \item[0a.] \textit{(Violation of non emptiness)}: A vector $\vecx\in[0,1]^d$ such that $\parallelbody{F(\vecx)}{-\epsilon}=\emptyset$.
    \item[0b.] \textit{(Violation of $L$-Hausdorff Lipschitzness)}\\
    Four vectors $\vecp, \vecq, \vecz,\vecw \in [0, 1]^d$ and a constant $\epsilon>0$ such that $\vecw=\strongalgoprojectionto{F(\vecq)}{\vecq}{\parameters}$ and $\vecz=\strongalgoprojectionto{F(\vecp)}{\vecw}{\parameters}$ but\\
    $\|\vecz-\vecw\|> L\|\vecp-\vecq\|+3\cdot\eps$.
\item[1.] vectors $\vecx, \vecz \in [0, 1]^d$ such that 
$\norm{\vecx - \vecz} \le \alpha$ and $\vecz \in F(\vecx)\Leftrightarrow d(\vecx,F(\vecx))\leq \alpha$.
  \end{enumerate}
\end{nproblem}}
  \end{minipage}
\end{tabular}
  \end{table}}

Notice that for the definition of \Kakutani\ with $\mathrm{SO}_F$ and $\mathrm{WSO}_F$, we request a more relaxed version for the Lipschitzness of the corresponding algorithmic operators $\strongalgoprojectionto{F(\vecx)}{\cdot}{\epsilon},\algoprojectionto{F(\vecx)}{\cdot}{\epsilon}$.
Below, we prove that the relaxed algorithmic Lipschitzness parameters are reasonable for an $(\eta,\sqrt{d},L)$ well-conditioned correspondence (See the proof in Appendix~\ref{app:kakutani}):
\begin{lemma}\label{lem:approximation-lipschitzness} Let $F: [0,1]^d\rightrightarrows[0,1]^d$ be an  $(\eta,\sqrt{d},L)$ well-conditioned  correspondence, and two vectors $\vecp,\vecq\in[0,1]^d$. Then, it holds
\[\|\algoprojectionto{F(\vecp)}{\algoprojectionto{F(\vecq)}{\vecq}{\parameters}}{\parameters}-{\algoprojectionto{F(\vecp)}{\vecq}{\parameters}}\|\leq L\|\vecp-\vecq\|+3(1+\hat{c}_{d,\eta})\epsilon\]
where $\hat{c}_{d,\eta}$ is the constant of Theorem~\ref{thm:approximation-projection}
\end{lemma}

In Appendix~\ref{sec:Kakutani:inclusion} we use the above tools that we built from convex optimization and a classical formulation of a high-dimensional instance of Sperner's lemma from \citep{chen2021complexity} to show Theorem \ref{thm:Kakutani:member}. 
\begin{theorem}
\label{thm:Kakutani:member}
  The computational problems of $\Kakutani$ with $\mathrm{WSO}_F,\mathrm{SO}_F,\mathrm{ProjO}_F$ are in $\PPAD$.
\end{theorem}
Intuitively for a simplicization of $[0,1]^d$ we will assign a color to each point of  $\vecx$  as follows: 
\begin{itemize}[topsep=0pt,leftmargin=3ex]
\setlength{\itemsep}{0pt}
\setlength{\parskip}{.2ex}
\item If $\vecx$ is fixed point then we are done; otherwise we compute $G(\vecx)=\Pi_{F(\vecx)}(\vecx)-\vecx$, where
$\Pi_{F(\vecx)}(\vecx)$ is the projection of $\vecx$ in $F(\vecx)$. Then if $G(\vecx)$ belongs to the positive orthant then it is colored $0$, otherwise it is colored with the first lexicographically coordinate which is non-positive. We tie-break at the boundaries to ensure that coloring is a \emph{Sperner's} one. Sperner's lemma implies the existence of a panchromatic simplex $S$.
\item It follows from our coloring that by proving the Lipschitzness of $\vecp(\vecx)=\Pi_{F(\vecx)}(\vecx)$,
  we can show that when the simplicization is \emph{fine} enough,  
  there exists a point in a panchromatic simplex yields an $\epsilon$-\Kakutani\ fixed point. The main difficulty even under perfect projections is to show that $\hat{\vecp}(\vecx)=\Pi_{F(\vecx)}(\vecw)$ for some arbitrary $\vecw$. The proof of this  property passes through a geometrical argument via Apolloneous triangle theorem.
\end{itemize}

\noindent The computational version of Kakutani's fixed point theorem is evidently $\PPAD$-hard, shown by a reduction from $\textsc{Brouwer}$. Full proof is presented in Appendix \ref{sec:Kakutani:hardness}.
\begin{lemma} \label{lem:Kakutani:hardness}
  The computational problems of $\Kakutani$ with $\mathrm{WSO}_F,\mathrm{SO}_F,\mathrm{ProjO}_F$ are $\PPAD$-hard.
\end{lemma}
\subsection{Robust Berge's Maximum Theorem} \label{sec:Berge}

In equilibrium existence problems Kakutani is usually applied together with the seminal Claude Berge's Maximum Theorem. It then comes with no surprise that our inclusion to PPAD proof of those problems needs to use not only the inclusion of Kakutani to PPAD but also the some version of the Maximum Theorem. Unfortunately, the continuity guarantees that the Maximum Theorem provides are not enough to apply our computational version of Kakutani that requires Lipschitzness. For this reason we need a robust version of Berge's Maximum Theorem. Interestingly, delving into the inclusion proof (See Appendix~\ref{app:kakutani}), it is noteworthy that lemmas~\ref{lem:our-berge-for-projections}, \ref{lem:our-berge-for-strong-projections}
 correspond actually to a computational robustification of Berge's Maximum Theorem with the functions:
{\small\[ f^*(a)=\|\projectionto{F(a)}{a}-a\|=\max_{b\in F(a)}\{-\|b-a\|_2^2/2\} \ \ \& \ \ g^*(a)=\projectionto{F(a)}{a}=\arg\max_{b\in F(a)}\{-\|b-a\|_2^2/2\},\]} if we apply the theorem for functions $f:[0,1]^d\times[0,1]^d\to\R$ and $g:[0,1]^d\rightrightarrows[0,1]^d$ defined by $f(a,b):=-\|b-a\|_2^2/2$ and $g(a):= F(a)$. 
\begin{theorem}[Berge's Maximum Theorem \citep{bonsall1963c}]\label{thm:berge-generalized} Let $A\subseteq\R^n$ and $B\subseteq\R^m$. Let $f:A\times B\to \R$ be a continuous function and $g:A\rightrightarrows B$ continuous as well as non-empty, compact-valued correspondence. Define $f^*:A\to\R$ and $g^*:A\rightrightarrows B$ by $f^*(a)=\max_{b\in g(a)} f(a,b)$ and $g^*(a)=\arg\max_{b\in g(a)} f(a,b)$. Then $f^*$ is continuous and $g^*$ upper semi-continuous as well as non-empty, compact-valued correspondence.
\end{theorem}
Notice that if $g:=F$ then it is trivially continuous by the Hausdorff's Lipschitzness assumption of our correspondences. However, even in the perfect computation regime where $\eta=\eps=0$, Berge's theorem does not transfer trivially the Lipschitz condition from $g(a)=F(a)$ to $g^*(a)=\projectionto{F(a)}{a}$. In the  section~\ref{sec:concave}, we will leverage a quantified version of Berge's theorem for general strongly convex functions. Namely, we can show the following result:

\begin{theorem}[Robust Berge's Maximum Theorem]\label{thm:our-berge-for-optimization}
Let $A\subseteq\R^n$ and $B\subseteq\R^m$. Consider a continuous function $f:A\times B\to \R$ that is $\mu-$strongly concave $\forall a\in A$, $L-$Lipschitz in $A\times B$ and a $L'$-Haussodorf Lipschitz, non-empty, convex-set, compact-valued correspondence $g:A\rightrightarrows B$. By defining $f^*(a)=\max_{b\in g(a)} f(a,b)$ and $g^*(a)=\arg\max_{b\in g(a)} f(a,b)$, we observe $f^*$ is continuous and $g^*$ is upper semi-continuous and single-valued, i.e., continuous. Furthermore, $f^*$ and $g^*$ are Lipschitz and $\Big(L'+2\sqrt{\tfrac{4}{\mu}} \sqrt{(L+L\cdot L')} \Big)$- (1/2) H"{o}lder continuous respectively (for sufficiently small differences).
\end{theorem}

\section{Computational Complexity of Concave Games} \label{sec:concave}
In this section, we explore the computational complexity of finding an 
approximate equilibrium in concave games defined in the celebrated work of
\cite{rosen1965existence}. We first give the definitions of the corresponding
computational problems that we explore and then we characterize their 
computational complexity. Before presenting the computational definition we
first formally define the notion of an $n$-person concave game and the 
equilibrium concept that we are interested in. 

\begin{definition}[Concave Games] \label{def:concaveGames}
    An $n$-person concave games is a tuple $(\calI, \calR, \calU, S)$ described
  as follows:
  \begin{enumerate}
    \item[-] $\calI$ is a partition of the set of coordinates $[d]$. We use 
    $I_i$ to denote the $i$th set of this partition. $I_i$ corresponds to the 
    indices of variables that are controlled by player $i$. We have that 
    $k_i = \abs{I_i}$. Unless we mention otherwise, we have that 
    $I_i = \b{\sum_{j = 1}^{i - 1} k_j, \sum_{j = 1}^{i} k_j}$.
    \item[-] $\calR$ is a family of \textit{strategy domains} $R_i$, one for 
    each player $i$. $R_i$ is a convex subset of $\R^{k_i}$ and in this paper
    we assume without loss of generality that 
    $R_i = [-1, 1]^{k_i}$. We also use $k = \sum_{i = 1}^n k_i$.
    \item[-] $\calU$ is a set of continuous \textit{utility functions} 
    $u_i : \R^k \to [0, 1]$ one for each agent $i \in [n]$ that is convex
    with respect to the subvector $\vecx_i \in R_i$.
    \item[-] $S$ is a convex compact set, subset of $[-1, 1]^k$, that imposes 
    one common \textit{convex constraint} of the form $\vecx \in S$. 
    Additionally, without loss of generality we 
    assume that $\vec{0} \in S$.
  \end{enumerate}
  When $\calR$ and $\calI$ are fixed we may skip $\calR$ and $\calI$ from the
  notation of an $n$-person concave game and use just $(\calU, S)$.
\end{definition}


\begin{definition}[Equilibrium in Concave Games] \label{def:concave:equilibrium}
    Let $(\calU, S)$ be an $n$-person concave game. A vector 
  $\vecx^{\star} \in S$ is an \textit{equilibrium} of $(\calU, S)$ if for every
  $i \in [n]$ and every $\vecy_i \in [-1, 1]^{k_i}$ such that 
  $(\vecy_i, \vecx^{\star}_{-i}) \in S$ it holds that
  \[ u_i(\vecx^{\star}) \ge u_i(\vecy_i, \vecx^{\star}_{-i}). \]
\end{definition}


\noindent As Rosen showed in his celebrated work \cite{rosen1965existence} an 
equilibrium in any $n$-person concave game is guaranteed to exist.

\begin{theorem}[\cite{rosen1965existence}] \label{thm:concave:existence}
    For any $n$-player concave game $(\calU, S)$ an equilibrium $\vecx^{\star}$ 
  of $(\calU, S)$ always exists.
\end{theorem}

\noindent Since we will be  working with computational versions of the 
problem of finding an equilibrium point in concave games, we also need a notion
of approximate equilibrium to account for the bounded accuracy of computational 
methods. For this definition we also need the following notion of approximate equilibrium..



\begin{definition}[$(\eps, \eta)$-Approximate Equilibrium] \label{def:concave:approximateEquilibrium}
    Let $(\calU, S)$ be an $n$-person concave game, 
  $S_{\eta} = \bar{\class{B}}(S, \eta)$, and $S_{-\eta} = \bar{\class{B}}(S, -\eta)$. A vector 
  $\vecx^{\star} \in S_{\eta}$ is an $(\eps, \eta)$-\textit{equilibrium} of 
  $(\calU, S)$ if for every $i \in [n]$ and every $\vecy_i \in [-1, 1]^{k_i}$
  such that $(\vecy_i, \vecx^{\star}_{-i}) \in S_{-\eta}$ it holds that
  \[ u_i(\vecx^{\star}) \ge u_i(\vecy_i, \vecx^{\star}_{-i}) - \eps. \]
  When $\eta = 0$ refer to $\vecx^{\star}$ as an $\eps$-approximate equilibrium.
\end{definition}

\begin{remark}[Discussion about $\eta$]
Definition \ref{def:concave:approximateEquilibrium} defers from the standard notion of approximate equilibrium due to the presence of the approximation parameter $\eta$. We include this parameter $\eta$ in the definition in order to capture instances where we only have a weak separation oracle on the constraint set $S$. In certain such instances the bit complexity of any $\eps$-approximate equilibrium is infinite. This means that the computational problem of finding an $\eps$-approximate equilibrium is not well defined when we represent numbers using the binary representation. Therefore, the presence of $\eta$ is inevitable when we only have weak separation oracle access to $S$.

We note also that our reductions have running time that scale as $\poly \log(1/\eta)$ which means that we can assume that $\eta$ is exponentially small and hence the difference between $S_{\eta}$, $S$, and $S_{-\eta}$ is significant only in very pathological instances.

Finally, as we will see later in this section, when we have access to $S$ via a strong separation oracle then we can show results for the classical $\eps$-approximate equilibrium problems.
\end{remark}

\subsection{Computational Problems of Finding Equilibrium in Concave Games} \label{sec:concave:definitions}

  In order to define the computational version of finding an $(\epsilon,\eta)$-approximate equilibrium in an
$n$-person concave games we first need to formally define the computational 
representation of the ingredients of an $n$-person concave game. In particular, 
the representation of the utility functions $(u_i)_{i = 1}^n$ and the 
representation of the set $S$. For these representations we use again the definition
of \textit{linear arithmetic circuits} and we refer to Appendix E of \cite{fearnley2021complexity}, where it is shown that linear arithmetic circuits approximate well-behaved functions which are enough for our results. For the sake of completeness, below we recall their formal definition:

\begin{definition}[Linear Arithmetic Circuits] \label{def:linearArithmeticCircuit}
    A \textit{linear arithmetic circuit} $\calC$ is a circuit represented as a
  directed acyclic graph with nodes labelled either as input nodes, or as output
  nodes or as gate nodes with one the following possible gates 
  $\{+, -, \min, \max, \times \zeta\}$, where the $\times \zeta$ gate refers to 
  the multiplication by a constant. We use $\size(\calC)$ to refer to the number
  of nodes of $\calC$.
\end{definition}

\paragraph{Representation of utility functions.} A utility function $u_i$ is 
represented using one of the following ways: (1) via a general circuit $\calC_{u_i}$
that takes as input a point $\vecx \in [-1, 1]^k$ and computes in the output the
value $u_i(\vecx)$ as  well as the subgradient of $u_i$ at the point $\vecx$, or
(2) as a sum of monomials in the variables $x_1, \ldots, x_k$, (3) a linear arithmetic circuit. For the first 
representation we follow the paradigm of \cite{daskalakis2021complexity} and 
assume that the correctness of the computation of the subgradient is given as a promise. For the second representation, computation of value and subgradient of an arbitrary utility function is easy due to its succinct description.
More interestingly, by representing the utility functions with linear arithmetic circuits, there are methods to compute one vector that belongs to their subgradients using automatic differentiation techniques without the need for a circuit that computes them, see, e.g., \cite{barton2018computationally}. Also, restricting our attention to linear arithmetic circuits we do not lose representation power since it has been shown in \cite{fearnley2021complexity} that linear arithmetic circuits can efficiently approximate any polynomially computable, Lipschitz function over a bounded domain.

The concavity of $u_i$ with respect to the variables that 
are controlled by agent $i$ is equivalent with the following condition for every
$\vecx_i, \vecy_i \in [-1, 1]^{k_i}$, $\vecx_{-i} \in [-1, 1]^{k - k_i}$ and
every $\lambda \in [0, 1]$:
$u_i(\lambda \vecx_i + (1 - \lambda) \vecy_i, \vecx_{-i}) \ge \lambda u_i(\vecx_i, \vecx_{-i}) + (1 - \lambda) u_i(\vecy_i, \vecx_{-i})$.
If this concavity property does not hold then we can provide a witness for the
refutation of this property by providing the vectors 
$\vecx_i, \vecy_i$, $\vecx_{-i}$ and the number $\lambda$ for which this property fails. Another way to represent the utility functions is as a sum of monomials, which makes the problem much more structured but as we see even when the utility functions have constant degree the problem remains $\PPAD$-hard.

\paragraph{Representation of a convex set.} The convex and compact set $S$ that  imposes the constraints for concave game is represented via some linear arithmetic circuit $\calC_S$ which will represent either a strong or weak separation oracle for the corresponding convex set, similarly with the previous section (See the discussion of Section~\ref{sec:RepresentationSetValued}). 

\paragraph{Definitions of computational problems.} Now that we have discussed the 
representation of functions and convex sets we are ready to present three
different definitions of computational problems associated with concave games: 

\noindent (1) The most general definition that captures all the continuous games with
  continuous and concave utility functions and arbitrary convex constraints.
  
{\small  \begin{nproblem}[\textsc{ConcaveGames}]
    \textsc{Input:} We receive as input all the following:
    \begin{enumerate}[]
      \item[-] $n$ circuits $(\calC_{u_i})_{i = 1}^n$ representing the utility functions $(u_i)_{i = 1}^n$,
      \item[-] an arithmetic circuit $\calC_S$ representing a weak/strong separation oracle for a constrained well-bounded convex set $S$, i.e.  $\exists \veca_0\in \R^d\ : \parallelbody{\veca_0}{r}\subseteq S
\subseteq\parallelbody{0}{R}\subseteq [-1,1]^k$
      \item[-] a Lipschitzness parameter $L$, and
      \item[-] accuracy parameters $\eps, \eta$.
    \end{enumerate}
    \smallskip

    \noindent \textsc{Output:} We output as solution one of the following.
    \begin{enumerate}
      \item[0a.] \textit{(Violation of Lipschitz Continuity)} \\
                 A certification that there exist at least two vectors $\vecx, \vecy \in [-1, 1]^k$ and an index $i \in [n]$
                 such that 
                 $\abs{u_i(\vecx) - u_i(\vecy)} > L \cdot \norm{\vecx - \vecy}$.
      \item[0b.] \textit{(Violation of Concavity)} \\ 
                 An index $i \in [n]$, three vectors
                 $\vecx_i, \vecy_i \in [-1, 1]^{k_i}$, 
                 $\vecx_{-i} \in [-1, 1]^{k - k_i}$ and a number 
                 $\lambda \in [0, 1]$ such that 
                 $u_i(\lambda \vecx_i + (1 - \lambda) \vecy_i, \vecx_{-i}) < \lambda u_i(\vecx_i, \vecx_{-i}) + (1 - \lambda) u_i(\vecy_i, \vecx_{-i})$.
    \item[0c.] \textit{(Violation of  almost non emptiness)} \\ A certification that $vol(S)\le vol(\parallelbody{0}{r})$.
    
      \item[1.]  An $(\eps, \eta)$-approximate equilibrium as per Definition
                 \ref{def:concave:approximateEquilibrium}.
    \end{enumerate}
  \end{nproblem}}
\begin{remark}
Similarly with the previous section, (\emph{Violation of almost-non emptiness}) includes multiple different malicious cases: \emph{(i)} the emptiness of the constraint set, \emph{(ii)} the inconsistency of the separation oracle or \emph{(iii)} the well-bounded conditions for the size of the constraint set. Again, following the convention of the previous section, we can always  interpret  \emph{a-fortiori} our set $S$ to be any convex set, which is circumvented by the separating hyperplanes provided by our (strong/weak) oracle $\text{SO}_S/\text{WSO}_S$. Finally,
(\emph{Violation of Concavity}) is meaningful as output whenever the form of utilities is explicitly given otherwise concavity holds as a promise.
\end{remark}
\medskip

\noindent (2) The version where a strong separation oracle is provided and stronger notion of approximate equilibrium can be computed with $\eta=0$.

{\small
  \begin{nproblem}[\textsc{ConcaveGames with SO}]
  \textsc{Input:} Same with \textsc{ConcaveGames} with $\calC_S$ representing a strong separation oracle.
  
    \noindent \textsc{Output:} 
    \begin{enumerate}
      \item[0a.]- 0c. Same with \textsc{ConcaveGames}.
      \item[1.] An $\eps$-approximate equilibrium as per Definition \ref{def:concave:approximateEquilibrium}.
    \end{enumerate}
  \end{nproblem}
}
\medskip

\noindent (3) The version where the utility functions are restricted to be 
  strongly-concave and given as explicit polynomials.
{\small  
  \begin{nproblem}[\textsc{StronglyConcaveGames}]
    \textsc{Input:} We receive as input all the following:
    \begin{enumerate}
      \item[-] $n$ polynomials $(u_i)_{i = 1}^n$ given as a sum of monomials,
      \item[-] a Lipschitzness parameter $L$, a strong concavity parameter  
               $\mu$, and
      \item[-] accuracy parameters $\eps$.
    \end{enumerate}
    \smallskip

    \noindent \textsc{Output:} We output as solution one of the following.
    \begin{enumerate}
      \item[0a.] \textit{(Violation of Lipschitz Continuity)} same as Output 0a. of 
                 $\textsc{ConcaveGames}$.
      \item[0b.] \textit{(Violation of Strong Concavity)} \\
                 An index $i \in [n]$, three vectors
                 $\vecx_i, \vecy_i \in [-1, 1]^{k_i}$, 
                 $\vecx_{-i} \in [-1, 1]^{k - k_i}$ and a number 
                 $\lambda \in [0, 1]$ such that 
                 \begin{align*}
                   u_i(\lambda \cdot \vecx_i + (1 - \lambda) \cdot \vecy_i, \vecx_{-i}) & < \lambda \cdot u_i(\vecx_i, \vecx_{-i}) + (1 - \lambda) \cdot u_i(\vecy_i, \vecx_{-i}) \\
                   & ~~~~~~~~~ + \frac{\lambda (1 - \lambda)}{2} \cdot \mu \cdot \norm{(\vecx_i, \vecx_{-i}) - (\vecy_i, \vecx_{-i})}_2^2.
                 \end{align*}
      \item[1.]  An $\eps$-approximate equilibrium as per Definition
                 \ref{def:concave:approximateEquilibrium} with $S = [-1, 1]^k$.
    \end{enumerate}
  \end{nproblem}
}
  

  
\medskip
\noindent Now that we have defined the computational problems that we are going to
explore in this section we state a simple reduction among these problems.

\begin{lemma} \label{lem:concave:problemRelations}
    The following relations holds for the above computational problems:
  \begin{enumerate}
    \item[-] $\textsc{StronglyConcaveGames} \le_{\FP} \textsc{ConcaveGames}$.
  \end{enumerate}
\end{lemma}

\begin{proof}
    Easily follows from the definitions of the problems given the proofs of Appendix E of \cite{fearnley2021complexity} that shows that linear arithmetic circuits can efficiently approximate any polynomially computable function so they can approximate polynomials as well.
\end{proof}

\noindent In Appendix~\ref{app:concaveGames} we provide firstly the proof of inclusion to \PPAD for \textsc{ConcaveGames} with either weak or strong separation oracles and then the hardness result for the easier problem of \textsc{StronglyConcaveGames}, demonstrating the following important result:

\begin{theorem} \label{main:thm:concave:complete}\label{thm:concave:complete}
     The computational problems $\textsc{ConcaveGames}$ and $\textsc{StronglyConcaveGames}$ equipped with $({\sc Weak/Strong})$ separation oracle are $\PPAD$-complete.
\end{theorem}
Some commentary on the outcome: The difficulty of concave games can be directly inferred from normal form games. However, the most intriguing part of this proof lies in the provision of a stricter example through strongly concave games, where every opponent's strategy profile has a unique best response. Concerning the inclusion, a significant portion of our analysis is dedicated to proving that dilation is a Lipschitz Haussdorf operation for the players' strategy constant sets. This finding represents a vital aspect of our results and contributes to the overall understanding of the mechanics in play.

\section{Computational Complexity of Walrasian Equilibria} \label{sec:Walras}
In this section, we delve into the intricacies of identifying a Walrasian equilibrium in markets. Building on the pioneering work of Léon Walras, we examine the computational complexity of determining approximate equilibrium for prices and quantities in markets with concave utility functions. We begin by outlining the specific computational problems that we aim to address and proceed to classify their computational complexity. To set the stage, we first provide a formal definition of an $n$-agent market under Walrasian model and the corresponding equilibrium concept.

Briefly speaking, in Walrasian model we examine a pure exchange economy \textendash a market system without the presence of production. The economy consists of $n$ individuals (agents) and $d$ goods. Each individual is endowed a specific bundle of goods. Before the end of the world, there will be a chance for trade at specific prices. Our objective is to determine if there exist prices portfolio $\vecp\in\mathbb{R}_{+}^{d}$ at which everyone can trade their desired quantities and demand will equal supply while maximizing the preference of each player.

\begin{definition}[The Walrasian Model] \label{def:WalrasianModel}
    An exchange economy of $n$ agents and $d$ commodities under Walrasian Model is a tuple $( \calE, \calB, \calU)$ described as follows:
  \begin{enumerate}
    \item[-] $\calE$ is the collection of each player's endowment.
    We use $\vece_i\in \mathbb{R}_+^{d}$ to denote the set of the goods that are endowed initially to the player $i$.\footnote{By definition, we assume that $(\vece_i)_k>0\ \ \forall k\in[d]$.}
    \item[-] $\calB(\vecp)$ is a family of \textit{allocation constraints} $\calB_i$, one for 
    each player $i$: $\calB_i(\vecp):=\left\{\small \vecx \in \mathbb{R}_+^{d} \ | \ \vecp\cdot\vecx\leq \vecp\cdot\vece_i\ \right\}$ 
    \item[-] $\calU$ is a set of continuous \textit{utility functions} 
    $u_i : \R^d_+ \to [0, 1]$ one for each agent $i \in [n]$ that is convex
    with respect to the subvector $\vecx_i \in \calB_i(\vecp)$.
  \end{enumerate}
  When $\vecp$ is fixed we may skip $\calB$ from the
  notation of an $(n,d)$-exchange economy and use just $(\calE, \calU)$.
\end{definition}

 In words, the Walrasian equilibrium is a state in which the supply of goods and services in a market is equal to the demand for them, and all prices are such that there is no incentive for buyers or sellers to change their behavior. The model is used to study the relationships between different markets and the overall economy, and to analyze the effects of changes in economic policy or external conditions on the economy.
\begin{definition}[Competitive (Walrasian) Equilibrium in an Exchange Economy] \label{def:market:equilibrium}
    Let $(\calE, \calU)$ be an $(n,d)$-exchange economy. A vector pair
  $\vecp^{\star}, \vecx^{\star} \in (\R_+^d,\calB(\vecp^\star))$ is a \textit{Competitive (Walrasian) equilibrium} of $(\calE, \calU)$ economy if 
  \begin{enumerate}
      \item Agents are maximizing their utilities: For all $i\in[n]$, $\vecx_i^\star\in\arg\max_{\vecx} u_i(\vecx)$ 
      where $\vecx_i^\star\in \calB_i(\vecp^\star)$
      \item Markets are clear: For all $m\in [d]$, $\sum_{i\in[n]}\vecx_i=\sum_{i\in [n]}\vece_i$
  \end{enumerate}
\end{definition}

\begin{remark}
 The budget constraint is slightly different than in standard price theory. Recall that the familiar budget constraint is $\vecp \cdot \vecx \leq W$, where $W$ is the consumer’s initial wealth. Here the consumer’s “wealth” is $\vecp\cdot \vece_i$, the amount she could get if she sold her entire endowment. 
\end{remark}
%

\noindent As \cite{cassel1924theory} showed initially, and later in their celebrated work of \citeauthor{arrowdebreu1954} a Walrasian 
equilibrium in any $(n,d)$-exchange economy is guaranteed to exist.

\begin{theorem}[\cite{arrowdebreu1954}] \label{thm:market:existence}
     For any $(n,d)$-exchange economy  $(\calU, \calE)$  with concave increasing continuous utilities and strictly positive endowments, a Walrasian equilibrium $(\vecp^{\star},\vecx^{\star})$ always exists.
\end{theorem}


\noindent 
Similarly with concave games, in the Walrasian model of exchange economy, we also need to consider the concept of approximate equilibrium to account for the limitations of computational methods when searching for an equilibrium point in the market. 



{
\begin{definition}[$\eps$-Approximate Equilibrium] \label{def:markets:approximateEquilibrium}
    Let $(\calE, \calU)$ be an $(n,d)$-exchange economy. A vector pair
  $\vecp^{\star}, \vecx^{\star} \in (\R_+^d,\calB(\vecp^\star))$ is a \textit{Competitive (Walrasian) equilibrium} of $(\calE, \calU)$ economy if 
  \begin{enumerate}
      \item Agents are almost maximizing their utilities: For all $i\in[n]$, $\vecx_i^\star\geq\max_{\vecx\in\calB_i(\vecp^\star)} u_i(\vecx)-\epsilon$ 
      where $\vecx_i^\star\in \calB_i(\vecp^\star)$
      \item Markets are almost-clear: 
      $\sum_{i\in[n]}\vecx_i\in [1-\epsilon,1+\epsilon]\sum_{i\in [n]}\vece_i$
\end{enumerate}
\end{definition}
}

\subsection{Computational Problems of Finding Equilibrium in Walras Model} \label{sec:markets:definitions}

In order to define the computational version of finding an $(\epsilon,\eta)$-approximate equilibrium in an
$n,d$-exchange economy we will follow exactly the same formalism that we introduce in Section~\ref{sec:concave:definitions} for the computational 
representation of the utilities of an $n$-person concave game. Having said this, we are ready to define formally the computation version of Walrasian Equilibrium.

{\small
  \begin{nproblem}[\textsc{Walrasian}]
    \textsc{Input:} The input consists of all of the following:
    \begin{enumerate}[]
      \item[-] $n$ circuits $(\calC_{u_i})_{i = 1}^n$ representing the utility functions $(u_i)_{i = 1}^n$,
      \item[-] a Lipschitzness parameter $L$, and
      \item[-] accuracy parameter $\eps$.
    \end{enumerate}
    \smallskip

    \noindent \textsc{Output:} We output as solution one of the following.
    \begin{enumerate}
      \item[0a.] \textit{(Violation of Lipschitz Continuity)}
                 A certification that there exist at least two vectors $\vecx, \vecy \in \R_+^d$ and an index $i \in [n]$
                 such that 
                 $\abs{u_i(\vecx) - u_i(\vecy)} > L \cdot \norm{\vecx - \vecy}$.
      \item[0b.] \textit{(Violation of Concavity)} \\ 
                 An index $i \in [n]$, three vectors
                 $\vecx_i, \vecy_i \in \R_+^d$, 
                 $\vecx_{-i} \in (\R_+^d)^{n-1}$ and a number 
                 $\lambda \in [0, 1]$ such that 
                 $u_i(\lambda \vecx_i + (1 - \lambda) \vecy_i, \vecx_{-i}) < \lambda u_i(\vecx_i, \vecx_{-i}) + (1 - \lambda) u_i(\vecy_i, \vecx_{-i})$.
    
      \item[1.]  An $\eps$-approximate Walrasian equilibrium as per Definition
                 \ref{def:markets:approximateEquilibrium}.
    \end{enumerate}
  \end{nproblem}}

It is worth mentioning that a series of papers has provided \PPAD-hardness for specific cases of concave utility functions. In Appendix~\ref{sec:markets:inclusion}, we close the gap providing the membership proof for the general concave utilities, proving the following theorem:
\begin{theorem} \label{thm:warlasian:inclusion}
     The computational problem $\textsc{Walrasian}$ is in $\PPAD$.
\end{theorem}
In order to establish the \PPAD membership of the Walras market equilibrium, we construct a reduction to Kakutani, utilizing two "meta-players": the Price agents and the Excessive Cumulative Demand agent.
A crucial requirement for applying Kakutani's theorem is establishing the Lipschitzness of the argmax operator that we construct over the solution space of strategies that "empty" the market. To quantify this Lipschitzness, we draw upon the Hoffman bounds from linear algebra, which proves instrumental in the process.

\section{Conclusions}
We have mapped the complexity of two very general and important fixpoint theorems, namely Kakutani's and Rosen's; for the latter, completeness holds even when the concave functions have a rather simple and explisit polynomial form.  There are of course several problems that remain open, and here are two interesting ones:
\begin{itemize}

\item Rosen defines in \cite{rosen1965existence} a rather opaque special case of concave games that he calls {\em diagonally dominant}, and proves that such games have a unique equilibrium through an interesting algorithm. What is the complexity of this special case?  We suspect that it may lie within the class CLS.

\item In the proof, we define an algorithmic version of Berge's Maximum Theorem. Is this problem PPAD-complete?  More interestingly, 
is the inverse Berge theorem \cite{komiya1997} (given an upper semicontinuous map, find a convex function such that the given map is obtained by applying Berge's Theorem) PPAD-complete?

\end{itemize}

\clearpage

\bibliographystyle{ACM-Reference-Format}
\bibliography{sample-bibliography,ec-23}
\clearpage
\appendix

\renewcommand{\contentsname}{Organization of the appendix}
\addtocontents{toc}{\protect\setcounter{tocdepth}{2}}
\tableofcontents

\clearpage

\section{Convex Programming in \FNP, Syntactic Representation of Feasible Set \& Consistency of Function values and its gradients.}
\label{app:SyntacticRepresentationSets}
\paragraph{Challenges in definition of Convex Optimization in \FNP.} 

Represent convex sets in a computationally convenient way via separating oracles, below we revisit the different approaches to formulate a convex minimization problem. In more details, we have that:
\begin{definition}[Convex Program]
Given a convex set $\calX\subseteq\R^d$ and a convex function $f:\calX\to\R$, a convex program is the following optimization problem: 
\begin{equation}
    \inf_{\vecx\in\calX} f(\vecx)\label{eq:convex-program}
\end{equation}
\end{definition}
We say that a convex program is \emph{unconstrained} when $\calX=\R^d$, i.e., when we are optimizing over all inputs, and we call it \emph{constrained} when the feasible set $\calX$ is a strict subset of $\R^d$. Further, when $f$ is differentiable with continuous derivative, we call the problem \emph{smooth} convex program and \emph{non-smooth} otherwise. Additionally, if $\calK\subseteq\R^d$ is closed and bounded (compact) then we are assured that the infimum value is attained by a point $\vecx\in\calX$.
Unfortunately, we cannot hope to obtain always the exact optimal value in a constrained convex program. For instance, convex programs like $\min_{x\ge 1}\sqrt{2}x$ or $\min_{x\ge 1}(x+\tfrac{2}{x})$, have irrational solutions either optimal value or point and thus they cannot be represented in finite, let alone polynomial, bit representation. However, even we define the approximate \FNP problem as
\begin{equation}
    \text{Given a rational $\epsilon>0$ compute a point $\vecx'\in\calX$ such that $\min_{\vecx\in\calX} f(\vecx)\leq f(\vecx')\leq \min_{\vecx\in\calX}+\epsilon$},
\end{equation}
Similar problems arise for the arbitrary constrained set. For example if $\calX$ could be a singleton or a set whose boundary are vectors with irrational coordinates. This is the reason that seminal work of \cite[Ch. 2]{grotschel2012geometric} introduced the notion of weak separation oracles and the corresponding counterparts, allowing margins in inequalities and around the surface of the constrained set $\calX$.

\paragraph{Syntactic Representation of Feasible Set \& Consistency of Function values and its gradients.} 
In defining the computation problem associated with convex optimization, an extra discussion is necessary 
about the syntactic certification of structural properties like convexity, non-emptiness or compactness in \FNP.
\begin{enumerate}[topsep=0pt,left=2ex,itemsep=0pt,parsep=0.0ex]
    \item[-] \emph{For a convex feasible set $\calX$}:
\begin{enumerate}[topsep=0pt,left=2ex,itemsep=0pt,parsep=0.0ex]

\item
\begin{minipage}{.5\textwidth}
    In principle, the existence of a (exact/almost) separating hyperplane between any point of the domain and a set  is possible only for (perfect/approximately) convex sets. Therefore, without loss of generality we can assume that $\mathcal{X}$ is convex, but only do this for simplicity. The set $\mathcal{X}$ can implicitly be defined as the intersection of all the separating hyperplanes that our circuit returns. If there is a combination of separating hyperplanes that are conflicting then it means that $\mathcal{X}$ is empty and we could accept such an instance as a solution. This way we could avoid the promise that $\mathcal{X}$ is a-priori convex but it makes many of the definitions and the proofs much more complicated with very limited additional merit and for this reason we choose to omit it.
\end{minipage}
 \begin{minipage}{.3\textwidth}

\tikzset{every picture/.style={line width=0.75pt}} 

\begin{tikzpicture}[x=0.75pt,y=0.75pt,yscale=-0.7,xscale=0.7]

\draw  [draw opacity=0][fill={rgb, 255:red, 208; green, 2; blue, 27 }  ,fill opacity=0.34 ] (240,128.27) .. controls (240,108.39) and (267.51,92.27) .. (301.46,92.27) .. controls (335.4,92.27) and (362.91,108.39) .. (362.91,128.27) .. controls (362.91,148.16) and (335.4,164.27) .. (301.46,164.27) .. controls (267.51,164.27) and (240,148.16) .. (240,128.27) -- cycle ;
\draw  [fill={rgb, 255:red, 208; green, 2; blue, 27 }  ,fill opacity=0.45 ] (360,129) .. controls (360,129) and (360,129) .. (360,129) .. controls (360,129) and (360,129) .. (360,129) .. controls (360,144.46) and (334.26,157) .. (302.5,157) .. controls (270.74,157) and (245,144.46) .. (245,129) .. controls (245,113.54) and (270.74,101) .. (302.5,101) -- (302.5,129) -- cycle ;
\draw    (330,64) -- (425,157) ;
\draw    (177,103) -- (215.4,141.79) -- (287.91,213.27) ;
\draw    (268.91,208.27) -- (441.91,143.27) ;
\draw    (179,112) -- (338,67) ;
\draw  [color={rgb, 255:red, 0; green, 0; blue, 0 }  ,draw opacity=1 ][fill={rgb, 255:red, 208; green, 2; blue, 27 }  ,fill opacity=0.17 ][line width=1.5]  (235.5,129) .. controls (235.5,105.94) and (265.5,87.25) .. (302.5,87.25) .. controls (339.5,87.25) and (369.5,105.94) .. (369.5,129) .. controls (369.5,152.06) and (339.5,170.75) .. (302.5,170.75) .. controls (265.5,170.75) and (235.5,152.06) .. (235.5,129) -- cycle ;
\draw   (199,76) -- (408,76) -- (408,194) -- (199,194) -- cycle ;
\draw  [fill={rgb, 255:red, 74; green, 144; blue, 226 }  ,fill opacity=0.43 ] (305.85,79.35) -- (341.43,80.1) -- (405.11,142.86) -- (405.46,154.12) -- (303.85,190.35) -- (269.85,189.35) -- (204.85,125.35) -- (204.85,109.35) -- cycle ;
\draw    (331,231) .. controls (371,201) and (314.66,202.23) .. (354.66,172.23) ;

\draw (345,113.4) node [anchor=north west][inner sep=0.75pt]  [font=\scriptsize]  {$\mathcal{X} ''$};
\draw (410,79.4) node [anchor=north west][inner sep=0.75pt]  [font=\small]  {$\mathcal{B} ox$};
\draw (289,125.4) node [anchor=north west][inner sep=0.75pt]  [font=\scriptsize]  {$\mathcal{X} '$};
\draw (299.46,172.67) node [anchor=north west][inner sep=0.75pt]  [font=\scriptsize]  {$\mathcal{X} '''$};
\draw (194,233) node [anchor=north west][inner sep=0.75pt]  [font=\scriptsize] [align=left] {Implicit Representation of any $\calX',\calX'',\calX'''$\\ via 4 weak separation hyperplanes \& $\mathcal{B}ox$.};
\end{tikzpicture}
\end{minipage}
\\
    \item Actually, any set which lies in the interior of the intersection a given collection of hyperplanes provided by a $\mathrm{WSO}_{\calX}$ is information-theoretically equivalent. Thus, we can similarly assume that w.l.o.g $\calX$ is closed as well. 
    \item Additionally, as we mentioned in the previous section, we follow the premise that $\calX$ is circumvented in a box, for simplicity let' say $[0,1]^d=$\textsc{HyperCube}. Algorithmically for the compactness of the set $\calX$, we can always clip the coordinates of a candidate solution inside the box. In general, for the case of box constraints $\mathcal{B}ox=\{\vecx\in \R^d|\ell_i\leq x_i \leq u_i\quad\forall i\in[d]\}$, we can always apply as a pre-processing step the corresponding ``exact'' separation hyperplanes. 
    \item Finally, for a collection of separating hyperplanes, the emptiness of their intersection can be tested accurately via ellipsoid in polynomial time. 
\end{enumerate}
\item[-] \emph{For the objective function $f$:}
\begin{enumerate}
    \item[] For simplicity, we can assume that we have access either to some linear arithmetic circuit, a Turing machine or some 
black-box oracles. In general, however, given only queries in a zeroth \& first-order oracle or examining the description of a circuit or a Turing machine, it is computationally difficult to examine the consistency of the function values and their gradients. Fortunately, for both membership and hardness results in Kakutani's fixpoints case, our proofs leverage instances where the aforementioned consistency can be syntactically guaranteed. (See 
\cite{fearnley2021complexity}) For the case of generalized Nash Equilibria in concave games, we defer the corresponding discussion for the utilities of the players in Section~\ref{sec:concave}.
\end{enumerate}
\end{enumerate}

\clearpage
\section{The Disparity of Solution Guarantees in Weak Oracle Model.}
\label{app:disparity}
From the statement of the Theorem~\ref{thm:approx-minimization}, an obvious disparity arises as an unavoidable curse of weak separation oracle; while the output of the algorithm belongs to $\parallelbody{\calX}{\eps}$, the performance guarantee refers on the deeper set
$\parallelbody{\calX}{-\eps}$. The following lemmas aims to resolve this issue for the squared distance $\ell_2^2$: 
\begin{lemma}\label{lem:projection-normal-cone}
     Let $\vecx,\vecy$ be vectors in $\R^d$ and $\calX$ be an arbitrary compact convex set. Then it holds that
     \[\vecy=\projectionto{\calX}{\vecx}\Leftrightarrow\vecx \in \vecy+\calN_\calX(\vecy)\]
where $\calN_\calX(\vecy)=\{\vecz\in\R^d: \langle\vecz,\veck-\vecy\rangle\leq 0\ \ \forall \veck\in\calX\}$ corresponds to the normal cone of the convex set $\calX$ at the point $\vecy$.
\end{lemma}
\begin{proof}
Indeed, let's define the projection of a point over an arbitrary set $\calX$ as an unconstrained optimization problem of a lower semi-continuous extended convex function. More precisely, it holds that \[\vecy=\projectionto{\calX}{\vecx}=
\arg\min_{\vecz\in\calX}\{
\|\vecx-\vecz\|_2^2/2\}
=\arg\min_{\vecz\in\R^d}\{
\|\vecx-\vecz\|_2^2/2+\mathbf{1}_\calX(\vecz)
\}\] From the generalized Fermat's theorem\footnote{
Generalized Fermat's theorem Statement (See Theorem 8.15  \cite{rockafellar2009variational}): If a function $f:\mathbb{R}^n\to\bar{\mathbb{R}}$ is nondifferential, convex, proper and it has a local minimum at $\bar{x}$, then $0 \in \partial f(\bar{x})$.
} and the fact that every stationary point for a convex function corresponds to a global minimizer,
it holds that:
\[\vecy=\projectionto{\calX}{\vecx}\Leftrightarrow
 0\in\partial\{\|\vecx-\vecz\|_2^2/2+\mathbf{1}_\calX(\vecz)\}(\vecy)\Leftrightarrow
0\in\vecy-\vecx +\partial\mathbf{1}_\calX(\vecy)\Leftrightarrow \vecx \in \vecy+\calN_\calX(\vecy)\]
where we used the fact from subdifferential calculus that $\calN_\calX(\vecy)=\partial\mathbf{1}_\calX(\vecy)$. 
\end{proof}
\begin{lemma}\label{lem:from-out-to-border}
Consider a convex compact set $\calX\subseteq [0,1]^d$ and an arbitrary point $\vecx$ in $\R^d$. Then it holds that $\vec{x},\projectionto{\calX}{\vecx},
\projectionto{\parallelbody{\calX}{\eps}}{\vecx}
$ are co-linear for every $\eps>0$.
\end{lemma}
\begin{proof}
Let $\vecy_\eps$ and $\vecy$ be now the corresponding projections on ${\parallelbody{\calX}{\epsilon}}$ and ${\calX}$.
Notice that if $\vecx\in{\parallelbody{\calX}{\epsilon}}$ then the statement of the lemma holds trivially. Thus, for the rest of the proof we will assume that $\vecx\neq\vecy_\epsilon\neq\vecy$.

By Lemma \ref{lem:projection-normal-cone}, we know that $\vecx-\vecy=\vecw\in \calN_{\calX}(\vecy)$. Since $\vecx\neq\vecy_\eps$ or equivalently $\vecx\not\in {\parallelbody{\calX}{\epsilon}}$, it holds that $\|\vecx-\vecy\|=\|\vecw\|>\eps$.
Additionally, from the definition of normal cone, it holds that if $\vecw\in \calN_{\calX}(\vecy)$ then $\eps\frac{\vecw}{\|\vecw\|_2}\in  \calN_{\calX}(\vecy) $. Furthermore, it is easy to see that $\vecy'=\vecy+\eps\frac{\vecw}{\|\vecw\|_2}\in\parallelbody{\calX}{\epsilon} $.

Notice now that $\parallelbody{\calX}{\eps}$ can be written as the Minkowski sum of two convex sets, namely $ \parallelbody{\calX}{\eps}= \calX + \parallelbody{\mathbf{0}}{\eps} $. 
Using the subdifferential calculus for a Minkowski sum of two convex sets, we get:
\begin{align*}
        \calN_{\parallelbody{\calX}{\eps} }(\vecy')
&=\calN_{\calX+\parallelbody{\mathbf{0}}{\eps} }(\vecy')
=\partial \mathbf{1}_{\calX+\parallelbody{\mathbf{0}}{\eps} }(\vecy')
=\partial(  \mathbf{1}_{\calX} \#\mathbf{1}_{\parallelbody{\mathbf{0}}{\eps} })(\vecy'=\vecy+\eps\frac{\vecw}{\|\vecw\|})\\
&=\partial\mathbf{1}_{\calX}(\vecy) \cap\partial\mathbf{1}_{\parallelbody{\mathbf{0}}{\eps} }(\eps\frac{\vecw}{\|\vecw\|})=\calN_{\calX}(\vecy)\cap\calN_{\parallelbody{\mathbf{0}}{\eps} }(\eps\frac{\vecw}{\|\vecw\|})\\&\overset{(\star)}{=}\calN_{\calX}(\vecy)\cap\{t\eps\frac{\vecw}{\|\vecw\|}: t\geq 0\}\overset{(\star\star)}{=}\{t\vecw: t\geq 0\}
\end{align*}
\emph{Explanation:$
\begin{cases}\text{$(\star)$ holds by the fact that $\calN_{\parallelbody{\mathbf{0}}{1}}(\vecz)=\begin{cases}\emptyset& \text{if } \|\vecz\|<1\\\R_{\ge 0} \vecz& \text{if } \|\vecz\|=1\\\end{cases}$}\\\text{ 
$(\star\star)$ holds by the fact that $\vecw\in\calN_\calX(\vecy)$. }
\end{cases}$
}\\

Finally, after some calculations we have that 
\begin{align*}
    \vecx=\vecy+\vecw=\vecy+\eps\frac{\vecw}{\|\vecw\|}+\underbrace{(\|w\|-\epsilon)}_{t'\ge 0}\vecw=\vecy'+t'\vecw \in \vecy'+\calN_{\parallelbody{\calX}{\eps} }(\vecy')  
\end{align*}
Again, by Lemma \ref{lem:projection-normal-cone}, the above expression yields $\vecy'=\projectionto{\parallelbody{\calX}{\eps}}{\vecx}=\vecy_\eps$ and consequently $\vecx-\vecy_\eps=t'\vecw$, which conclude the proof since $\vecx-\vecy_\eps\parallel	 \vecx-\vecy$.
\end{proof}

\begin{lemma}\label{lem:from-border-to-in}
For any well-bounded convex set $\calX$, i.e.  $\exists \veca_0\in \R^d\ : \parallelbody{\veca_0}{r}\subseteq \calX
\subseteq\parallelbody{0}{R}$, and an arbitrary positive constant $\eps$ such that $\eps\in (0,r)$, it holds that $\|\vecx-\projectionto{\parallelbody{\calX}{-\eps}}{\vecx}\|\le \frac{R}{r}\eps$\ \  for every $\vecx\in \calX$.
\end{lemma}
\begin{proof}
It suffices to prove the statement the extreme case, when  $\vecx\in\partial\calX$. Consider the function $f(\vecz)=\tfrac{\eps}{r}\vecz + (1-\tfrac{\eps}{r})\vecx$. By convexity, it holds that if $\vecz\in\calX$ then $f(\vecz)\in \calX$. If $\veca=f(\veca_0)$ then it is easy to check that $f(\parallelbody{\veca_0}{r})=\parallelbody{\veca}{\eps}\subseteq\calX$. Since, $\parallelbody{\veca}{\eps}\subseteq\calX$, by the definition of inner parallel body we have that $\veca\in \parallelbody{\calX}{-\eps}$. Therefore,
\[\|\vecx-\projectionto{\parallelbody{\calX}{-\eps}}{\vecx}\|\le \|\vecx-\veca\|=\|
\tfrac{\eps}{r}(\veca_0)+ (1-\tfrac{\eps}{r})\vecx-\vecx
\|\le\frac{\|\vecx-\veca_0\|}{r}\eps\leq\frac{R}{r}\epsilon\]
\end{proof}

\begin{figure}[H]
\centering
\begin{subfigure}{.5\textwidth}
  \centering\begin{tikzpicture}[scale=1.7]
\draw[color=black,thick] (0,0) ellipse (0.5cm and 0.5cm);
\fill[color=cyan] (0,0) ellipse (0.5cm and 0.5cm);
\fill[color=lightgray] (0,0) -- (-1/2,-1)-- (1/2,-1);
\fill[color=cyan] (1/2.25,0.5/2.25)--(0,0) --(1/2,-1)-- (1,-1)--cycle;
\fill[color=cyan] (-1/2.25,0.5/2.25)--(0,0) --(-1/2,-1)-- (-1,-1)--cycle;
\draw[red,thick] (0,0) -- (2,1);
\draw[red,thick] (0,0) -- (-2,1);
\fill[color=red,opacity=0.15] (0,0) -- (-2,1)-- (2,1);
\draw[olive,thick]  (0,0.5) -- (0,1);
\draw[olive,thick, dashed] (0,0) -- (0,0.5) -- (0,1);
\draw[black,thick] (1/2.25,0.5/2.25) -- (1,-1);
\draw[black,thick] (-1/2.25,0.5/2.25) -- (-1,-1);
\draw[black,thick] (0,0) -- (1/2,-1);
\draw[black,thick] (0,0) -- (-1/2,-1);
 \draw[black,thick] (0.2,-0.45) node[above right]{$\bar{\class{B}}({\calC},\eps)$};
\draw[white,thick] (0,-0.75) node {$\calC$};

 \draw[black,thick] (0.2,-0.45) node[above right]{$\bar{\class{B}}({\calC},\eps)$};
\draw[red,thick] (0.5,0.75) node[right] {$\calN_\calC$};
\draw[olive,thick] (-0.2,0.75) node[left] {$\calN_{\bar{\class{B}}({\calC},\eps)}$};

\draw[violet,thick] (0,-0.1) node[left] {$\vecy$};
\draw[violet,thick] (0,0.45) node[left] {$\vecy_\eps$};
\draw[violet,thick] (0,1) node[left] {$\vecx$};
\draw[violet,thick] (0,0) node {$\star$};
\draw[violet,thick] (0,0.45) node {$\star$};
\draw[violet,thick] (0,1) node {$\star$};
\end{tikzpicture}
\caption{Illustration of Lemma~\ref{lem:from-out-to-border}}
\end{subfigure}%
\begin{subfigure}{.5\textwidth}
  \centering

  \begin{tikzpicture}[bullet/.style={circle,fill,inner sep=1.5pt,node contents={}},scale=0.9]
 \draw (-2,-2) circle (1cm);
  \draw (-2,0.6) circle (0.5cm);
    \draw (-2,0.6) circle (0.5cm);
 \draw (-2,-0.65) circle (0.5cm);
 \draw (-1,1) node{$S$} ;
  \draw (-3.25,-3)node{$\parallelbody{S}{-\eps}$} ;
  \draw (-2,0.6) node[bullet,label=left:$ $] ;
   \draw (-4,0.6) node {$\projectionto{\vecx}{\parallelbody{S}{-\eps}} $} ;
  \draw (-2,-0.65) node[bullet,label=left:$\veca$] ;
  \draw (-2,-2) node[bullet,label=left:$\veca_0$] ;
  \draw (-2,-2) -- (-2,1.1);
 \draw (-2,1.1) node[bullet,label=$\vecx\equiv(A)$] ;

\filldraw[fill=lightgray,opacity=0.2] (1,0)          to[out=105,in=0] ++ 
          (-1.75*1.2,1*1.2) to[out=180,in=105] ++ 
          (-3.5*1.2,-4*1.2) node[above right]{$ $} 
                    to[out=-75,in=180] ++ 
            (3.25*1.2,-1*1.2) to[out=0,in=-75] cycle;
            
\draw[fill=lightgray,opacity=0.2] (0.5,-0.35)          to[out=105,in=0] ++ 
          (-1.75/1,1/1) to[out=180,in=105] ++ 
          (-3.5/1,-4/1) node[above right]{$ $} 
                    to[out=-75,in=180] ++ 
            (3.25/1,-1/1) to[out=0,in=-75] cycle;            
\end{tikzpicture}
\caption{Illustration of Lemma~\ref{lem:from-border-to-in}}
\end{subfigure}
\end{figure}
\begin{remark}
From the above argumentation, an important question arises about the necessity of the assumption about inner-radius $r>0$ in the well-condition characterization of a convex set, namely the fact that
$\exists \veca_0\in \R^d\ : \parallelbody{\veca_0}{r}\subseteq \calX \subseteq\parallelbody{0}{R}$. First, we stress here that one cannot derive, in oracle polynomial time, a rational number $r>0$ such that the convex body $\calX$ contains a ball of radius $r$, from a weak separation oracle for $\calX$, even if one knows
a radiues $R$ such that $\calX\subseteq\parallelbody{0}{R}$. Indeed, let $\calX=\{\sqrt{2}\}$. Notice that $\calX$ is trivially a closed convex set with $R=1$, and suppose that we have access even to a strong separation oracle which turns out to answer, for any $\vecy\in\Q,\delta>0$:
{\small\begin{itemize}[topsep=0pt,left=2ex,itemsep=0pt,parsep=0.0ex]
\item if $\vecy=\sqrt{2}$, then it asserts that $\vecy\in\calX$.
\item if $\vecy\neq\sqrt{2}$, then it outputs the separation hyperplane $\veca=\sign(\vecy-\sqrt{2})$, i.e., that $\veca^\top\vecx<\veca^\top\vecy$ for all $\vecx\in\calX$.
\end{itemize}}
From the above description, it is clearly impossible form information theoretic perspective to compute in polynomial time of a lower bound of $r$, just by leveraging a $\mathrm{WSO}_{\calX}$. Notice additionally that our accuracy parameter should be always bounded by the maximal inner radius $r$, otherwise $\parallelbody{\calX}{\epsilon}\Big|_{\epsilon>r}=\emptyset$. Finally, there are examples of convex sets (see figure below) where the ratio between $r$ and $R$  can be arbitrarily large and for $\epsilon\approx (1-\delta)r$, Lemma~\ref{lem:from-border-to-in} is tight.
\end{remark}
\begin{figure}[h!]
    \centering

\tikzset{every picture/.style={line width=0.75pt}} 

\begin{tikzpicture}[x=0.75pt,y=0.75pt,yscale=-0.75,xscale=0.75]

\draw [color={rgb, 255:red, 74; green, 144; blue, 226 }  ,draw opacity=1 ][fill={rgb, 255:red, 139; green, 87; blue, 42 }  ,fill opacity=0.23 ]   (143.29,171) -- (604.29,197) -- (143.38,221) ;
\draw  [draw opacity=0][fill={rgb, 255:red, 139; green, 87; blue, 42 }  ,fill opacity=0.23 ] (143.31,221) .. controls (143.23,221) and (143.14,221) .. (143.06,221) .. controls (129.1,221) and (117.79,209.81) .. (117.79,196) .. controls (117.79,182.19) and (129.1,171) .. (143.06,171) .. controls (143.14,171) and (143.21,171) .. (143.29,171) -- (143.06,196) -- cycle ; \draw  [color={rgb, 255:red, 74; green, 144; blue, 226 }  ,draw opacity=1 ] (143.31,221) .. controls (143.23,221) and (143.14,221) .. (143.06,221) .. controls (129.1,221) and (117.79,209.81) .. (117.79,196) .. controls (117.79,182.19) and (129.1,171) .. (143.06,171) .. controls (143.14,171) and (143.21,171) .. (143.29,171) ;  
\draw [color={rgb, 255:red, 208; green, 2; blue, 27 }  ,draw opacity=1 ]   (143.06,196) -- (143,171) ;
\draw [shift={(143,171)}, rotate = 269.86] [color={rgb, 255:red, 208; green, 2; blue, 27 }  ,draw opacity=1 ][fill={rgb, 255:red, 208; green, 2; blue, 27 }  ,fill opacity=1 ][line width=0.75]      (0, 0) circle [x radius= 3.35, y radius= 3.35]   ;
\draw [shift={(143.06,196)}, rotate = 269.86] [color={rgb, 255:red, 208; green, 2; blue, 27 }  ,draw opacity=1 ][fill={rgb, 255:red, 208; green, 2; blue, 27 }  ,fill opacity=1 ][line width=0.75]      (0, 0) circle [x radius= 3.35, y radius= 3.35]   ;
\draw [color={rgb, 255:red, 0; green, 0; blue, 0 }  ,draw opacity=1 ] [dash pattern={on 0.84pt off 2.51pt}]  (604.29,197) -- (143.06,196) ;
\draw [shift={(143.06,196)}, rotate = 180.12] [color={rgb, 255:red, 0; green, 0; blue, 0 }  ,draw opacity=1 ][fill={rgb, 255:red, 0; green, 0; blue, 0 }  ,fill opacity=1 ][line width=0.75]      (0, 0) circle [x radius= 3.35, y radius= 3.35]   ;
\draw [shift={(604.29,197)}, rotate = 180.12] [color={rgb, 255:red, 0; green, 0; blue, 0 }  ,draw opacity=1 ][fill={rgb, 255:red, 0; green, 0; blue, 0 }  ,fill opacity=1 ][line width=0.75]      (0, 0) circle [x radius= 3.35, y radius= 3.35]   ;
\draw [color={rgb, 255:red, 208; green, 2; blue, 27 }  ,draw opacity=1 ]   (143.03,183.5) .. controls (103,167) and (170,165) .. (139,153) ;
\draw [color={rgb, 255:red, 74; green, 144; blue, 226 }  ,draw opacity=1 ]   (218.03,216.5) .. controls (273,240) and (164,218) .. (224,256) ;
\draw [color={rgb, 255:red, 0; green, 0; blue, 0 }  ,draw opacity=1 ]   (143.06,196) ;
\draw [shift={(143.06,196)}, rotate = 0] [color={rgb, 255:red, 0; green, 0; blue, 0 }  ,draw opacity=1 ][fill={rgb, 255:red, 0; green, 0; blue, 0 }  ,fill opacity=1 ][line width=0.75]      (0, 0) circle [x radius= 3.35, y radius= 3.35]   ;
\draw [shift={(143.06,196)}, rotate = 0] [color={rgb, 255:red, 0; green, 0; blue, 0 }  ,draw opacity=1 ][fill={rgb, 255:red, 0; green, 0; blue, 0 }  ,fill opacity=1 ][line width=0.75]      (0, 0) circle [x radius= 3.35, y radius= 3.35]   ;
\draw [color={rgb, 255:red, 0; green, 0; blue, 0 }  ,draw opacity=1 ] [dash pattern={on 0.84pt off 2.51pt}]  (256.03,155.5) .. controls (311,179) and (202,157) .. (262,195) ;

\draw (127.36,141.4) node [anchor=north west][inner sep=0.75pt]    {$r$};
\draw (229.36,246.4) node [anchor=north west][inner sep=0.75pt]    {$R$};
\draw (128.36,196.4) node [anchor=north west][inner sep=0.75pt]  [font=\footnotesize]  {$0$};
\draw (240.36,133.4) node [anchor=north west][inner sep=0.75pt]    {$\|\vecx-\projectionto{\parallelbody{\calX}{-\eps}}{\vecx}\|$};
\draw (615.36,185.4) node [anchor=north west][inner sep=0.75pt]    {$x$};

\end{tikzpicture}

\end{figure}
Thus, only for the case of weak separation oracles, to bridge this disparity  we will assume that  the set-valued maps are $(\eta,\sqrt{d},L)$-well conditioned, i.e., $\forall\vecx\in[0,1]^d\ \exists\veca_0\in F(\vecx):\parallelbody{\veca}{\eta}\subseteq F(\vecx)$. Fortunately, in any of the aforementioned weak-version algorithm, 
(\textsc{Opt.}/ \textsc{Proj.}),  inner radius $\eta$ is polynomially refutable by ellipsoid method. In other words, if an iteration of central-cut ellipsoid method discovers an ellipsoid $\calX \subseteq E_{(k)}$ such that $vol(E)<vol(\parallelbody{0}{\eta})$ then the algorithm outputs immediately a corresponding failure certificate. Consequently, in any \textsc{Weak}-version of a problem, we rephrase the refutation guarantee as follows:
{\small\begin{nproblem}[\textsc{Weak Convex } (\textsc{Feasibility}/\textsc{Projection}/\textsc{Optimization})]
\begin{enumerate}[topsep=0pt,left=2ex,itemsep=0pt,parsep=0.0ex]
            \item[0.] \textit{(Violation of non-emptiness)} \\ A failure symbol $\bot$ followed by a polynomial-sized witness that certifies that\\
            either $\calX=\emptyset$ or $vol(\calX)<vol(\parallelbody{0}{\eta})$.
\end{enumerate}
\end{nproblem}}
\smallskip

An immediate application of Lemma~\ref{lem:from-border-to-in} and \ref{lem:from-out-to-border} yields the following theorem:
\begin{theorem}[Restated Theorem~\ref{thm:approximation-projection}]\label{app:thm:approximation-projection} Let $F: [0,1]^d\rightrightarrows[0,1]^d$ be an  $(\eta,\sqrt{d},L)$ well-conditioned  correspondence, and two vectors $\vecx,\vecy\in[0,1]^d$. There exists a constant $\hat{c}_{d,\eta}\geq 1$, such that $\|\algoprojectionto{F(\vecx)}{\vecy}{\parameters}-
\projectionto{\parallelbody{F(\vecx)}{\eps}}{\vecy}
\|_2\leq \hat{c}_{d,\eta}\cdot \eps$
\end{theorem}
\clearpage
\section{Omitted Proofs of Section~\ref{sec:Kakutani}: Inclusion \& Hardness of \Kakutani\ in \PPAD }
\label{app:kakutani}
We start with the connection between the mathematical and the algorithmic projection operator:
\begin{lemma}[Restated Lemma~\ref{lem:approximation-lipschitzness}]\label{app:lem:approximation-lipschitzness} Let $F: [0,1]^d\rightrightarrows[0,1]^d$ be an  $(\eta,\sqrt{d},L)$ well-conditioned  correspondence, and two vectors $\vecp,\vecq\in[0,1]^d$. Then, it holds
\[\|\algoprojectionto{F(\vecp)}{\algoprojectionto{F(\vecq)}{\vecq}{\parameters}}{\parameters}-{\algoprojectionto{F(\vecp)}{\vecq}{\parameters}}\|\leq L\|\vecp-\vecq\|+3(1+\hat{c}_{d,\eta})\epsilon\]
where $\hat{c}_{d,\eta}$ is the constant of Theorem~\ref{thm:approximation-projection}
\end{lemma}
\begin{proof}
    \begin{align}
        \|\algoprojectionto{F(\vecp)}{\algoprojectionto{F(\vecq)}{\vecq}{\parameters}}{\parameters}-{\algoprojectionto{F(\vecp)}{\vecq}{\parameters}}\|
        \label{align:eq2}&\leq\|\projectionto{\parallelbody{F(\vecp)}{\epsilon}}{\algoprojectionto{F(\vecq)}{\vecq}{\parameters}}-{\projectionto{\parallelbody{F(\vecp)}{\epsilon}}{\vecq}}\|+2\hat{c}_{d,\eta}\epsilon\\
        \label{align:eq3}&\leq\|\projectionto{\parallelbody{F(\vecp)}{\epsilon}}{\projectionto{\parallelbody{F(\vecp)}{\epsilon}}{\vecq}}-{\projectionto{\parallelbody{F(\vecp)}{\epsilon}}{\vecq}}\|+3\hat{c}_{d,\eta}\epsilon\\
        \label{align:eq4}&\leq \|\projectionto{F(\vecp)}{\projectionto{F(\vecq)}{\vecq}}-{\projectionto{F(\vecp)}{\vecq}}\|+3\epsilon+3\hat{c}_{d,\eta}\epsilon\\
        \label{align:eq5}&\leq L\|\vecp-\vecq\|+3(1+\hat{c}_{d,\eta})\epsilon
    \end{align}
where \eqref{align:eq5} uses  $L-$Hausdorff Lipschitzness of $F$, \eqref{align:eq4} holds by $ \|\projectionto{\parallelbody{F(\vecp)}{\epsilon}}{\vecx}-\projectionto{F(\vecp)}{\vecx}\|\leq \epsilon $, \eqref{align:eq3} is derived by $\|\projectionto{F(\vecp)}{\vecx}-\projectionto{F(\vecp)}{\vecy} \|\leq\|\vecx-\vecy\| $ and \eqref{align:eq2} holds by Theorem~\ref{thm:approximation-projection}.
\end{proof}

\subsection{Inclusion of {\sc Kakutani} to $\PPAD$} \label{sec:Kakutani:inclusion}
\begin{theorem}[Restated Theorem~\ref{thm:Kakutani:member}]\label{app:thm:Kakutani:member}
  The computational problems of $\Kakutani$ with $\mathrm{WSO}_F,\mathrm{SO}_F,\mathrm{ProjO}_F$ are in $\PPAD$.
\end{theorem}
The goal of this subsection is to establish the existence of an approximate Kakutani's Fixed point and to perform the closely related task of placing the problem of computing $\eps$-approximate \Kakutani\ fixed points for \emph{well-conditioned} correspondences in \PPAD. The main tools we will use are Sperner's Lemma  and its search problem.

\paragraph{High-dimensional Sperner's Lemma in Hypercube.}  
\begin{wrapfigure}{r}{0.3\textwidth}
\vspace{-2em}
  \begin{center}
    \includegraphics[width=0.28\textwidth]{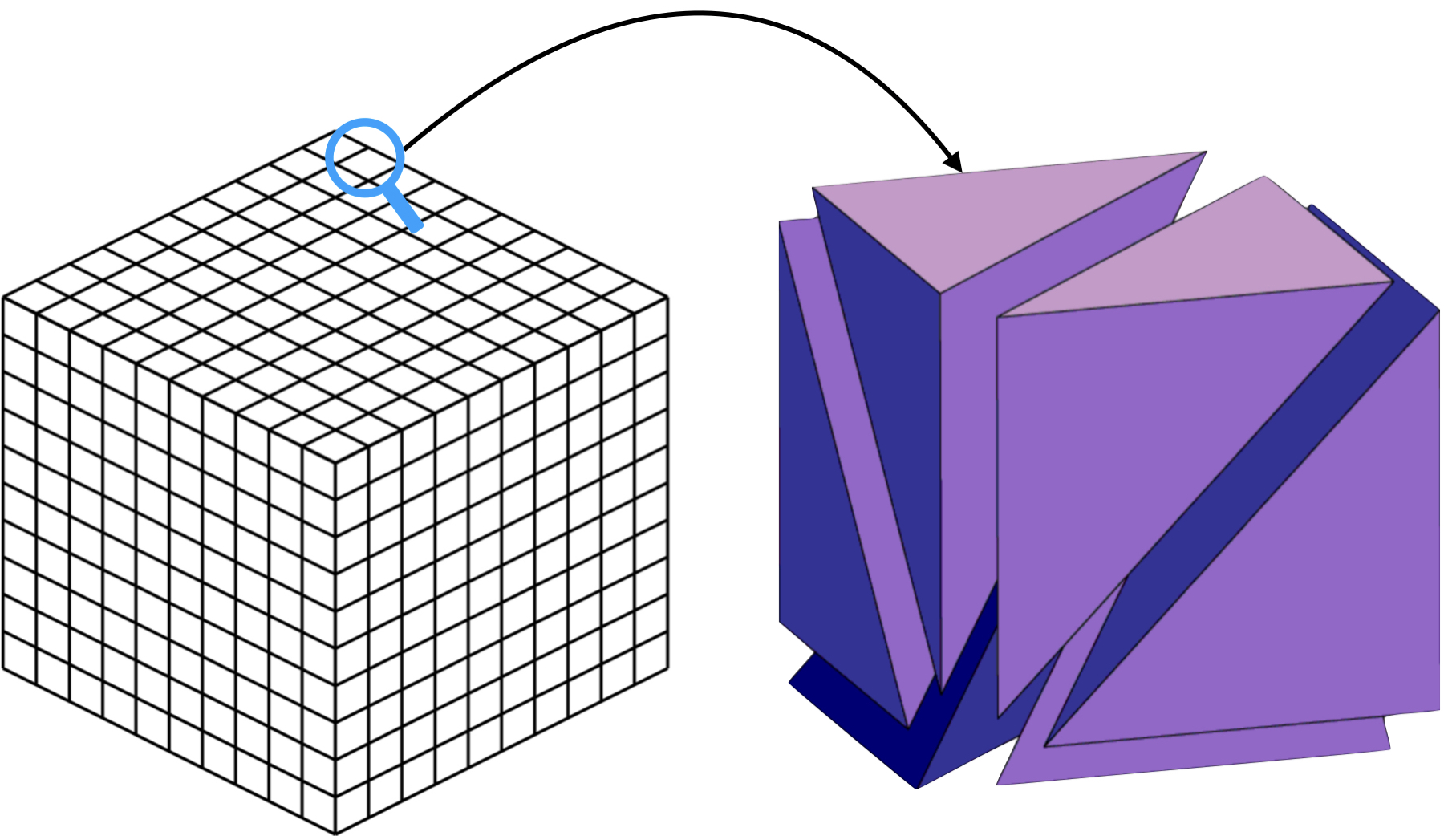}
  \end{center}
\end{wrapfigure}
Consider the $d$-dimensional 
hypercube, denoted by $\textsc{HyperCube}=[0,1]^d$. A canonical simplicization of the hypercube is a partition of \textsc{HyperCube} into cubelets  and  division of  each  cubelet  into  simplices  in  the  canonical  way, known also as \emph{Coxeter-Freudenthal-Kuhn triangulation} \cite{kuhn1960some} such that any two simplices  either are disjoint or share a full face of a certain dimension.

A \emph{Sperner coloring} $T$ of a simplicization of~$\textsc{HyperCube}$ is then an assignment of $d+1$ colors $\{0,1,\ldots,d\}$ to  vertices of the simplicization (union of the vertices of simplices that make up the simplicization) such that:
\begin{itemize}[topsep=0pt,leftmargin=5ex]
\setlength{\itemsep}{0pt}
\setlength{\parskip}{.2ex}
    \item  For all $i\in\{1,\ldots,d\}$, none of the vertices on the face $x_i=0$ uses color $i$.
    \item Moreover, color $0$ is not used by any vertex on a face $x_i= 1$, for some $i\in\{1,\ldots,d\}$.
\end{itemize}
A \emph{panchromatic} simplex of $T$ is one in the simplicization whose vertices have all the $d+1$ colors.
We are ready now to review Sperner's lemma:
\begin{lemma}[Sperner's Lemma \citep{sperner1928neuer}] 
Every Sperner coloring
  $T$ of any simplicization of $\textsc{HyperCube}$ has a panchromatic simplex.  In  fact,  there  is  an  odd  number  of those.
\end{lemma}
Before proceeding with the formal proof of \PPAD membership (with its added burden of rigorously attending to complexity-theoretic details), we provide an informal argument for the existence of  an approximate \Kakutani\ fixed point which forms the basis of its \PPAD membership proof. We~will assign a color to each point of simplicization $\vecx$ (informally) as follows: 
\begin{itemize}[topsep=0pt,leftmargin=3ex]
\setlength{\itemsep}{0pt}
\setlength{\parskip}{.2ex}
\item If $\vecx$ is fixed point then we are done; otherwise we compute $G(\vecx)=\Pi_{F(\vecx)}(\vecx)-\vecx$, where
$\Pi_{F(\vecx)}(\vecx)$ is the projection of $\vecx$ in $F(\vecx)$. Then if $G(\vecx)$ belongs to the positive orthant then it is colored $0$, otherwise it is colored with the first lexicographically coordinate which is non-positive. We tie-break at the boundaries to ensure that coloring is a \emph{Sperner's} one. Sperner's lemma implies the existence of a panchromatic simplex $S$.
\item It follows from our coloring that by proving the Lipschitzness of $\vecp(\vecx)=\Pi_{F(\vecx)}(\vecx)$,
  we can show that when the simplicization is \emph{fine} enough,  
  there exists a point in a panchromatic simplex yields a \Kakutani\ fixed point.
\end{itemize}
With the blueprint of the proof in place, we now proceed with the formal proof that places the problem of computing $\eps$-approximate \Kakutani\ fixed point in \PPAD. In order to follow 
the aforementioned proof sketch,
the main challenges that should be addressed are \emph{i)} to prove the fault tolerance of the above argument even if we can only compute an approximation of the projection of a point $\vecx$ in a convex set $S$ instead of the exact one, \emph{ii)} to resolve the boundary issues that may be apparent when approximate projection process outputs ``accidentally'' a point outside of \textsc{HyperCube}, \emph{iii)} to develop a robustification of Maximum Theorem
of Claude  \cite{bonsall1963c}, in order to prove the Lipschitz-continuity of the described map $G(\vecx)$ and finally \emph{iv)} to address the syntactic challenges whenever $F(\vecx)$ violates any of the \emph{well-conditioned} assumptions. 

Our first preliminary step is to describe rigorously the high dimension of Sperner that we will leverage in our membership reduction. Thus, we proceed to describe formally the canonical simplicization of the hypercube. The domain of this problem is a $d$-dimensional grid $[N]^d$, with $N$ discrete
  points in each direction. It is not restrictive to assume that $N = 2^{\ell}$ 
  for some natural number $\ell$. Hence, we can represent any number in $[N]$ with
  a binary string of length $\ell$, i.e., a member of the set $\{0, 1\}^{\ell}$.
  For this reason, in the rest of the proof we will use members of $[N]$ and 
  members of $\{0, 1\}^{\ell}$ interchangeably and it will be clear from the context which representation we are referring to.

  The input to the Sperner 
  problem that we use for our proof is a boolean circuit $\calC_l$ with $\ell d$
  inputs and has $ \left \lceil \log d \right \rceil$ output gates to encode the output of the Sperner coloring $T$ of $(d+1)$ colors: $\{0,1,\ldots,d\}$. For the total version of the problem, we are asked to find either a panchromatic simplex or a violation of the rules of proper Sperner coloring.
 \newcommand{\HD}{\textsc{HD}}
  \newcommand{\Sperner}{\textsc{Sperner}}
  \newcommand{\HDSperner}{\HD$-$\Sperner}
 {\small \begin{nproblem}[\HD$-$\Sperner]\label{d:Hd-Sperner}
    \textsc{Input:} A boolean circuit 
    $\calC_l : \underbrace{\{0, 1\}^\ell \times \dots \times \{0, 1\}^\ell}_{d \text{ times}} \to \{0, 1, \ldots, d\}^d$\\
    \smallskip
    
    \noindent \textsc{Output:} One of the following:
    \begin{enumerate}[topsep=0pt,left=2ex,itemsep=0pt,parsep=0.0ex]
            \item[0a.] \textit{(Violation of color in $x_i=0$ boundary)} \\ A vertex $\vecv \in \p{\{0, 1\}^\ell}^d$ with $v_i = 0$ such that
      $\calC_l(\vecv) = i$.
                  \item[0a.] \textit{(Violation of color in $x_i=1$ boundary)} \\ A vertex $\vecv \in \p{\{0, 1\}^\ell}^d$ with $v_i = 1$ such that
      $\calC_l(\vecv) = 0$.
    \item[1.] \textit{(Panchromatic Simplex)} \\
    One sequence of $d+1$ vertices $\vecv^{(0)}$, $\dots$, $\vecv^{(d)}$ 
      with 
      $\vecv^{(i)} \in \p{\{0, 1\}^\ell}^d$ such that 
      $\calC_l(\vecv^{(i)}) = i$.
    \end{enumerate}
  \end{nproblem}}
  \begin{remark}
By the above definition, notice that we have actually asserted that $[0,1]^d$ has been divided in cubelets of length $\tfrac{1}{N -1}$, where $N=2^\ell$.
  \end{remark}
The following PPAD membership result follows from similar ideas of \cite{Papadimitriou1994CC}:
\begin{theorem}[Theorem 3.4 \citep{chen2021complexity}]
Given as input a Boolean circuit as described that encodes a Sperner coloring of Kuhn’s simplicization for some $N=2^\ell$ and $d$, the problem of \HD$-$\Sperner\  is in \PPAD.
\end{theorem}

\paragraph{Main Membership Reduction \& A Lipschitz extension of Berge's Maximum Theorem.} We are ready to proceed to our reduction from \HDSperner.
For every point of the $d$-dimensional hypercube we compute in polynomial time in $(d,\log(1/\eps))$ the following vector field: \[G(\vecv)=\algoprojectionto{F(\vecv)}{\vecv}{\parameters}-\vecv.\] In the case that $\algoprojectionto{F(\vecv)}{\vecv}{\parameters}$ would output $\bot$, then we are done and since we will merely output the certificate for the violation of $\eta$-non-emptiness that \ProbWAP\ provided. 
Otherwise, we proceed with the construction of the coloring circuit.

In particular, color $i$ is allowed if $G(\vecv)_i\leq 0$. Color $0$ is allowed if $G(\vecv)_i\geq 0$ for all $i\in\{1,\ldots,d\}$. Ties are broken to avoid violating the coloring requirements of Sperner.\footnote{
 Notice that by construction, \ProbWAP\ syntactically would never project at a point outside of the boundary of the hypercube. In other words, even if whole $F(\vecx)$ is outside of $[0,1]^d$, by applying our box constraints in Ellipsoid method, our algorithm would be restricted only on $F(\vecx)\cap [0,1]^d$. Therefore, as long as \ProbWAP\ will not provide a violation of $\eta$-non-emptiness, a valid instance of Sperner coloring is always possible and polynomially computable. }.
 Hence, we obtain a valid instance of \HDSperner\ .
 Since, a valid coloring has been enforced
 a panchromatic simplex is returned as  a solution of this instance, let $\vecv^{(0)},\vecv^{(1)},\ldots,\vecv^{(d)}$ be the vertices colored by $0,1,\ldots,d$
 We show how to obtain a solution to our original \Kakutani \ \ instance. By our coloring rule it holds the following inequality holds for all $i \in \{1, \ldots, d\}$:
  \begin{align}
  G(\vecv^{(0)})_i\cdot G(\vecv^{(i)})_i&\leq 0\Rightarrow 
 |G(\vecv^{(0)})_i|\leq |G(\vecv^{(0)})_i - G(\vecv^{(i)})_i|\notag
 \end{align}
 And correspondingly we have that:
 \begin{align}
  \left|\left(\algoprojectionto{F(\vecv^{(0)})}{\vecv^{(0)}}{\parameters}-\vecv^{(0)}\right)_i\right|& \leq
  \left|\left(
 \algoprojectionto{F(\vecv^{(0)})}{\vecv^{(0)}}{\parameters} -   \algoprojectionto{F(\vecv^{(i)})}{\vecv^{(i)}}{\parameters}\right)_i 
   -\left(\vecv^{(0)}-\vecv^{(i)}\right)_i 
 \right|
\notag\\
   &\leq \left|\left(\algoprojectionto{F(\vecv^{(0)})}{\vecv^{(0)}}{\parameters} - \algoprojectionto{F(\vecv^{(i)})}{\vecv^{(i)}}{\parameters}\right)_i\right|+\left|\left(\vecv^{(0)}-\vecv^{(i)}\right)_i\right|\notag\\
  &\leq \left\|\algoprojectionto{F(\vecv^{(0)})}{\vecv^{(0)}}{\parameters} - \algoprojectionto{F(\vecv^{(i)})}{\vecv^{(i)}}{\parameters}\right\|_{2}+\gridRefinement \label{eq:bound-from-coloring}
\end{align}
\noindent where  $\vecx_i$ corresponds to the $i$-th coordinate of vector $\vecx$ and  $\gridRefinement$ is a bound for the longest distance between two vertices of a panchromatic simplex in a cubelet, namely $\gridRefinement=\frac{\sqrt{d}}{N-1}$. 
Intuitively, in order to prove that $\vecv^{(0)}$ is a $\alpha$-approximate \Kakutani\ fixed point for some  $\alpha=\alpha(\gridRefinement,\eta,\epsilon)$, it suffices to prove some Lipschitz condition of the form for the map $\widehat{\Phi}(\vecv)
~=~\algoprojectionto{F(\vecv)}{\vecv}{\parameters} $, e.g.
$\left\|\widehat{\Phi}(\vecv^{(0)})- \widehat{\Phi}(\vecv^{(i)})\right\|_{2} 
\le O_{d,\eta}(\epsilon+\gridRefinement)$.

\noindent Thus, the rest of this section is devoted to prove the following lemma of independent interest:
\begin{lemma}\label{lem:our-berge-for-projections}
Let $F$ be an $L$-Hausdorff Lipschitz continuous and well-bounded correspondence, such that $\exists\veca_0\in[0,1]^d:\parallelbody{\veca_0}{\eta}\subseteq F(\vecx)\subseteq \parallelbody{\mathbf{0}}{\sqrt{d}}$ for all $\vecx\in[0,1]^d$. Additionally, let  any two vectors $\vecp,\vecq\in[0,1]^d$ with distance at most $\|\vecp-\vecq\|\leq \gridRefinement$ and any positive constants $\epsilon,\epsilon^\circ$, such that $\|\algoprojectionto{F(\vecp)}{\algoprojectionto{F(\vecq)}{\vecq}{\parameters}}{\parameters}-{\algoprojectionto{F(\vecq)}{\vecq}{\parameters}}\|\leq L\gridRefinement+\hat{\calL}_{d,\eta}\cdot\eps+\hat{\calL^{\circ}}_{d,\eta}\eps^\circ$ for some constants $\hat{\calL}_{d,\eta},\hat{\calL^{\circ}}_{d,\eta}$. Then, it holds that :
     \[ 
     \left\|
\algoprojectionto{F(\vecp)}{\vecp}{\parameters} - \algoprojectionto{F(\vecq)}{\vecq}{\parameters}\right\|
\le 2\sqrt[4]{d}\sqrt{2error_{d,\eta,L}(\xi,\eps,\eps^\circ)}+error_{d,\eta,L}(\xi,\eps,\eps^\circ)\]
where $error_{d,\eta,L}(\xi,\eps,\eps^\circ)=(L+1)\gridRefinement+{\hat{\calL'}_{d,\eta}\cdot\eps}+\hat{\calL^{\circ}}_{d,\eta}\eps^\circ$ and $\hat{\calL'}_{d,\eta}=\hat{\calL}_{d,\eta}+2\hat{c}_{d,\eta}$.
\end{lemma}
\begin{lemma}\label{lem:our-berge-for-strong-projections}
Let $F$ be an $L$-Hausdorff Lipschitz well-bounded correspondence, such that $\exists\veca_0\in[0,1]^d:\parallelbody{\veca_0}{\eta}\subseteq F(\vecx)\subseteq \parallelbody{\mathbf{0}}{\sqrt{d}}$ for all $\vecx\in[0,1]^d$. Then  for any two vectors $\vecp,\vecq\in[0,1]^d$ with distance at most $\|\vecp-\vecq\|\leq \gridRefinement$ and any constant $\epsilon>0$, such that $\|\strongalgoprojectionto{F(\vecp)}{\strongalgoprojectionto{F(\vecq)}{\vecq}{\parameters}}{\parameters}-{\strongalgoprojectionto{F(\vecq)}{\vecq}{\parameters}}\|\leq L\gridRefinement+3\eps$, it holds that :
     \[ 
     \left\|
\strongalgoprojectionto{F(\vecp)}{\vecp}{\parameters} - \strongalgoprojectionto{F(\vecq)}{\vecq}{\parameters}\right\|
\le 
 2\sqrt[4]{d}\sqrt{2(L+1)\gridRefinement+10\eps}
+ (L+1)\gridRefinement+5\eps\]
\end{lemma}

\begin{proof}
\begin{align}
\left\|\widehat{\Phi}(\vecp)- \widehat{\Phi}(\vecq)\right\|&= \left\|\algoprojectionto{F(\vecp)}{\vecp}{\parameters} - \algoprojectionto{F(\vecq)}{\vecq}{\parameters}\right\|\\
&= \begin{Bmatrix}\underbrace{\left\|\algoprojectionto{F(\vecp)}{\vecp}{\parameters} - \algoprojectionto{F(\vecp)}{\vecq}{\parameters}\right\|}_{(A)}\\+\\ \underbrace{\left\|\algoprojectionto{F(\vecp)}{\vecq}{\parameters} - \algoprojectionto{F(\vecq)}{\vecq}{\parameters}\right\|}_{(B)}\end{Bmatrix}
\end{align}

For the term $(A)$, we have to bound the distance of the approximate projection of two $\gridRefinement$-close vectors to the same set. To do so, we apply the approximation bound of Theorem~\ref{thm:approximation-projection} to both terms and 1-Lipschitzness of projection operator:
\begin{align}
    (A)&=\left\|\algoprojectionto{F(\vecp)}{\vecp}{\parameters} - \algoprojectionto{F(\vecp)}{\vecq}{\parameters}\right\|\le
    \begin{Bmatrix}
    \left\|\algoprojectionto{F(\vecp)}{\vecq}{\parameters} - \projectionto{\parallelbody{F(\vecp)}{\eps}}{\vecq}\right\|\\ +\\ \left\|\algoprojectionto{F(\vecp)}{\vecp}{\parameters} - \projectionto{\parallelbody{F(\vecp)}{\eps}}{\vecp}\right\| \\+\\
    \left\|\projectionto{\parallelbody{F(\vecp)}{\eps}}{\vecp} - \projectionto{\parallelbody{F(\vecp)}{\eps}}{\vecq}\right\|\notag
    \end{Bmatrix}\\
    &\le {\hat{c}_{d,\eta}\epsilon}+{\hat{c}_{d,\eta}\epsilon}+\left\|\vecp-\vecq\right\|_{\infty}\leq 2\cdot{\hat{c}_{d,\eta}\epsilon}+\gridRefinement
\end{align}

For the more challenging term (B), we have to bound the distance of the approximate projection of a single point to two $L \gridRefinement$-Hausdorff distance close sets. To do so, we will prove that both  $\Gamma(\vecx)=\|\vecq-\algoprojectionto{F(\vecx)}{\vecq}{\parameters}\|$ and $\Delta(\vecx)=\algoprojectionto{F(\vecx)}{\vecq}{\parameters}$ are approximately Lipschitz continuous functions. Let's examine firstly the $\Gamma$ function:
\begin{align}
    \Gamma(\vecp)
    =&
    \|\vecq-\algoprojectionto{F(\vecp)}{\vecq}{\parameters}\|
    \underset{\textbf{Thm}.\ref{thm:approximation-projection}}{\leq}
    \|\vecq-\projectionto{\parallelbody{F(\vecp)}{\epsilon}}{\vecq}\|_2+{\hat{c}_{d,\eta}\cdot\eps}\\
    \underset{}{\leq}&\label{eq:step-explanation-1}
    \|\vecq-\projectionto{\parallelbody{F(\vecp)}{\epsilon}}{\algoprojectionto{F(\vecq)}{\vecq}{\parameters}}\|+{\hat{c}_{d,\eta}\cdot\eps}\\
    \leq&\label{eq:step-explanation-2}
    \|\vecq-\algoprojectionto{F(\vecp)}{\algoprojectionto{F(\vecq)}{\vecq}{\parameters}}{\parameters}\|+2{\hat{c}_{d,\eta}\cdot\eps}\\
    \underset{}{\leq}&
    \|\vecq-\algoprojectionto{F(\vecq)}{\vecq}{\parameters}\|+\|\algoprojectionto{F(\vecp)}{\algoprojectionto{F(\vecq)}{\vecq}{\parameters}}{\parameters}-{\algoprojectionto{F(\vecq)}{\vecq}{\parameters}}\|+2{\hat{c}_{d,\eta}\cdot\eps}\\
      =&\Gamma(\vecq)+\|\algoprojectionto{F(\vecp)}{\algoprojectionto{F(\vecq)}{\vecq}{\parameters}}{\parameters}-{\algoprojectionto{F(\vecq)}{\vecq}{\parameters}}\|+2{\hat{c}_{d,\eta}\cdot\eps}\\
      \underset{}{\leq}&\Gamma(\vecq)+L\|\vecp-\vecq\|+\hat{\calL^{\circ}}_{d,\eta}\eps^\circ+\hat{\calL}_{d,\eta}\cdot\eps+2\cdot {\hat{c}_{d,\eta}\cdot\eps}\le
      \Gamma(\vecq)+L\gridRefinement+\hat{\calL'}_{d,\eta}\eps+\hat{\calL^{\circ}}_{d,\eta}\eps^\circ
\end{align}
where \eqref{eq:step-explanation-1} is derived by the optimality of the projection for function $(\ell_2^2)_{\vecq}(\vecz)=\|\vecq-\vecz\|^2/2$ in $\parallelbody{F(\vecq)}{\eps}$, namely  for all $ \vecx \in  \parallelbody{F(\vecp)}{\epsilon}: \|\vecq-\projectionto{\parallelbody{F(\vecp)}{\epsilon}}{\vecq}\|_2^2\leq \|\vecq-\projectionto{\parallelbody{F(\vecp)}{\epsilon}}{\vecx}\| $,\\
while \eqref{eq:step-explanation-2} is again by Theorem~\ref{thm:approximation-projection} and $\hat{\calL'}_{d,\eta}=\hat{\calL}_{d,\eta}+2\hat{c}_{d,\eta}$.
Symmetrically, for $\vecp$ we get that:
\[
\Gamma(\vecq)\le       \Gamma(\vecp)+L\gridRefinement+\hat{\calL'}_{d,\eta}\eps+\hat{\calL^{\circ}}_{d,\eta}\eps^\circ
\]
which yield that $
|\Gamma(\vecq)- \Gamma(\vecp)|\le L\gridRefinement+\hat{\calL'}_{d,\eta}\eps+\hat{\calL^{\circ}}_{d,\eta}\eps^\circ$. \\
Secondly, for the function $\Delta(\vecx)$, we derive the following bound:
\begin{align}
    \|\Delta(\vecq)-\Delta(\vecp)\|&=
    \|\algoprojectionto{F(\vecq)}{\vecq}{\parameters}-\algoprojectionto{F(\vecp)}{\vecq}{\parameters}\|\\
    &\le \begin{Bmatrix}
    \|\algoprojectionto{F(\vecq)}{\vecq}{\parameters}-\algoprojectionto{F(\vecp)}{\algoprojectionto{F(\vecq)}{\vecq}{\parameters}}{\parameters}\|\\+\\\|\algoprojectionto{F(\vecp)}{\algoprojectionto{F(\vecq)}{\vecq}{\parameters}}{\parameters}-{\algoprojectionto{F(\vecp)}{\vecq}{\parameters}}\|
    \end{Bmatrix}\\
    &\le L\|\vecp-\vecq\|+{\hat{\calL}_{d,\eta}\cdot\eps}+\hat{\calL^{\circ}}_{d,\eta}\eps^\circ+\|\algoprojectionto{F(\vecp)}{\algoprojectionto{F(\vecq)}{\vecq}{\parameters}}{\parameters}-{\algoprojectionto{F(\vecp)}{\vecq}{\parameters}}\|\\
        &\le L\gridRefinement+{\hat{\calL}_{d,\eta}\cdot\eps}+\hat{\calL^{\circ}}_{d,\eta}\eps^\circ+
        \|\algoprojectionto{F(\vecp)}{\algoprojectionto{F(\vecq)}{\vecq}{\parameters}}{\parameters}-{\algoprojectionto{F(\vecp)}{\vecq}{\parameters}}\|
\end{align}
where in the last inequality, we have applied the approximate version of Lipschtzness assumption.
\begin{figure}[h!]
    \centering
  \begin{tikzpicture}[bullet/.style={circle,fill,inner sep=1.5pt,node contents={}},scale=0.9]
 \draw (-4,-2.5) node{\small$\parallelbody{F(\vecp)}{\eps}$} ;
 \draw (-4,-0.5) node{\small$F(\vecp)$} ;
 \draw (2,-2.5) node{$\parallelbody{F(\vecq)}{\eps}$} ;
 \draw (2,-0.5) node{$F(\vecq)$} ;
  \draw (-2.3,0.5) node[bullet,label=left:$ $] ;
    \draw (-2,0) node{$\star$} ;
  \draw (2.3,0.5) node[bullet,label=left:$ $] ;
  \draw (-1.9,-1.1) node{$\square$} ;
  \draw (-2.2,-1.5) node{$(M)$} ;
    \draw (0.8,0) node{$\star$} ;  
    \draw (0,1.4) node{$(K)$};
    \draw (0,1) node{$\blacklozenge$};
    \draw (-2.3,0) node{$(\Lambda)$}; 
    \filldraw[thick,orange,opacity=0.2] (0,1) -- (-1.9,-1.1) -- (-2,0)-- cycle;
    \draw[thick,orange] (0,1) -- (-1.9,-1.1) -- (-2,0)-- cycle;
    \draw[thick,cyan] (-1.95,-0.55)--(0,1);
     \draw (-1.95,-0.55) node[bullet,label=left:$(N)$] ;
    \draw (0,1) -- (-2.3,0.5);
    \draw (0,1) -- (2.3,0.5); 
    \draw[dashed] (0,1) -- (-2,0);
    \draw[dashed] (0,1) -- (0.85,0); 
    \draw[dashed]  (0.85,0) -- (-1.9,-1.1);
\def\x{1}
\def\y{2}
\filldraw[fill=lightgray,opacity=0.2] (-2,0)          to[out=105,in=0] ++ 
          (-1.75*\x,1*\x) to[out=180,in=105] ++ 
          (-3.5*\x,-4*\x) node[above right]{$ $} 
                    to[out=-75,in=180] ++ 
            (3.25*\x,-1*\x) to[out=0,in=-75] cycle;
            
\draw[fill=lightgray,opacity=0.2] (-3,0)          to[out=105,in=0] ++ 
          (-1.75/\y,1/\y) to[out=180,in=105] ++ 
          (-3.5/\y,-4/\y) node[above right]{$ $} 
                    to[out=-75,in=180] ++ 
            (3.25/\y,-1/\y) to[out=0,in=-75] cycle;            
\def\xx{0.5}
\def\yy{5}
\filldraw[fill=lightgray,opacity=0.2] (3,0)          
to[out=105,in=0] ++ 
          (-2*\xx,1*\xx) to[out=180,in=105] ++ 
          (-3.5*\xx,-6*\xx) node[above right]{$ $} 
                    to[out=-75,in=180] ++ 
            (3.25*\xx,-3*\xx) to[out=0,in=-75] cycle;
\filldraw[fill=lightgray,opacity=0.2] (2.5,0)          
to[out=105,in=0] ++ 
          (-2/\yy,1/\yy) to[out=180,in=105] ++ 
          (-3.5/\yy,-6/\yy) node[above right]{$ $} 
                    to[out=-75,in=180] ++ 
            (3.25/\yy,-3/\yy) to[out=0,in=-75] cycle;
\end{tikzpicture}

\end{figure}
For the analysis of the last term, we will dig into the geometry of the problem. Let's focus our attenction on the triangle $\overset{\triangle}{K\Lambda M}$, where $(K)=\vecq$, $(\Lambda)=\algoprojectionto{F(\vecp)}{\vecq}{\parameters}$ and $(M)=\algoprojectionto{F(\vecp)}{\algoprojectionto{F(\vecq)}{\vecq}{\parameters}}{\parameters}$
By Apollonius' theorem  we get that:
\begin{equation}
|K\Lambda|^2 + |KM|^2 = 2(|KN|^2+|\Lambda N|^2)
\label{eq:Apollonius-theorem}
\end{equation}
where $KN$ is a median ($|\Lambda N|=|KN|=|K\Lambda|/2$). Thus, we can rewrite the median as 
\begin{equation}
\label{eq:median-alternative}
|KN|^2=\frac{
|K\Lambda|^2 + |KM|^2 }{2} -\frac{|\Lambda M|^2}{4}
\end{equation}
Additionally, due to the convexity of the triangle, it holds that:
\begin{equation}
\min\{|K\Lambda|^2, |KM|^2\}\le|KN|^2\le\max\{|K\Lambda|^2, |KM|^2\}
\label{eq:bounds-triangle}    
\end{equation}

Combining \eqref{eq:bounds-triangle} and \eqref{eq:median-alternative}, we get that 
\begin{align*}
 \min\{|K\Lambda|^2, |KM|^2\} \le   |KN|^2&=\frac{
|K\Lambda|^2 + |KM|^2 }{2} -\frac{|\Lambda M|^2}{4}&\Leftrightarrow\\
\min\{|K\Lambda|^2, |KM|^2\} &\le
\frac{
\max\{|K\Lambda|^2, |KM|^2\} +
\min\{|K\Lambda|^2, |KM|^2\} }{2} -\frac{|\Lambda M|^2}{4}&\Leftrightarrow\\
{|\Lambda M|^2}
 &\le
2(\max\{|K\Lambda|^2, |KM|^2\} -
\min\{|K\Lambda|^2, |KM|^2\} )&\Leftrightarrow\\
{|\Lambda M|}
 &\le\sqrt{
2}\sqrt{(\max\{|K\Lambda|, |KM|\} -
\min\{|K\Lambda|, |KM|\} )}\sqrt{(
|K\Lambda|+ |KM| )}&\Leftrightarrow\\
{|\Lambda M|}&\le
\sqrt{2}\sqrt{\left||K\Lambda|- |KM|\right| }\sqrt{2\cdot \text{SpaceDiameter}}
\end{align*}
Going back to our problem, we get that:
\begin{equation}
\|\algoprojectionto{F(\vecp)}{\algoprojectionto{F(\vecq)}{\vecq}{\parameters}}{\parameters}-{\algoprojectionto{F(\vecp)}{\vecq}{\parameters}}\|\le 2\sqrt[4]{d}\sqrt{\left|\Gamma(\vecp)-\|\vecq-\algoprojectionto{F(\vecp)}{\algoprojectionto{F(\vecq)}{\vecq}{\parameters}}{\parameters}\|\right|}
\end{equation}
In this point, we will leverage the approximate lipschitzness of $\Gamma(\cdot)$:
\begin{align*}
    \left|\Gamma(\vecp)-\|\vecq-\algoprojectionto{F(\vecp)}{\algoprojectionto{F(\vecq)}{\vecq}{\parameters}}{\parameters}\|\right|
    &=
    \left|\Gamma(\vecp)-\Gamma(\vecq)+\Gamma(\vecq)-\|\vecq-\algoprojectionto{F(\vecp)}{\algoprojectionto{F(\vecq)}{\vecq}{\parameters}}{\parameters}\|\right|\\
    \le
    \begin{Bmatrix}
    \left|\Gamma(\vecp)-\Gamma(\vecq)\right|
    \\+\\
    \left|\Gamma(\vecq)-\|\vecq-\algoprojectionto{F(\vecp)}{\algoprojectionto{F(\vecq)}{\vecq}{\parameters}}{\parameters}\|\right|
    \end{Bmatrix}
     &\le
    \begin{Bmatrix}
    L\gridRefinement+{\hat{\calL'}_{d,\eta}\cdot\eps}+\hat{\calL^{\circ}}_{d,\eta}\eps^\circ+
    \\+\\
    \left|
    \|\vecq-\algoprojectionto{F(\vecq)}{\vecq}{\parameters}\|
    -\|\vecq-\algoprojectionto{F(\vecp)}{\algoprojectionto{F(\vecq)}{\vecq}{\parameters}}{\parameters}\|\right|
    \end{Bmatrix}\\
    \le
    \begin{Bmatrix}
    L\gridRefinement+\hat{\calL'}_{d,\eta}\eps+\hat{\calL^{\circ}}_{d,\eta}\eps^\circ
    \\+\\
    \left|
    \|\algoprojectionto{F(\vecq)}{\vecq}{\parameters}-\algoprojectionto{F(\vecp)}{\algoprojectionto{F(\vecq)}{\vecq}{\parameters}}{\parameters}\|\right|
    \end{Bmatrix}
    &\le
    2L\gridRefinement+2\hat{\calL'}_{d,\eta}\eps+2\hat{\calL^{\circ}}_{d,\eta}\eps^\circ
\end{align*}
Wrapping everything up, we get that 
\[\|\Delta(\vecq)-\Delta(\vecp)\|\le 
 2\sqrt[4]{d}\sqrt{2L\gridRefinement+2\hat{\calL'}_{d,\eta}\eps+2\hat{\calL^{\circ}}_{d,\eta}\eps^\circ}+
 L\gridRefinement+{\hat{\calL}_{d,\eta}\cdot\eps}+\hat{\calL^{\circ}}_{d,\eta}\eps^\circ
 \]
\end{proof}
\noindent Going back to our reduction, Eq.\eqref{eq:bound-from-coloring} yields that
$$\|G(\vecv^{(0)})\|=\|\algoprojectionto{F(\vecv^{(0)})}{\vecv^{(0)}}{\parameters}-\vecv^{(0)}\|\leq 3d^{3/4}\sqrt{error_{d,\eta,L}(\xi,\eps,\eps^\circ)}+error_{d,\eta,L}(\xi,\eps,\eps^\circ)+\sqrt{d}\xi$$
For well-conditioned maps, we have that $\hat{\calL}_{d,\eta}=3(1+\hat{c}_{d,\eta})$ and $\eps^\circ=0$ (see Lemma~\ref{lem:approximation-lipschitzness}).
Let us choose now sufficiently small $ \epsilon$ and the mesh of the grid $N$ sufficiently large:
\[
\begin{cases}\epsilon\leq \min\{\tfrac{\alpha/10}{13},\tfrac{(\alpha/10)^2}{117d^{3/2}}\}=O(poly(\alpha,1/d))\\
N\geq\max\{\tfrac{d}{(\alpha/10)},\tfrac{\sqrt{d}(L+1)}{(\alpha/10)},\tfrac{9d^{2.5}}{(\alpha/10)^2},\tfrac{9d^2(L+1)}{(\alpha/10)}\}=O(L,poly(1/\alpha,d))
\end{cases}
\]
such that $error_{d,\eta,L}(\xi,\eps,\eps^\circ)\leq \alpha/10$ and $3d^{3/4}\sqrt{error_{d,\eta,L}(\xi,\eps,\eps^\circ)}\leq \alpha/10$ and $\sqrt{d}\xi(N)\leq \alpha/10$.
Thus, for each of the vertices $\calV=\{\vecv^{(0)},\ldots,\vecv^{(d)}\}$ of the panchromatic triangle,
we compute the corresponding vectors: $(a)\ \ 
\algoprojectionto{F(\vecv)}{\vecu}{\parameters}$ and $(b)\ \ 
\algoprojectionto{F(\vecv)}{\algoprojectionto{F(\vecu)}{\vecu}{\parameters}}{\parameters}
$, for any $\vecv,\vecu\in\calV$. If any $\eta-$non-emptiness or $L$-almost Lipschitzness with violation appeared, we output the corresponding certificate, as it has been described in the initial sections. Otherwise, for well-suited choices of $\epsilon$ and $\gridRefinement(N)$ 
\[(i)\ 
\|G(\vecv^{(0)})\|=\|\algoprojectionto{F(\vecv^{(0)})}{\vecv^{(0)}}{\parameters}-\vecv^{(0)}\|\leq \tfrac{1}{2}\alpha\text{ and } (ii)\  \hat{c}_{d,\eta}\eps\leq \tfrac{\alpha}{10}
\]
and consequently by Theorem~\ref{thm:approximation-projection} and Lemma~\ref{lem:from-out-to-border}, we get that
\[
\|G(\vecv^{(0)})\|=\|\projectionto{F(\vecv^{(0)})}{\vecv^{(0)}}-\vecv^{(0)}\|\leq \alpha.
\]
which conclude our inclusion proof. It is easy to see that the cases with $\mathrm{SO}_F,\mathrm{ProjO}_F$ run for similar $\epsilon$ and $\gridRefinement$.
\begin{remark}[Different metrics]\label{rem:different_metrics}
It is easy to verify and for the rest of our work we will consider it as proven that the aforementioned inclusion proof of {\sc Kakutani} also holds with the adequate choice of accuracy parameters for any Hausdorff Lipschitzness/H\"{o}lder Continuity of any metric of the form $\hausdorff(F(x),F(y))\le L \|x-y\|_p^q$ for any power $p,q>0$.
\end{remark}


\subsection{Hardness of {\sc Kakutani} in $\PPAD$} \label{sec:Kakutani:hardness}
It is not hard to see that the above computational version of Kakutani's fixed  point theorem is $\PPAD$-hard to compute, simply via a reduction from $\textsc{Brouwer}$.
{\small\begin{nproblem}[\textsc{Brouwer}]
  \textsc{Input:} Scalars $L$ and $\gamma$ and a polynomial-time Turing machine $\calC_M$ evaluating a $L$-Lipschitz function  $M : [0, 1]^d \to [0, 1]^d$.
  \smallskip
  
  \noindent \textsc{Output:} A point $\vecz \in [0, 1]^d$ such that $\norm{\vecz - M(\vecz)} \le \gamma$.
\end{nproblem}}
\begin{remark}
To satisfy syntactically the Lipschtz continuity of the map, we can assume that $\calC_M$ is given by a linear arithmetic circuit.
Notice that using the Appendix A.2 in \cite{fearnley2021complexity} we conclude that such a circuit is a Lipschitz function with Lipschitz constant less than or equal to  $\exp(\size^2(\calC_M))$. 
\end{remark}
While not stated exactly in this form, the following is a straightforward implication of the results presented in \cite{chen2009settling}.
\begin{lemma}[\cite{chen2009settling}]. \textsc{Brouwer} is $\PPAD$-complete even when $d = 2$. Additionally, \textsc{Brouwer} is
$\PPAD$-complete even when $\gamma = \poly(1/d)$ and $L = \poly(d)$.
\end{lemma}

\begin{lemma}[Restated Lemma~\ref{lem:Kakutani:hardness}] \label{app:lem:Kakutani:hardness}
  The computational problems of $\Kakutani$ with $\mathrm{WSO}_F,\mathrm{SO}_F,\mathrm{ProjO}_F$ are $\PPAD$-hard.
\end{lemma}

\begin{proof}
  Given $M$, $L$, $\gamma$ we define the set-valued map $F$ as
  follows:
  \[ F(\vecx) = \bar{\class{B}}(M(\vecx), \gamma/2). \]
  It is easy to see that $F$ is $L$-Lipschitz with respect to the
  Hausdorff distance given that $M$ is $L$-Lipschitz. Also, it is 
  easy to see that $F$ satisfies the non-emptyness condition of 
  Kakutani and that any $\gamma/2$-approximate fixed point of $F$
  corresponds to a $\gamma$-approximate fixed point of $M$, and the
  $\PPAD$-hardness follows.
\end{proof}
\clearpage
\section{Proof of Robust Berge's Theorem }
\label{app:sec:BergeProof}
\begin{proof}[Proof of Theorem~\ref{thm:our-berge-for-optimization}]
    We will follow the proof strategy that we follow  previously leveraging Apollonius theorem.
    Firstly, let's see why $g^*$ is a singled-value function. Suppose that $b_1,b_2\in g^{*}(a)$.
    Thus $f^{*}(a)=\max\{f(a,b):b\in g(a)\}=f(a,b_1)=f(a,b_2)$.
    Let $b_t=b_1\cdot t+(1-t)\cdot b_2$ for some $t\in (0,1)$. Since $g(a)$ is a convex set correspondence, it holds that also $b_t\in g(a)$. Then by convexity we have that:
    \begin{equation}
    \footnotesize f^{*}(a)\geq f(a,b_t)=f(a,b_1\cdot t+(1-t)\cdot b_2)\ge 
    f(a,b_1)\cdot t+(1-t)\cdot f(a,b_2)=f^{*}(a)\cdot t+(1-t)\cdot f^{*}(a)=f^{*}(a)
    \end{equation}
    So, by definition $b_t\in g^{*}(a)$. However, since $f$ is strongly concave, it has unique maximizer, which means that $g^{*}(a)$ is singled-value correspondence necessarily. By initial Berge's theorem, we know that $g^{*}$ would be upper semi-continuous which for the singled-value case corresponds to the classical notion of continuity. We will prove now that $f^{*}(a),g^{*}(a)$ serves some form of Lipschitz 
    continuity. More precisely, for two arbitrary inputs $a_1,a_2\in A$ it holds that:\[
    \begin{cases}
    f^{*}(a_1)=
    \{
    \max f(a_1,b):b\in g(a_1)
    \}=f(g^{*}(a_1),a_1)\\
    f^{*}(a_2)=
    \{
    \max f(a_2,b):b\in g(a_2)
    \}=f(g^{*}(a_2),a_2)\\
    \end{cases}.\]
    Then for the case of $f^*$ it holds the following:
    \begin{align}
    f^{*}(a_1)=
    \{
    \max f(a_1,b):b\in g(a_1)
    \}&\ge f(a_1,b) \ \ \forall b\in g(a_1)\Rightarrow\\
      f^{*}(a_1)&\ge f(\projectionto{g(a_1)}{g^{*}(a_2)},a_1)\\
      &\ge -L\| (\projectionto{g(a_1)}{g^{*}(a_2)},a_1)-(g^{*}(a_2),a_2)\|+f^{*}(a_2)
    \end{align}
Therefore, we have that
    \begin{align}
      f^{*}(a_2)-f^{*}(a_1)&\le L\|a_1-a_2\|+L\| \projectionto{g(a_1)}{g^{*}(a_2)}-g^{*}(a_2)\|\\
      f^{*}(a_2)-f^{*}(a_1)&\le L\|a_1-a_2\|+L\hausdorff(g(a_1),g(a_2)) \\
      f^{*}(a_2)-f^{*}(a_1)&\le (L+L\cdot L')\|a_1-a_2\|
      \end{align}
    
\noindent    Applying symmetrically the same argument for $f^*(a_2)$, we get that $$|f^{*}(a_1)-f^{*}(a_2)|\leq (L+L\cdot L')\left\|a_1-a_2\right\|$$
    
\noindent     About $g^*(a)$, we will leverage the generalization of Apollonius theorem. More precisely, for a $\mu$-strongly concave function
    \[-f(a,\tfrac{x+y}{2})\le -\tfrac{f(a,x)+f(a,y)}{2}-\tfrac{\mu}{8}\|x-y\|_2^2\ \ \forall a\in A\]
    or equivalently, 
        \[\tfrac{f(a,x)+f(a,y)}{2}\le f(a,\tfrac{x+y}{2}) -\tfrac{\mu}{8}\|x-y\|_2^2\ \ \forall a\in A\]
  Since $f(a,\tfrac{x+y}{2})\le\max\{ f(a,x), f(a,y) \}$ we get that
        \[ \|x-y\|\le\sqrt{\tfrac{8}{\mu}\cdot\frac{\max\{ f(a,x), f(a,y) \}-\min\{ f(a,x), f(a,y) \}}{2} }\ \ \forall a\in A\]

    For $K_1=g^{*}(a_1)$,  $K_2= g^{*}(a_2) $ and $K_3=\projectionto{g(a_1)}{K_2}$, we get that 
    \begin{align}
    \|g^{*}(a_1)-g^{*}(a_2)\|=\|K_1-K_2\|&\leq \|K_1-K_3\|+\|K_2-K_3\|\\
    &\leq \hausdorff(g(a_1),g(a_2)) + \|(K_2K_3)\|\\
    &\leq L'\|a_1-a_2\|+\sqrt{\tfrac{4}{\mu}}\sqrt{\left(f(K_1,\alpha_1)-f(K_3,\alpha_1)\right)}\\
    &\leq L'\|a_1-a_2\|+\sqrt{\tfrac{4}{\mu}}\sqrt{\left|f(K_1,\alpha_1)-f(K_3,\alpha_1)+\right|}\\
    &\leq L'\|a_1-a_2\|+\sqrt{\tfrac{4}{\mu}}\sqrt{\begin{Bmatrix}\Big|f(K_1,\alpha_1)-f(K_2,\alpha_2)\\+\\f(K_2,\alpha_2)-f(K_3,\alpha_1)\Big|\end{Bmatrix}}\\
    &\leq L'\|a_1-a_2\|+\sqrt{\tfrac{4}{\mu}}\begin{Bmatrix} \sqrt{| f^{*}(\alpha_1)-f^*(\alpha_2)|}\\+\\\sqrt{|f(K_2,\alpha_2)-f(K_3,\alpha_1)|}\end{Bmatrix}
        \end{align}
    \begin{align}
    \|g^{*}(a_1)-g^{*}(a_2)\|&\leq L'\|a_1-a_2\|+\sqrt{\tfrac{4}{\mu}}\begin{Bmatrix} \sqrt{(L+L\cdot L')\|a_1-a_2\| }\\+\\\sqrt{L\|K_2-K_3\|+L\|\alpha_2-\alpha_1)\|}\end{Bmatrix}\\
    &\leq L'\|a_1-a_2\|+\sqrt{\tfrac{4}{\mu}}\begin{Bmatrix} \sqrt{(L+L\cdot L')\|a_1-a_2\| }\\+\\\sqrt{L\hausdorff(g(a_1),g(a_2)) +L\|\alpha_2-\alpha_1)\|}\end{Bmatrix}\\   
        &\leq L'\|a_1-a_2\|+2\sqrt{\tfrac{4}{\mu}} \sqrt{(L+L\cdot L')\|a_1-a_2\| }\\
       &\leq \underbrace{L'+2\sqrt{\tfrac{4}{\mu}} \sqrt{(L+L\cdot L')}}_{\kappa}\max\{\|a_1-a_2\|^{1/2},\|a_1-a_2\| \}\\
       &\leq\kappa\max\{\|a_1-a_2\|^{1/2},\|a_1-a_2\| \}                
    \end{align}
Thus if the difference of the parameters is less than 1, then $g^*$ is $\kappa\cdot$ (1/2)-H\"{o}lder continuous, otherwise $\kappa$ Lipschitz.
    
\end{proof}


\clearpage
\section{Omitted Proofs of Section~\ref{sec:concave}: Inclusion \& Hardness of {\sc Concave Games} to $\PPAD$} \label{app:concaveGames}

\subsection{Inclusion of {\sc ConcaveGames} to $\PPAD$} \label{sec:concave:inclusion}

Below we show the inclusion to $\PPAD$ for $\textsc{ConcaveGames}$ while the inclusion 
for $\textsc{StronglyConcaveGames}$ follows as a special case. 

\begin{theorem} \label{thm:concave:inclusion}
     The computational problem $\textsc{ConcaveGames}$ is in $\PPAD$.
\end{theorem}

\begin{proof}
We define the function
\[ \phi(\vecx, \vecy) = \sum_{i = 1}^n u_i(\vecy_i, \vecx_{-i}) - \gamma \cdot \norm{\vecy}_2^2 \]
Observe that $\phi$ is a $2 \gamma$-strongly concave function of $\vecy$ and
that we have access to the subgradient 
$\vecg \triangleq \partial_x \phi$ of $\phi$ with respect to the vector
$\vecx$.

We are ready to proceed to our reduction from the Kakutani problem that we 
present in the previous section. Our point-to-set map is the following:
\[ F(\vecx) = \{ \vecy \in S_\eta \mid \phi(\vecx, \vecy) \ge \max_{\vecy \in S_{-\eta}} \phi(\vecx, \vecy) - \eps \} \]
Below we will prove that $F$ is approximately Hausdorff (1/2)-H\"{o}lder Lipschitz. As we notice in the end of the previous section, Hausdorff Lipschitzness can be substituted by any other metric notion under polynomial rescaling of the precision parameters. Additionally, thanks to the generality of our Computational Kakutani's inclusion lemmas (See Lemma~\ref{lem:our-berge-for-projections}) we can incorporate systematic small computational errors in our Lipschitz condition (Check the role of $\eps^\circ$ in Lemma~\ref{lem:our-berge-for-projections} )

 Firstly notice that using the fact that $u_i$'s are
$G_i$-Lipschitz functions $\phi$ is also $G = \sum_i G_i + 2 \gamma d$ Lipschitz. Thanks to our Lipschitz version of the maximum theorem in Theorem \ref{thm:our-berge-for-optimization}
if we set $\kappa_{G,\gamma}=\Big(2\sqrt{\tfrac{4}{2\gamma}G} \Big)$, then the following two maps are $\kappa$ - (1/2) H\"{o}lder continuous:
\[
\begin{cases}
H_{+}(\vecx)=\{ \vecy \in S_\eta \mid \phi(\vecx, \vecy) = \max_{\vecy \in S_{\eta}} \phi(\vecx, \vecy)\}\\
H_{-}(\vecx)=\{ \vecy \in S_{-\eta} \mid \phi(\vecx, \vecy) = \max_{\vecy \in S_{-\eta}} \phi(\vecx, \vecy) \}
\end{cases} 
\]

Indeed, it is easy to see that we just apply our Robust Berge's Maximum theorem for the constant correspondences  $g(\veca)=S_\eta$ and $g(\veca)=S_{-\eta}$ for the parametrized objective function $f(\veca,\vecb)=\phi(\veca,\vecb)$.
Thanks again to Theorem~\ref{thm:our-berge-for-optimization}, we are able to prove the following Lipschitz lemma for the maximum value of function in different dilations of a constrained set $S$:
\begin{lemma}\label{lem:max-dilation-lipschitz}
For a function $f:A\subseteq\R^d\to\R$, which is $\mu$-strongly concave and $G$-lipschitz and a well-bounded convex set $S$, i.e.  $\exists \veca_0\in \R^d\ : \parallelbody{\veca_0}{r}\subseteq S
\subseteq\parallelbody{0}{R}\subseteq A$, it holds that:
\[
|\max_{\veca\in S_{\eta_1}} f(\veca)-\max_{\veca\in S_{\eta_2}} f(\veca)|\leq C_{G,\mu}\| \eta_1-\eta_2 \|
\]
for some constant $C_{G,\mu}$.
\end{lemma}
\begin{proof}
    Notice that it suffices to prove that the correspondence $g(\eta)=S_{\eta}$ for any $\eta\in\R$ is Hausdorff Lipschitz. We split the proof in three cases:
    \begin{enumerate}
        \item Let $\eta_1\ge \eta_0 \ge 0$.
        It is easy to check that $\hausdorff(S_{\eta_1},S_{\eta_0})=\max_{\vecx\in S_{\eta_1}}\metric(\vecx,S_{\eta_0})$.
        However, $\max_{\vecx\in S_{\eta_1}}\metric(\vecx,S_{\eta_0})
        =\max_{\vecx\in S_{\eta_1} }\|\vecx-\projectionto{(S_{\eta_0})_{(\eta_1-\eta_0)}}{\vecx}\|
        =\eta_1-\eta_0$, thanks to Lemma~\ref{lem:from-out-to-border}
        
        \item Let  $0\ge\eta_1\ge \eta_0 $.
                Again, it is easy to check that $\hausdorff(S_{\eta_1},S_{\eta_0})=\max_{\vecx\in S_{\eta_1}}\metric(\vecx,S_{\eta_0})$.
        However, $\max_{\vecx\in S_{\eta_1}}\metric(\vecx,S_{\eta_0})
        =\max_{\vecx\in S_{\eta_0}}\|\vecx-\projectionto{(S_{\eta_1})_{-(\eta_1-\eta_0)}}{\vecx}\|
        \leq \frac{R}{r}(\eta_1-\eta_0)$, thanks to Lemma~\ref{lem:from-border-to-in}
        
        \item Let $\eta_1\ge 0 \ge \eta_0$. By triangular inequality, we get that $
        \hausdorff(S_{\eta_1},S_{\eta_0})\leq (1+\frac{R}{r})|\eta_1-\eta_0|
        $
    \end{enumerate}
\end{proof}

Now, we define $H_{+}^{\epsilon}$ and $H_{-}^{\epsilon}$ maps:
\[
\begin{cases}
H_{+}^{\eps}(\vecx)=\{ \vecy \in S_\eta \mid \phi(\vecx, \vecy) \ge \max_{\vecy \in S_{\eta}} \phi(\vecx, \vecy)-\eps\}\\
H_{-}^{\eps}(\vecx)=\{ \vecy \in S_{-\eta} \mid \phi(\vecx, \vecy) \ge \max_{\vecy \in S_{-\eta}} \phi(\vecx, \vecy) -\eps\}
\end{cases} 
\]
By optimality KKT conditions for maximization of a concave function in a constraint set $C$, we have that \[ \partial \phi(\vecx,\vecy^*)^\top (\vecy^*-\vecy) \ge 0 \ \ \  \forall \vecy\in C \text{ and  }\vecy^*=\argmax_{\vecy\in C} \phi(\vecx,\vecy)\] and by $(2\gamma)-$strong-concavity of $\phi(\vecx,\cdot)$ we have that
\[
\phi(\vecx,\vecy^*)-\phi(\vecx,\vecy)\ge \partial \phi(\vecx,\vecy^*)^\top (\vecy^*-\vecy)+\gamma\lVert \vecy-\vecy^* \rVert^2\ge \gamma\lVert \vecy-\vecy^* \rVert^2
\]

Thus, for any $\vecy\in H_+^{\eps}(\vecx)$ it holds that $\|\vecy-\argmax_{\vecy\in S_\eta} \phi(\vecx,\vecy)\|\leq \sqrt{\frac{\eps}{\gamma}}$ or equivalently \[\hausdorff(H_+^{\eps}(\vecx),H_+(\vecx))\leq \sqrt{\frac{\eps}{\gamma}}.\]
Applying Lemma~\ref{lem:max-dilation-lipschitz}, we have that 
for every $\vecy \in F(\vecx)$ it holds that \[\phi(\vecx, \vecy) \ge \max_{\vecy \in S_{-\eta}} \phi(\vecx, \vecy) - \eps \geq \max_{\vecy \in S_{\eta}} \phi(\vecx, \vecy) - \eps -2(1+\tfrac{R}{r})\eta
\]
Therefore for every $\vecy \in F(\vecx)$, it holds that $\vecy\in H_{+}^{\eps+2(1+\tfrac{R}{r})\eta}(\vecx)$. Consequently,
\[
\begin{cases}
\hausdorff(F(\vecx_1),H_{+}(\vecx_1))\leq \sqrt{\frac{\eps+2(1+\tfrac{R}{r})\eta}{\gamma}}\\
\hausdorff(F(\vecx_2),H_{+}(\vecx_2))\leq \sqrt{\frac{\eps+2(1+\tfrac{R}{r})\eta}{\gamma}}\\
\hausdorff(H_{+}(\vecx_1),H_{+}(\vecx_2))=
\metric(H_{+}(\vecx_1),H_{+}(\vecx_2))
\leq \kappa_{G,\gamma}\|\vecx_1-\vecx_2\|^{1/2}
\end{cases}
\]
Hence, we showed that $F$ is approximate Hausdorff-(1/2)H\"{o}lder and more concretely:
\[
\hausdorff(F(\vecx_2),F(\vecx_1))\leq 
\kappa_{G,\gamma}\|\vecx_1-\vecx_2\|^{1/2}+ 2\sqrt{\tfrac{\eps+2(1+\tfrac{R}{r})\eta}{\gamma}}
\]

Additionally, in order to employ a reduction via Computational Kakutani's Problem version, we need to prove that the correspondence $F(\vecx)$ contains a ball of lower-bounded radius. For this purpose, let $\vecy_{_{( \vecx,-\eta )}}^{\star }$ denote  $\displaystyle\argmax_{\vecy\in S_{-\eta }} \phi ( \vecx,\vecy)$. By construction we have that: 

\begin{enumerate}
    \item $\exists \ \hat{\vecv}: \|\hat{\vecv}\|=1\ \& \ \mathcal{B}:=  \parallelbody{\vecy_{_{( \vecx,-\eta )}}^{\star }-\tfrac{\eta}{2}\hat{\vecv}}{\min\{\eta/2,\eps/G\}}
\subseteq S_{-\eta}$
    \item By Lipschitzness of $\phi(\vecx,\cdot)$, we have that
    $\vecy\in\mathcal{B}: |\phi(\vecx,\vecy)-\phi(\vecx,\vecy_{_{( \vecx,-\eta )}}^{\star })|\leq G\min\{\eta/2,\eps/G\}$
\end{enumerate}
\noindent
From (2)    $
    \phi(\vecx,\mathcal{B})
    \subseteq\left[
    \phi(\vecx,\vecy_{_{( \vecx,-\eta )}}^{\star })-\eps,
    \phi(\vecx,\vecy_{_{( \vecx,-\eta )}}^{\star })\right]\overset{(1)}{\Rightarrow} \mathcal{B}\subseteq F(\vecx)
    $.
Therefore, $F(\vecx)$ always contains a ball of radius $\min\{\eta/2,\eps/G\}$.
\begin{figure}[h!]
\centering

\tikzset{every picture/.style={line width=0.75pt}} 

\begin{tikzpicture}[x=0.75pt,y=0.75pt,yscale=-1,xscale=1]

\draw  [fill={rgb, 255:red, 245; green, 166; blue, 35 }  ,fill opacity=0.23 ] (99.57,80.58) .. controls (90.99,50.67) and (112.82,26.43) .. (148.34,26.43) .. controls (183.86,26.43) and (219.62,50.67) .. (228.2,80.58) .. controls (236.79,110.49) and (214.96,134.73) .. (179.44,134.73) .. controls (143.92,134.73) and (108.16,110.49) .. (99.57,80.58) -- cycle ;
\draw  [fill={rgb, 255:red, 208; green, 2; blue, 27 }  ,fill opacity=0.27 ] (83.19,87.79) .. controls (70.06,42.06) and (97.37,5) .. (144.18,5) .. controls (190.99,5) and (239.58,42.06) .. (252.71,87.79) .. controls (265.84,133.51) and (238.53,170.57) .. (191.72,170.57) .. controls (144.91,170.57) and (96.32,133.51) .. (83.19,87.79) -- cycle ;
\draw  [fill={rgb, 255:red, 245; green, 166; blue, 35 }  ,fill opacity=0.65 ] (105.56,90.94) .. controls (105.56,85.91) and (109.64,81.83) .. (114.67,81.83) .. controls (119.7,81.83) and (123.78,85.91) .. (123.78,90.94) .. controls (123.78,95.97) and (119.7,100.05) .. (114.67,100.05) .. controls (109.64,100.05) and (105.56,95.97) .. (105.56,90.94) -- cycle ;
\draw  [fill={rgb, 255:red, 80; green, 227; blue, 194 }  ,fill opacity=0.39 ] (87.81,96.42) .. controls (87.81,85.55) and (96.62,76.74) .. (107.49,76.74) .. controls (118.36,76.74) and (127.17,85.55) .. (127.17,96.42) .. controls (127.17,107.29) and (118.36,116.1) .. (107.49,116.1) .. controls (96.62,116.1) and (87.81,107.29) .. (87.81,96.42) -- cycle ;
\draw    (99.87,103.4) .. controls (106.98,118.44) and (280.69,123.64) .. (309,117) ;
\draw    (114.67,90.94) .. controls (121.79,105.97) and (295,95) .. (310,84) ;
\draw    (108.98,103.4) .. controls (125,146.07) and (262.5,147.71) .. (268,163) ;
\draw    (319,22) .. controls (205,24) and (81.65,84.47) .. (99.57,80.58) ;
\draw    (121.85,85.45) -- (107.49,96.42) ;
\draw  [fill={rgb, 255:red, 245; green, 166; blue, 35 }  ,fill opacity=0.65 ] (90.76,103.4) .. controls (90.76,98.37) and (94.84,94.29) .. (99.87,94.29) .. controls (104.9,94.29) and (108.98,98.37) .. (108.98,103.4) .. controls (108.98,108.44) and (104.9,112.51) .. (99.87,112.51) .. controls (94.84,112.51) and (90.76,108.44) .. (90.76,103.4) -- cycle ;
\draw    (107.05,97.92) -- (92.69,108.89) ;

\draw (222.01,127.59) node [anchor=north west][inner sep=0.75pt]  [font=\large]  {$ \begin{array}{l}
S\\
\end{array}$};
\draw (179.31,50.86) node [anchor=north west][inner sep=0.75pt]  [font=\large]  {$ \begin{array}{l}
S_{-\eta }{}\\
\end{array}$};
\draw (96.48,85.35) node [anchor=north west][inner sep=0.75pt]  [color={rgb, 255:red, 208; green, 2; blue, 27 }  ,opacity=1 ]  {$ \begin{array}{l}
\star \\
\end{array}$};
\draw  [fill={rgb, 255:red, 80; green, 227; blue, 194 }  ,fill opacity=0.24 ]  (318.83,18.5) -- (381.83,18.5) -- (381.83,46.5) -- (318.83,46.5) -- cycle  ;
\draw (321.83,22.9) node [anchor=north west][inner sep=0.75pt]  [font=\scriptsize]  {$\overline{B}\left( \vec\vecy_{_{( \vecx,-\eta )}}^{\star } ,\eta \right)$};
\draw  [fill={rgb, 255:red, 208; green, 2; blue, 27 }  ,fill opacity=0.22 ]  (268.83,161) .. controls (268.83,154.65) and (277.33,149.5) .. (287.83,149.5) -- (393.83,149.5) .. controls (404.32,149.5) and (412.83,154.65) .. (412.83,161) .. controls (412.83,167.35) and (404.32,172.5) .. (393.83,172.5) -- (287.83,172.5) .. controls (277.33,172.5) and (268.83,167.35) .. (268.83,161) -- cycle  ;
\draw (271.83,153.9) node [anchor=north west][inner sep=0.75pt]  [font=\scriptsize]  {$\vec\vecy_{_{( \vecx,-\eta )}}^{\star } =\arg\max_{y\in S_{-\eta }} \phi ( x,y)$};
\draw  [fill={rgb, 255:red, 245; green, 166; blue, 35 }  ,fill opacity=0.44 ]  (306.83,76.5) .. controls (306.83,68.77) and (318.92,62.5) .. (333.83,62.5) -- (375.83,62.5) .. controls (390.74,62.5) and (402.83,68.77) .. (402.83,76.5) .. controls (402.83,84.23) and (390.74,90.5) .. (375.83,90.5) -- (333.83,90.5) .. controls (318.92,90.5) and (306.83,84.23) .. (306.83,76.5) -- cycle  ;
\draw (309.83,66.9) node [anchor=north west][inner sep=0.75pt]  [font=\scriptsize]  {$\overline{B}\left( \vecy_{_{( \vecx,-\eta )}} -\tfrac{\eta }{2}\hat{\vecv} ,\tfrac{\eta }{2}\right) \ \ \ $};
\draw  [fill={rgb, 255:red, 245; green, 166; blue, 35 }  ,fill opacity=0.44 ]  (309.83,116.5) .. controls (309.83,108.77) and (321.92,102.5) .. (336.83,102.5) -- (379.83,102.5) .. controls (394.74,102.5) and (406.83,108.77) .. (406.83,116.5) .. controls (406.83,124.23) and (394.74,130.5) .. (379.83,130.5) -- (336.83,130.5) .. controls (321.92,130.5) and (309.83,124.23) .. (309.83,116.5) -- cycle  ;
\draw (312.83,106.9) node [anchor=north west][inner sep=0.75pt]  [font=\scriptsize]  {$\overline{B}\left( \vecy_{_{( \vecx,-\eta )}} +\tfrac{\eta }{2}\hat{\vecv} ,\tfrac{\eta }{2}\right) \ \ \ $};

\end{tikzpicture}

\end{figure}

Additionally, using the
sub-gradient oracles\footnote{ For simplicity, we can assume that we have access to exact subgradients for rational inputs. Our results still holds for approximate subgradients using techniques of \cite{lee2015faster}. } to $u_i$ and the corresponding separation oracle of constraint set $S$,  we can construct a weak separation oracle for 
$F(\vecx)$. 

Indeed, recalling the framework of \ProbWCCO, we can compute a solution $\vecy_{sol}\in S_{\min\{\eps,\eta\}}\subseteq S_{\eta}$ using subgradient ellipsoid central cut method  such that  $\phi(\vecx,\vecy_{sol})\ge \max_{\vecy\in S_{-\min\{\eps,\eta\}}} \phi(\vecx,\vecy)-\min\{\eps,\eta\}\ge 
 \max_{\vecy\in S_{-\eta}} \phi(\vecx,\vecy)-\eps 
$. Equipped with that value, it suffices to substitute a WSO for $F(\vecx)$ with an separation oracle for
$\tilde{F}_{\vecy_{sol}}(\vecx)=\{ \vecy \in S_\eta \mid -\phi(\vecx, \vecy) \le \gamma'=-\phi(\vecx,\vecy_{sol}) \}$, where $-\phi(\vecx,\cdot)$ is a convex function.
Thus, using machinery similar with Algorithm~1 of the previous section, such a weak separation oracle is possible in polynomial time.

Therefore, we can give $F$ as input to the Kakutani problem that 
we presented in the previous section with accuracy parameter $\alpha=\eps/G$. The output of this Kakutani instance
will be a point $\vecx\in S_\eta$ such that $\norm{\vecx - \vecz} \le \eps/G$ for some
$\vecz \in F(\vecx)$. Now because $\phi$ is $G$-Lipschitz we have that
$\phi(\vecx, \vecx) \ge \max_{\vecy \in S_{-\eta}} \phi(\vecx, \vecy) - 2 \cdot \eps$.
The final thing that we should be careful with is to set $\gamma$ small 
enough. Indeed, if we set $\gamma \le \eps/d$ then we get that for every player $i$, similarly with \citet{rosen1965existence} argumentation,
$u_i(\vecx) \ge \max_{(\vecy_i, \vecx_{-i}) \in S_{-\eta}} u_i(\vecy_i, \vecx_{-i}) - 3 \cdot \eps$ \textendash since variable $\vecy_i$ appears only to $u_i(\cdot)$ component in the sum of $\phi(\vec,\cdot)$. 
Therefore $\vecx$ is a $(3 \eps,\eta)$-approximate equilibrium for the concave games
problem and the lemma follows.
\end{proof}

It is easy to verify from the proof that the existence of a strong separation oracle eliminates fully the dependency with $\eta$, providing the following result:

\begin{lemma} \label{lem:concave-SO:inclusion}
     The computational problem $\textsc{ConcaveGames with SO}$ is in $\PPAD$.
\end{lemma}
{\small
  \begin{nproblem}[\textsc{ConcaveGames with SO}]
  Similarly with \textsc{ConcaveGames}.
  
    \noindent \textsc{Output:} 
    \begin{enumerate}
      \item[1'.]  An $\eps$-approximate equilibrium as per Definition
                 \ref{def:concave:approximateEquilibrium}.
    \end{enumerate}
  \end{nproblem}
}

\subsection{Hardness of {\sc ConcaveGames} in $\PPAD$} \label{sec:concave:hardness}

  In this section we show that the easier problem 
$\textsc{StronglyConcaveGames}$ is $\PPAD$-hard even when the utility
functions are given as explicit polynomials of constant degree.
\smallskip

  The starting $\PPAD$-hard problem that we can use is the following 
problem from \cite{filos2021complexity} expressed in a more suitable 
way for our purposes. First we define the functions 
$T : \R \to [0, 1]$, $G_1 : [0, 1]^2 \to [0, 1]$, 
$G_- : [0, 1]^2 \to [0, 1]$ such that $T(x) = \max\{\min\{x, 1\}, 0\}$,
$G_1(x, y) = 1$ and $G_-(x, y) = T(x - y)$.
\smallskip

\newcommand{\gcircuit}{\textsc{gCircuit}}
{\small \begin{nproblem}[\gcircuit]\label{d:gcircuit}
    \noindent \textsc{Input:} A sequence $\vect \in \{``1", ~ ``-"\}^n$ and two 
    sequences of indices $\vecp \in [n]^n$, $\vecq \in [n]^n$
    that provide a coordinate-wise description of a function 
    $M : [0, 1]^n \to [0, 1]^n$, such that 
    $M_i(\vecx) = G_{t_i}(x_{p_i}, x_{q_i})$.
    \smallskip

    \noindent \textsc{Output:} A point $\vecx \in [0, 1]^n$ such that
    $\norm{\vecx - M(\vecx)}_{\infty} \le c$ where $c$ is a universal
    constant determined from Proposition 5.3 of
    \cite{filos2021complexity}.
  \end{nproblem}
 } 
\begin{theorem}[Proposition 5.3 of \cite{filos2021complexity}]
  The problem $\gcircuit$ is $\PPAD$-complete.
\end{theorem}

\noindent Next we will reduce $\gcircuit$ to $\textsc{StronglyConcaveGames}$. 
\begin{theorem}\label{thm:strong-concave-ppad-complete}
    The problem $\textsc{StronglyConcaveGames}$ is $\PPAD$-complete even in the
  two-player setting, i.e., when $n = 2$, and even when for $2$-strongly convex
  objectives, constant degree polynomials and constant required accuracy $\eps$.
\end{theorem}

\begin{proof}
    To do that we need to approximate the gates $G_1$ and $G_-$ up to
  error $c$ with polynomials of constant degree, that depends on $c$.
  Obviously, the gate $G_1$ is the constant polynomial that is equal to
  $1$. So it remains to find a polynomial that approximates $G_-$. 
  Since $G_-(x, y) = T(x - y)$ then it suffices to find a polynomial 
  that approximates the function $T$ up to $c$ error. In particular, we
  want a polynomial $p \in \R[z]$ that satisfies the following
  conditions 
  \begin{enumerate}
    \item for all $z \in [-1, 1]$ it holds that $p(z) \in [0, 1]$, and
    \item for all $z \in [-1, 1]$ it holds that $\abs{T(z) - p(z)} \le c/2$.
  \end{enumerate}
  For this we follow the proof of Weierstrass's Approximation Theorem
  from \cite{rudin1976principles} 
  (see Theorem 7.26 in \cite{rudin1976principles}).
  We define $Q_k(x) = a_k \p{1 - x^2}^k$ where we 
  pick $a_k$ such that
  \[ a_k = \p{\int_{-1}^1 \p{1 - x^2}^k ~ d \tau}^{-1}. \]
  Observe that $a_k$ is a rational number that can be computed in time
  $\poly(2^k)$. Also from equation (49) of \cite{rudin1976principles}
  we have that $a_k \le \sqrt{k}$. We now define
  \begin{align*} 
    r(x) & \triangleq \int_{-1}^1 T(x + \tau) \cdot Q_k(\tau) ~ d \tau \\
    & = x \cdot \int_{-x}^{1 - x} Q_k(\tau) ~ d \tau + \int_{-x}^{1 - x} \tau \cdot Q_k(\tau) ~ d \tau + \int_{1 - x}^1 Q_k(\tau) ~ d \tau.
  \end{align*}
  From the last expression it is clear that $r(x)$ is a polynomial of
  degree $2 \cdot k + 2$ which again can be efficiently computed in 
  time $\poly(2^k)$. Next, we have that
  \begin{align*} 
    \abs{r(x) - T(x)} & = \abs{\int_{-1}^1 T(x + \tau) \cdot Q_k(\tau) ~ d \tau - \int_{-1}^1 T(x) \cdot Q_k(\tau) ~ d \tau} \\
    & \le \int_{-1}^1 \abs{T(x + \tau) - T(x)} \cdot Q_k(\tau) ~ d \tau \\
    & \le \int_{-1}^{-\eps/2} \abs{T(x + \tau) - T(x)} \cdot Q_k(\tau) ~ d \tau + \int_{-\eps/2}^{\eps/2} \abs{T(x + \tau) - T(x)} \cdot Q_k(\tau) ~ d \tau \\
    & ~~~~~~~~ + \int_{\eps/2}^{1} \abs{T(x + \tau) - T(x)} \cdot Q_k(\tau) ~ d \tau \\
    & \le 2 \cdot a_k \cdot \p{1 - \p{\frac{\eps}{2}}^2}^k + \frac{\eps}{2} \\
    & \le 2 \cdot \sqrt{k} \cdot \p{1 - \p{\frac{\eps}{2}}^2}^k + \frac{\eps}{2}
  \end{align*}
  hence if we set 
  $k = \log^2\p{2/\eps}/\log\p{1 - \p{\frac{\eps}{2}}^2}$ we have that
  for every $x \in [-1, 1]$ it holds that $\abs{r(x) - T(x)} \le \eps$.
  From the above we have that the polynomial $r(x)$ has constant degree and
  approximates well enough the function $T(x)$. To construct $p(x)$ it remains 
  to make sure that it always takes values in the interval $[0, 1]$. Since 
  $T(x)$ takes values in $[0, 1]$ we have that $r(x)$ takes values in the 
  interval $[-\eps, 1 + \eps]$. So we can define
  \[ p(x) \triangleq \frac{r(x) + \eps}{1 + 2 \eps}. \]
  We then have for every $x \in [-1, 1]$ and assuming that $\eps \le 1$
  \begin{align*}
    \abs{p(x) - r(x)} & \le \abs{r(x)} \cdot \abs{1 - \frac{1}{1 + 2 \eps}} + \frac{\eps}{1 + 2 \eps} \\
    & \le \frac{5 \eps}{1 + 2 \eps} \le 5 \eps.
  \end{align*}
  Therefore if we apply triangle inequality then we have that for every 
  $x \in [-1, 1]$
  \[ \abs{p(x) - T(x)} \le 6 \eps \]
  so if we pick $\eps \le c / 12$ then the polynomial $p$ is a constant degree
  polynomial that satisfies both of the properties 1. and 2. from above. This 
  means that for every $x, y \in [0, 1]$ it holds that
  \[ \abs{p(x - y) - G_-(x, y)} \le c / 2 \]
  and $p(x - y)$ is a bi-variate polynomial of constant degree that can 
  efficiently be expressed as a some of monomials.
  
  No we are ready to express the two-player concave game that we need to 
  complete the reduction. Let $M$ be the input function of the $\gcircuit$ 
  problem. Then we define $\tilde{M}$ to be the same function as $M$ but where 
  we have replaces all the function $G_-(x, y)$ with $p(x - y)$ and we are 
  asking for an $c/2$-approximate fixed point. Both the first player will 
  control a vector $\vecx_1 \in [0, 1]^n$ and the second player a vector 
  $\vecx_2 \in [0, 1]^n$. The utility function of the first player now is
  \[ u_1(\vecx_1, \vecx_2) = 2 - \norm{\vecx_1 - \tilde{M}(\vecx_2)}_2^2 \]
  and the utility function of the second player is
  \[ u_2(\vecx_1, \vecx_2) = 2 - \norm{\vecx_2 - \vecx_1}_2^2. \]
  From these definitions it is easy to observe that both of $u_1$ and $u_2$ are
  constant degree polynomials that can be efficiently expressed as a sum of 
  monomials, and that $u_i$ is a $2$-strongly concave function of $\vecx_i$ for
  $i = 1, 2$. Finally, it is also clear that according to the Definition 
  \ref{def:concave:approximateEquilibrium} every $c^2/16$-approximate 
  equilibrium of this strongly concave game satisfies
  \[ \norm{\vecx_1 - \tilde{M}(\vecx_2)}_2^2 \le c^2/16 ~~~ \text{ and } ~~~ \norm{\vecx_2 - \vecx_1}_2^2 \le c^2/16 \]
  which in turn implies that 
  \[ \norm{\vecx_2 - \tilde{M}(\vecx_2)}_{\infty} \le \norm{\vecx_2 - \tilde{M}(\vecx_2)}_2 \le \frac{c}{2}. \]
  Therefore, from the construction of $\tilde{M}$ and in particular the 
  approximation properties of the polynomial $p$ we have that 
  $\norm{\vecx_2 - M(\vecx_2)}_{\infty} \le c$ and hence $\vecx_2$ is a solution
  to the initial $\gcircuit$ instance.
\end{proof}

\clearpage
\section{Omitted Proofs of Section~\ref{sec:Walras}: Inclusion of {\sc Walrasian} to $\PPAD$} \label{sec:markets:inclusion}
\begin{proof}
[Proof of Theorem~\ref{thm:warlasian:inclusion}]
Recall that what matters in the Walrasian model is relative prices, so we are always free to normalize one of the prices. Rather than set $\vecp_1 = 1$, however, it’s convenient to normalize the prices so that they all sum to $1$. 
In order to unify different degenerate cases, like commodities whose equilibrium price is zero, we will restrict prices domain to an inner simplex 
$
\Delta_{\xi} =\left\{\vecp\in \mathbb{R}_{ >\xi }^{d} :\sum _{i\in [ n]} \vecp_{i} =1\right\}
$ for some well suited constant $\xi\approx \poly(\epsilon)$. With this restriction we  avoid multiple technicalities which are  typically introduced in the topological proofs (See \cite{levin2006general}) to make 
the budget correspondence $\calB_i(\vecp)$  compact when prices are on the boundary of $\Delta_0$.
Then, we  define the individual Marshallian demands in such a way that they are upper semi-continuous in prices. More precisly we consider for each agent $i\in[n]$ the correspondence
\[\psi_i(\vecp)=\arg\max_{\vecc\in\calB_i(\vecp)} u_i(\vecc)\]
Using the initial Berge's Theorem (Theorem~\ref{thm:berge-generalized}) to ensure that agents' demand correspondences are upper semi-continuous, we encounter the issue of verifying the continuity of the budget correspondence $\calB_i(\vecp)$. Below, we will provide a stronger result proving that $\calB_i(\vecp)$ is Hausdorff Lipschitz set-valued map. 
We start with a necessary preliminary result for the boundness of $\calB_i(\vecp)$ for $\vecp\in\Delta_\xi$.
\begin{lemma}[Boundness of $\calB_i(\vecp)$]\label{lem:boundness:Bis}
If $\vecc\in\calB_i(\vecp)$ then $\vecc\in\parallelbody{\mathbf{0}}{d\|\vece_i\|/\xi}$
\end{lemma}
\begin{proof}
In order to bound $\displaystyle\max_{\vecx_i^{(\vecp)}\in \calB_i(\vecp)} \|\vecx_i^{(\vecp)}\|$, we notice that by construction it holds that:
\[
\begin{cases}
 \vecx_{i,k}\geq 0\  \forall k\in [d]\\
\vecp\cdot \vecx_i\leq \vecp\cdot \vece_i
\end{cases}\Rightarrow
\begin{cases}
 \vecx_{i,k}\geq 0\  \forall k\in [d]\\
\xi\sum_{k\in[d]} \vecx_{i,k}\leq\vecp\cdot \vecx_i\leq \vecp\cdot \vece_i\leq \sqrt{d}\|\vece_i\|
\end{cases}
\]
Therefore, we have that
$
\left
\{
0\leq \vecx_{i,k}\leq  \tfrac{\sqrt{d}}{\xi}\|\vece_i\| \ \ \forall k\in [d]
\right\}
$ which yields $\|\vecx_i\|\leq \frac{d}{\xi}\|\vece_i\|$
\end{proof}
In order to show such result we will exploit an important bound in parametric optimization of Linear programs with moving polytopes:

\begin{lemma}[Hoffman Bound (\cite{pena2021new})]
Let $A\in \R^{m\times n}$ and $LP(\vect)=\{\vecx\ |\ A\vecx\leq \vect\}$, 
for $\vect\in\R^m$. 
Then, there exists some constant $L_0>0$ such that for each $\vecx\in\R^n$ and $\vect\in\R^m$ with $LP(\vect)\neq\emptyset$ the following holds: 
There exists $\vecx_{\vect}\in LP(\vect)$ satisfying :
\[\|\vecx-\vecx_{\vect}\|\leq \calH_0(A)\max_{1\leq j\leq m}(A_j^\top \vecx-\vect_j,0)\]
 where $A_j$ is the $j-$the row of $A$ and $\calH_0(A)=(\min_{\vecv\in\R_+^m:\|\vecv\|=1}\|A^\top\vecv\|)^{-1}$.
\end{lemma}
\noindent
With that being said, we are prepared to demonstrate that $\calB_i(\vecp)$ is Hausdorff Lipschitz:
\begin{lemma} For any $i\in[n]$, it holds that 
  $\calB_i(\vecp)$ is $(d^{3/2}/\xi^2)\|\vece_i\|$-Hausdorff Lipschitz correspondence for $\vecp\in \Delta_\xi$.
\end{lemma}
\begin{proof}
Let $\vecp_1,\vecp_2$ be two arbitrary price vectors in $\Delta_\xi$ and $i$ some arbitrary agent in $[n]$. Additionally, let any $\vecx_i^{(\vecp_1)}\in \calB_i(\vecp_1)$. Then, by definition we have that $\vecp_1\cdot(\vecx_i^{(\vecp_1)}-\vece_i)\leq 0$. At the same time it holds:
\[
|\vecp_1\cdot \vece_i- \vecp_2\cdot \vece_i |\leq \|\vece_i\|\cdot\|\vecp_1-\vecp_2\|
\]
Thus we have that: 
\begin{align*}
\vecp_2\cdot (\vecx_i^{(\vecp_1)}-\vece_i)&\leq \vecp_2\cdot (\vecx_i^{(\vecp_1)}-\vece_i) - \left[\vecp_1\cdot(\vecx_i^{(\vecp_1)}-\vece_i)\right]\\
&\leq (\vecp_2-\vecp_1)\cdot \vecx_i^{(\vecp_1)}-\vece_i\cdot(\vecp_2-\vecp_1)\\
&\leq \|\vecp_2-\vecp_1\|\max_{\vecx_i^{(\vecp_1)}\in \calB_i(\vecp_1)} \|\vecx_i^{(\vecp_1)}\|+ \|\vece_i\|\cdot\|\vecp_1-\vecp_2\|\\
(\text{by Lemma~\ref{lem:boundness:Bis}})&\leq
 \|\vece_i\|\tfrac{d}{\xi}\cdot\|\vecp_1-\vecp_2\|
\end{align*}

\noindent Let us fix $A=\vecp_2$ and compute the Hoffman constant for the linear program  $LP(t)=\{\tilde{\vecx}|
A \tilde{\vecx}\leq t\}$ equals to $
 \calH_0(\vecp_2)=(\min_{\vecv\in\R_+^d:\|\vecv\|=1}\|\vecp_2^\top\vecv\|)^{-1}\leq \tfrac{\sqrt{d}}{\xi}$. Then, by Hoffman Lemma we get that there exists $\vecx_i^{(\vecp_2)}\in\calB_i(\vecp_2)$
 such that $\|\vecx_i^{(\vecp_1)}-\vecx_i^{(\vecp_2)}\|\leq \tfrac{\sqrt{d}}{\xi}\cdot\max(\vecp_2^\top\vecx_i^{(\vecp_1)}-\vecp_2\cdot \vece_i,0)$. Thus, until now we have proved that:
 \[\forall \vecx_i^{(\vecp_1)} \exists \vecx_i^{(\vecp_2)}: \|\vecx_i^{(\vecp_1)}-\vecx_i^{(\vecp_2)}\|\leq
  \|\vece_i\|\tfrac{d^{3/2}}{\xi^2}\cdot\|\vecp_1-\vecp_2\| \]
 Consequently, we have that 
 $\hausdorff(\calB_i(\vecp_1),\calB_i(\vecp_2))\leq 
  \|\vece_i\|\tfrac{d^{3/2}}{\xi^2}\|\vecp_1-\vecp_2\|$.
which concludes the proof.
\end{proof}
In order to enable our Robust Berge's Maximum theorem, as in the previous section we will add a small regularizer:
\[\widetilde{\psi_i^D}(\vecp)=\arg\max_{\vecc\in\calB_i(\vecp)} \widetilde{u_i}(\vecc)=\arg\max_{\vecc\in\calB_i(\vecp)} u_i(\vecc)-\gamma\|\vecc\|^2\]
It is easy to see that the modified utilities $\widetilde{u_i}$ are $\gamma_{\widetilde{u_i}}=2\gamma$-strongly concave and $L_{\widetilde{u_i}}=(L+\tfrac{\gamma\cdot d}{\xi}\|\vece_i\|)$ Lipschitz. Thus, leveraging our Robust Berge's Maximum principle (Theorem~\ref{thm:our-berge-for-optimization}), we get that $\widetilde{\psi_i^D}$ is $\calK_i$-(1/2) H\"{o}lder continuous for 
$\calK_i=\left(\tfrac{d^{3/2}}{\xi^2}\|\vece_i\|+2\sqrt{\frac{4}{2\gamma}\cdot(L+\tfrac{\gamma d }{\xi}\|\vece_i\|) }\sqrt{1+\tfrac{d^{3/2}}{\xi^2}\|\vece_i\|}\right)$.
Recalling the framework of $\ProbSCCO$, we can compute in polynomial time an allocation vector $\vecx_{i,sol}^{(\epsilon)}\in\calB_i(\vecp)$ using subgradient ellipsoid central cut method such that $\widetilde{u_i}(\vecx_{i,sol}^{(\epsilon)})\ge\widetilde{u_i}(\widetilde{\psi_i^D}(\vecp))-\epsilon$. Equipped with that value, we can construct a strong separation oracle for \[\widetilde{\Psi_i^D}(\vecp)=\{\vecx_i\in\calB_i(\vecp)|\widetilde{u_i}(\vecx_i)\geq \widetilde{u_i}(\vecx_{i,sol}^{(\epsilon)})\}\] Indeed for an arbitrary point $\vecx$, we can provide easily either a strong separation oracle from the polytope $\calB_i(\vecp)$, a subgradient separation oracle from the value level-set of $\widetilde{u_i}(\vecx_{i,sol}^{(\epsilon)})$ or an exact membership verification.
Using the machinery of Section 4.7 from \cite{grotschel2012geometric}, we can construct a strong separation oracle for the Minkowski sum (See Appendix~\ref{app:oracles})
 \[\widetilde{\Psi^D}(\vecp)=\sum_{i\in[n]} \widetilde{\Psi_i^D}(\vecp)=\left\{\sum_{i\in[n]} \vecx_{i}:\forall i\in[n] \ \ \vecx_{i}\in\calB_i(\vecp)\ \&\  \widetilde{u_i}(\vecx_i)\geq \widetilde{u_i}(\vecx_{i,sol}^{(\epsilon)})\right\}\]

Similarly with the concave games we can show that $\widetilde{\Psi^D}(\vecp),\widetilde{\Psi_i^D}(\vecp)$ is approximate Hausdorff-(1/2)H\"{o}lder.
Indeed, by $(2\gamma)-$strong-concavity of $\widetilde{u_i}(\vecx_i)$ for any $\vecx_i\in \widetilde{\Psi_i^D}(\vecp)$ and $\vecx_i^\star\in \widetilde{\psi_i^D}(\vecp)$ we have that
\[
\epsilon\ge\widetilde{u_i}(\vecx_i^\star)-
\widetilde{u_i}(\vecx_i)
\ge \partial \widetilde{u_i}(\vecx_i^\star)^\top (\vecx_i^\star-\vecx_i)+\gamma\lVert \vecx_i^\star-\vecx_i \rVert^2\ge \gamma\lVert \vecx_i^\star-\vecx_i \rVert^2
\]
Therefore it holds that $\hausdorff(\widetilde{\Psi_i^D}(\vecp),\widetilde{\psi_i^D}(\vecp))\leq \sqrt{\tfrac{\epsilon}{\gamma}}$.
Hence, we showed that $\widetilde{\Psi_i^D}$ is approximate Hausdorff-(1/2)H\"{o}lder and more concretely:
\begin{align*}
\hausdorff(\widetilde{\Psi_i^D}(\vecp_2),\widetilde{\Psi_i^D}(\vecp_1))
&\leq \hausdorff(\widetilde{\Psi_i^D}(\vecp_2),\widetilde{\psi_i^D}(\vecp_2))+\hausdorff(\widetilde{\psi_i^D}(\vecp_2),\widetilde{\psi_i^D}(\vecp_1))+\hausdorff(\widetilde{\Psi_i^D}(\vecp_1),\widetilde{\psi_i^D}(\vecp_1))\\
&\leq \calK_i\|\vecp_1-\vecp_2\|^{1/2}+ 2\sqrt{\tfrac{\eps}{\gamma}}
\end{align*}
while using triangular inequality we get that
\[
\hausdorff(\widetilde{\Psi^D}(\vecp_2),\widetilde{\Psi^D}(\vecp_1))
\leq \sum_{i\in[n]}\calK_i\|\vecp_1-\vecp_2\|^{1/2}+ 2n\sqrt{\tfrac{\eps}{\gamma}}
\]
\begin{remark}
Without violating any approximate Lipschitz condition, we can always restrict our correspondence using an extra  strong separation oracle in a halfspace $$Q^{\eps'}:=\{\sum_{i\in[n]}\vecx_i\ge (1-\eps')\cdot \sum \vece_i\}.$$ 
Additionally, we can always assume that the oracle $
\widetilde{\Psi^D}(\vecp) 
$ describes the complete information vector $(\vecx_1,\vecx_2,\cdots,\vecx_n,\sum_{i\in[n]}\vecx_i)$. Finally, in order to satisfy the budget constraint tightly, we will run our oracle for $\widetilde{\calE}(\alpha)=\{i\in[n]:\tilde{\vece_i}(\alpha)=\vece_i-\frac{\alpha}{n}\mathbf{1}_d\}$ or equivalently we will optimize over $\tilde{\calB}_i(\vecp)=\{\vecx_i:\vecp\cdot\vecx_i\leq \vecp\cdot\vece_i - \tfrac{\alpha}{n} \}$
In other words,
 \[\widetilde{\Psi^D}(\vecp):=\left\{
  \begin{pmatrix}\vecx_1\\\vecx_2\\\vdots\\\vecx_n\\\sum_{i\in[n]}\vecx_i
  \end{pmatrix}
 :\sum_{i\in[n]} \vecx_{i}\in Q^{\eps'} :\forall i\in[n] \ \ \vecx_{i}\in\tilde{\calB}_i(\vecp)\ \&\  \widetilde{u_i}(\vecx_i)\geq \widetilde{u_i}(\vecx_{i,sol}^{(\epsilon)})\right\}\]
\end{remark}
\noindent Then we construct the correspondence $\psi^P(\vecx):\R^d_+\rightrightarrows \Delta_\xi$:
\newcommand{\invw}{\rotatebox[origin=c]{180}{w}}
\[\widetilde{\psi^P}(\vecx)=\displaystyle\arg\max_{\vecp\in\Delta_\xi}
    \invw(\vecp)=\displaystyle\arg\max_{\vecp\in\Delta_\xi}
\vecp^\top(\vecx-\sum_{i\in[n]}\vece_i)-\gamma\|\vecp\|^2\]
It is easy to see that  $\invw(\vecp)$ is $2\gamma$-strongly concave and $(\gamma+2\tfrac{\gamma\cdot d}{\xi}\sum_{i\in[n]}\|\vece_i\|)$ Lipschitz. Thus, leveraging our Robust Berge's Maximum principle (Theorem~\ref{thm:our-berge-for-optimization}), we get that $\widetilde{\psi_i^D}$ is $\Lambda$-(1/2) H\"{o}lder continuous for
$\Lambda=\Big(\sqrt{8(1+2\tfrac{d}{\xi}\sum_{i\in[n]}\|\vece_i\|)} \Big)$.
Once again, we can compute in polynomial time via ellipsoid a price vector $\vecp_{sol}^{(\epsilon)}\in\Delta_\xi$ such that $\invw(\vecp_{sol}^{(\epsilon)})\ge\invw(\widetilde{\psi^P}(\vecx))-\epsilon$. Equipped with that value, we can construct a strong separation oracle for \[\widetilde{\Psi^P}(\vecx)=\{\vecp\in\Delta_\xi|\invw(\vecp)\geq \invw(\vecp_{sol}^{(\epsilon)}).\}\]
Indeed for an arbitrary point $\vecp$, we can provide easily either a strong separation oracle from the simplex $\Delta_\xi$, a subgradient separation oracle from the value level-set of $\invw(\vecp_{sol}^{(\epsilon)})$ or an exact membership verification. Using again the KKT conditions and the strong concavity we can have that $\widetilde{\Psi^P}$ is approximate Hausdorff-(1/2)H\"{o}lder and more concretely:
\begin{align*}
\hausdorff(\widetilde{\Psi^P}(\vecx_2),\widetilde{\Psi^P}(\vecx_1))
&\leq \hausdorff(\widetilde{\Psi^P}(\vecx_2),\widetilde{\Psi^P}(\vecx_2))+\hausdorff(\widetilde{\Psi^P}(\vecx_2),\widetilde{\Psi^P}(\vecx_1))+\hausdorff(\widetilde{\Psi^P}(\vecx_1),\widetilde{\Psi^P}(\vecx_1))\\
&\leq \Lambda\|\vecx_1-\vecx_2\|^{1/2}+ 2\sqrt{\tfrac{\eps}{\gamma}}
\end{align*}

Having constructed the aforementioned separating oracles, we are ready to reduce $\textsc{Walrasian}$ to $\ProbSKakutani$  for the concatenated correspondence  $F(\vecz=(\vecx,\vecp))=(\widetilde{\Psi^D}(\vecp),\widetilde{\Psi^P}(\vecx))$ which is $4\sqrt{\tfrac{\eps}{\gamma}}$-approximate $(\calK+\Lambda)$-Hausdorff-(1/2)H\"{o}lder smooth. It is easy to see that we can construct an strong separating oracle for $F$ since it is the intersection of $(\widetilde{\Psi^D}(\vecp),\Delta_\xi)$ and 
$(\R_+^d,\widetilde{\Psi^P}(\vecx))$ (See Appendix~\ref{app:oracles}).

Let $(\vecx^{out},\vecp^{out})$ be the $\alpha-$approximate Kakutani point.
For the output price $\vecp^{out}$, we have that $d(\vecx^{out},\widetilde{\Psi^D}(\vecp^{out}))\leq a$ or equivalently
$$ \left\| \begin{pmatrix}\vecx_1^{out},\vecx_2^{out},\cdots,\vecx_n^{out},\sum_{i\in[n]}\vecx_i^{out}
  \end{pmatrix} -   \begin{pmatrix}\vecx_1,\vecx_2,\cdots,\vecx_n,\sum_{i\in[n]}\vecx_i
  \end{pmatrix}\right\|\leq \alpha$$ where $\begin{pmatrix}\vecx_1,\vecx_2,\cdots,\vecx_n,\sum_{i\in[n]}\vecx_i
  \end{pmatrix}\in \widetilde{\Psi^D}(\vecp^{out})\in\widetilde{\Psi^D}(\vecp^{out})) $ . Thus, we get $(i) \vecp\cdot\vecx_i\leq \vecp\cdot\vece_i-\alpha/n$, $(ii) \|\vecp\|\leq 1 $, $(iii) \|\vecx_i-\vecx_i^{out}\|\leq \alpha $ which yields trivially $\vecx_i^{out}\in \calB_i(\vecp^{out})$.
Additionally by  Lipschitzness of $\widetilde{u_i}$ we have that for $\alpha\leq \epsilon/L_{\widetilde{u_i}}$, we get that
\[\widetilde{u_i}(\vecx_i^{out})\ge \max_{\vecx_i\in\widetilde{\calB_i(\vecp^{out})}}\widetilde{u_i}(\vecx_i) - 2\epsilon\]
We will use again Robust Berge's Theorem (Theorem~\ref{thm:berge-generalized})
for the moving constraint set-valued map: $$\calB_i(\vecp^{out})[\eta]=\left\{\small \vecx \in \mathbb{R}_+^{d} \ | \ \vecp\cdot\vecx\leq \vecp\cdot\vece_i\ -\eta \right\}\text{ which is } \tfrac{\sqrt{d}}{\xi}-\text{ Haussdorf Lipschitz .} $$
Thus for $\alpha\leq \epsilon/(\tfrac{\sqrt{d}}{\xi}+1)\cdot L_{\widetilde{u_i}}$, we have that
\[\widetilde{u_i}(\vecx_i^{out})\ge \max_{\vecx_i\in{\calB_i(\vecp^{out})}}\widetilde{u_i}(\vecx_i) - 3\epsilon\]
Choosing $\gamma$ such that $ \epsilon \leq  \gamma \max_{\vecx_i\in{\calB_i(\vecp^{out})}} \|\vecx_i\|^2\Rightarrow \gamma\geq \tfrac{\epsilon}{\Theta(\poly(1/\xi,d))}$,  we get that
\[{u_i}(\vecx_i^{out})\ge \max_{\vecx_i\in{\calB_i(\vecp^{out})}}{u_i}(\vecx_i) - 4\epsilon\]
By triangular inequality, we have also that $ 
\sum_{i\in[n]}\vecx_i^{out}\ge (1-2\eps)\cdot \sum \vece_i
$ for $\eps'\leq\eps$ and $\alpha\leq \eps$.
Similarly with the case of Demand player, we will employ the case of Price Player. By Lipschitzness of $\invw$, we get that
\[
    (\vecp^{out})^\top(\sum_{i\in[n]}\vecx_i^{out}-\sum_{i\in[n]}\vece_i)
\geq \max_{\vecp\in\Delta_\xi}
    \vecp^\top(\sum_{i\in[n]}\vecx_i^{out}-\sum_{i\in[n]}\vece_i)- 3\epsilon
\] 
for $\alpha\leq \epsilon/L_{\invw}$ and $\epsilon\leq \max_{\Delta_\xi} \gamma\|\vecp\|^2\Leftrightarrow \gamma\ge \epsilon$. Since $ \vecx_i\in \calB_i(\vecp^{out})$ for $i\in[n]$, we finally derive that:
\[
    3\epsilon\ge \vecp^\top(\sum_{i\in[n]}\vecx_i^{out}-\sum_{i\in[n]}\vece_i)  \ \ \forall {\vecp\in\Delta_\xi}
\]
Using vectors $\widehat{\vecp_k}=(\tfrac{\xi}{d-1},\cdots,\underset{k\text{-th coordinate}}{1-\xi},\cdots,\tfrac{\xi}{d})$ for $k\in[d]$, we can prove that 
\[
\sum_{i\in[n]}\vecx_i^{out}-\sum_{i\in[n]}\vece_i\leq 3\epsilon\mathbf{1}_{d}
\]
which conclude our proof for almost-clearance of the market.




\end{proof}

\clearpage
\section{Oracle Polynomial-Time Subgradient Ellipsoid Central Cut Method}
\label{app:ellipsoid}

In this section, for the sake of completeness, we will present a version of subgradient-cut method for convex constrained optimization when weak separation oracles are available and approximate value \& subgradient oracle for the function objective. It is worth mentioning that even the recent work of optimal combination of subgradient descent and ellipsoid method by \cite{rodomanov2022subgradient} has focused only in the strong oracle case, so this part is of independent interest.

\begin{algorithm}[!ht]
	\KwIn{Gradient and value oracles $\textrm{O}_{\mathrm{grad}}^f , \textrm{O}_{\mathrm{val}}^f$ with accuracies ($\eps_{\mathrm{grad}}, \eps_{\mathrm{val}}$ ) and weak separation oracle $\textrm{O}_{\mathrm{sep}}$ for set $\calX$ with margin $\delta$.}
	\For{$t\in[T_{\textrm{ellipsoid}}]$}{
	 \eIf{$\vecx_t\in \parallelbody{\calX}{\delta}$}
	 {Call a gradient oracle $\vecg_t\gets\textrm{O}_{\mathrm{grad}}^f(\vecx_t)$\;
 	  \eIf{$\|\vecg_t\|\leq G_{\mathrm{thres}}$}
 	  {\textbf{Output:} {$\vecx_t$}\;}
 	  {$\vecw_t\gets\vecg_t/\|\vecg_t\|_\infty$ (Output A)\;}
	 }
	 {Call a separation oracle $\vecw_t\gets \textrm{O}_{\mathrm{sep}}^f(\vecx_t)$\;
	 \textbf{if} the number of sep. oracle calls are more than $T_{\textrm{emptiness}}$ then \textbf {output} $\bot$\; }
	 Construct an ellipsoid $M_{t+1}$ such that :
	 \{$\vecx\in M_{t}:\vecw_t^\top (\vecx-\vecx_t)\leq \delta\}\subseteq M_{t+1}$\;
	 Let $\vecx_{t+1}$ be the centroid of $M_{t+1}$\;
	}
	\KwOut{The iteration $\bar{\vecx}\in \displaystyle{\argmin}\{\textrm{O}_{\mathrm{val}}^f(\vecx)|{\vecx\in \{\vecx_{1},\cdots,\vecx_{T_{\textrm{ellipsoid}}}\}\cap \parallelbody{\calX}{\delta} }\}$ (Output B)}
	\caption{Subgradient Central-Cut Ellipsoid Method}
	\label{alg:subgradient-cut-method}
\end{algorithm}
The cutting plane methods are distinguished by their construction of sets $M_t$ and selection of query points $\vecx_t$. These methods exhibit exponential decrease in the volume of $M_t$ as $t$ increases, leading to linear convergence guarantees in the presence of exact gradient and value oracles. \citeauthor{grotschel1981ellipsoid}'s Ellipsoid method provides the following guarantee, as stated in \cite[Chapter 3]{grotschel2012geometric}:

\begin{theorem}
There exists a cutting plane method, referred to as the ``central-cut Ellipsoid method'', with a decay rate of $\theta=O(1/d)$, such that:
\[\frac{Vol(M_t)}{Vol(M_1)}\leq e^{-\theta t}\]
where $Vol$ denotes the usual $d$-dimensional volume.
\end{theorem}

Let us denote $\mathcal{D}_{\epsilon}=\{\vecx\in \parallelbody{\calX}{-\delta}: \displaystyle\min_{\vecx\in \parallelbody{\calX}{-\delta}} f(\vecx) \le f(\vecx)\le
\displaystyle\min_{\vecx\in \parallelbody{\calX}{-\delta}}+\epsilon/2
\}$ be the set of all $\epsilon-$approximate and $\delta-$marginally inside $\calX$ optimal solutions for minimization task $f$.
We need to ensure that to ensure that $\mathcal{D}_{\epsilon}$ has non-zero
volume. If we assume that $\parallelbody{\calX}{-\delta}\neq \emptyset$, then by $L$-lipschitzness of $f$, we know that $\mathcal{D}_{\epsilon}$ includes a ball of radius $r(\eps,\delta)=\min\{\delta,\eps/L\}$.

Firstly, let's denote $\calT_{\textrm{active}}:=\{t\in[T_{\textrm{ellipsoid}}]|\vecx_t\in \parallelbody{\calX}{\delta}\}$.
\begin{enumerate}
    \item[Case 1:] Assume that for any $t\in \calT_{\textrm{active}}$ and for any $\vecx_\epsilon\in\calD_\epsilon$, we have that $\vecw_t^\top(\vecx_\epsilon-\vecx_t)\leq \delta$. This implies that $\forall t\in [T_{\textrm{ellipsoid}}]\ \ \forall \vecx_\epsilon\in\calD_\epsilon: \vecw_t^\top(\vecx_\epsilon-\vecx_t)\leq \delta$, since by definition of the separation oracle $(\vecw_t^\top(\vecx-\vecx_t)\leq \delta)$ for all $\vecx\in \parallelbody{\calX}{-\delta}$. Thus, it holds that 
    \[\forall t\in [T_{\textrm{ellipsoid}}]: \calD_\eps\subseteq M_t\Rightarrow vol(\calD_\epsilon)\leq vol(M_t)\]
    Next, we show that the above condition can hold only if 
    $T\leq C_0 \cdot d^2\log(d/r(\eps,\delta))$. Indeed, it holds that
    \[\begin{cases}
    \frac{Vol(M_T)}{Vol(M_1)}\leq e^{-\theta T}, \quad \theta=\Theta(\tfrac{1}{d})\\
    \frac{\pi^d}{(d/2+1)!}r(\delta,\eps)^d = Vol(\parallelbody{\calX}{r(\delta,\eps)})\\
    Vol(\parallelbody{\calX}{r(\delta,\eps)})
    \leq Vol(\calD_\eps)\leq Vol(M_T) \\
    Vol(M_1)\leq Vol(\calB{ox}) \\
    \end{cases}\implies T\leq C_0 \cdot d^2(\log(\tfrac{d}{2r(\eps,\delta)})
    \]
    for some positive constant $C_0$ independent of $d,\eps,\delta$. 
    If the number of used sep. oracles is greater than $T_{\textrm{emptiness}}$, then $
     Vol(\calX) \leq Vol(\parallelbody{\mathbf{0}}{\delta})$, or consequently 
     $\parallelbody{\calX}{-\delta}=\emptyset$.
     Otherwise, if we set $T_{\textrm{ellipsoid}}=\max\{C_0,10\}d^2(\log(\tfrac{d}{2r(\eps,\delta)})$, then for  $C_0 \cdot d^2(\log(\tfrac{d}{2r(\eps,\delta)})< T\leq T_{\textrm{ellipsoid}}$, either Case 2 or 3 hold.
    \item[Case 2:] If $\|\vecg_t\|\leq G_{\textrm{thres}}$ for appropriate choice of $G_{\textrm{thres}}$, we will show that $\vecx_t$ is an $\epsilon$-approximate minimizer.
    Indeed, by convexity $\min_{\vecx \in \parallelbody{\calX}{-\delta}} f(\vecx)\geq f(\vecx_t) + \min_{\vecx \in \parallelbody{\calX}{-\delta}}  \partial f(\vecx_t)^\top(\vecx-\vecx_t)$.
    By choosing $G_{\textrm{thres}}=
    O(\poly(d,\eps,\eps_{\textrm{grad}}))$ such that $\epsilon\geq (G_{\textrm{thres}}-\eps_{\textrm{grad}})\sqrt{d}$, we get that 
    $f(\vecx_t)\leq \min_{\vecx \in \parallelbody{\calX}{-\delta}} f(\vecx)+\epsilon $.
    \item[Case 3:] Assume then that there exists an element $\vecx_{\epsilon}$ iteration $t^\star\in [T_{\textrm{ellipsoid}}]$ such that $\vecw_{t^\star}(\vecx_{\epsilon}-\vecx_{t^\star})>\delta$.
    In this case, using the convexity of objective $f(\vecx_{t^\star})\leq f(\vecx_\epsilon)-\nabla f(\vecx_{t^\star})^\top (\vecx_\epsilon-\vecx_{t^\star})= 
    f(\vecx_\epsilon)-(\nabla f(\vecx_{t^\star})-\vecg_{t^\star})^\top (\vecx_\epsilon-\vecx_{t^\star}) - \vecg_{t^\star}^\top (\vecx_\epsilon-\vecx_{t^\star})\leq 
    f(\vecx_\epsilon)+\eps_{\textrm{grad}}\sqrt{d} -\delta
    $. If we set $\eps_{\textrm{grad}}\leq \tfrac{\eps}{\sqrt{d}}$ and $\delta\leq \tfrac{\eps}{2}$, then 
    $f(\vecx_{t^\star})\leq f(\vecx_\epsilon) + \tfrac{\eps}{2} \leq \displaystyle\min_{\vecx\in \parallelbody{\calX}{-\delta}}f(\vecx)+\epsilon $
\end{enumerate}
which conclude the proof of optimization guarantee for the problem of $\ProbWCCO$ at Theorem~\ref{thm:approx-minimization}

\clearpage
\section{Oracle Reductions}
\label{app:oracles}
The following proofs are for the sake of completeness and it can be specialized for strong oracle case as well. We invite the interested reader to see the corresponding chapters in classic book of \citeauthor{grotschel2012geometric} for Optimization with Separation and Membership oracles.

The above theorem proves that there exists an oracle-polynomial time algorithm that solves the weak optimization problem for every convex well-circumscribed body given by a weak separation oracle.

Below we prove the following inverse reduction:

\begin{theorem}\label{theorem:OPT2SEP}
There exist oracle-polynomial time algorithm that solve the following problem:
\begin{enumerate}
\item Construction of a Weak Separation oracle for a polynomially bounded convex body $\calK$, given a weak optimization algorithm.
\end{enumerate}
\end{theorem}
\begin{proof}
Let $\calK\subseteq \mathbb{R}^d$ be a convex body where we can solve the weak optimization problem with algo $\calA_{\calK}(f,\epsilon)$ and let the input of the under construction separation oracle $y\in \mathbb{Q}^{d}, \delta>0$.

We first call $\calA_{\calK}$ algorithm $2d$ times with input $\epsilon=\delta/2$ and $f(\vecx):=\pm \vece_i^\top \vecx $ for $i\in[d]$, where $\vece_i$ is the indicator vector of $i$-th coordinate. If $\calA_{\calK}$ ever answers that $\parallelbody{\calK}{-\eps}$ is empty then any vector with maximum value 1 will be a valid answer for the weak separation problem. Otherwise, we have obtained a box that contains $\parallelbody{\calK}{-\delta/2}$ and $\vecz\in \parallelbody{\calK}{\eps}$, and hence we obtain an $R'>0$ such that $\parallelbody{\calK}{-\delta/2}\subseteq \parallelbody{\mathbf{0}}{R'}$ and $\parallelbody{\mathbf{0}}{R'}\cap \calK \neq \emptyset$. Assuming that $\calK$ is polynomially bounded, the encoding length of $R'$ is also polynomially bounded in $R,\delta$. Let us define now $R:=3R'$ by a simple geometric argument one can prove:
\begin{claim}
It holds that
either $\parallelbody{\calK}{-\delta}=\emptyset$ or $\calK\subseteq\parallelbody{\mathbf{0}}{R}$
\end{claim}

We set $\calK':=\calK\cap \parallelbody{\mathbf{0}}{R}$. Note that $\calK'$ is by definition an upper-bounded circumscribed convex body. Now we can design a weak optimization subroutine for $\calK'$. More precisely, 
\begin{center}
\begin{minipage}{0.7\textwidth}
{\small  For any input $(f:\vecc^\top\vecx,\epsilon'>0)$, we call $\mathcal{A}_{\calK}(\vecc^\top\vecx,\min\{\delta,\epsilon'/10\})$. We may assume that $\calA_{\calK}$ does not give the answer $\parallelbody{\calK}{-\eps}=\emptyset$ nor does it give an output $\vecy'\in\parallelbody{\calK}{\epsilon}$ such that $\|\vecy'\|>R+\epsilon$, because in both cases we can conclude that $\parallelbody{\calK}{-\delta}$ equals empty set and thus the answer to the weak separation problem is trivial as we explain at the begin of the argument. So, without loss of generality, we assume $\mathcal{A}_{\calK}$ returns a vector $\vecy'\in\parallelbody{\calK}{\epsilon}\cap\parallelbody{\mathbf{0}}{R+\epsilon}$ such that $\vecc^\top \vecx\leq\vecc^\top\vecy +\epsilon $ for all $\vecc\in\parallelbody{\calK}{-\eps}$. Since it holds that $\parallelbody{\calK}{\eps}\cap\parallelbody{\mathbf{0}}{R+\eps}\subseteq \parallelbody{\calK'}{10\epsilon}$ and $\parallelbody{\calK'}{-\eps}\subseteq \parallelbody{\calK}{-\eps}$, we get that $\vecy'$ is a valid for the weak optimization problem of $\calK'$}
\end{minipage}
\end{center}
Thus, by the remark above we can solve the weak separation problem for $\calK'$ and the initial input $(\vecy,\delta)$. Note that $\parallelbody{\calK'}{\delta}\subseteq\parallelbody{\calK}{\delta}$ 
and by the aforementioned claim $\parallelbody{\calK'}{-\delta}=\parallelbody{\calK}{-\delta}$ and hence the output for the weak separation oracle for $\calK'$ is also valid for $\calK$.
\end{proof}
Theorem~\ref{theorem:OPT2SEP} permits reducing the construction of separation oracles for combination of convex sets, i.e., Minkowski Sum, Set Difference, Intersection just by constructing the corresponding weak optimization algorithms.

\begin{theorem}\label{theorem:OPTcombo}
There exist oracle-polynomial time algorithm that solve the following problems:
\begin{enumerate}
    \item Construction of a Weak Optimization oracle for the Minkowski sum of well-bounded convex sets  $\{\calK_1,\cdots,\calK_m\}$, given weak optimization algorithms for every convex set.
    \item Construction of a Weak Optimization oracle for the Intersection or Cartesian Product of well-bounded convex sets $\{\calK_1,\cdots,\calK_m\}$, given weak optimization algorithms for every convex set.
\end{enumerate}
\end{theorem}
\begin{proof}
We will prove the result for the $m=2$ and then by induction we can prove it for any polynomially bounded $m$.

We start our proof with the case of $\calK_1+\calK_2$. Firstly, it is easy to see that if $\calK_1,\calK_2$ are well bounded convex bodies, i.e.,
 $\exists \veca_1\in \R^d\ : \parallelbody{\veca_1}{r_1}\subseteq \calX
\subseteq\parallelbody{0}{R_1}$ and 
 $\exists \veca_2\in \R^d\ : \parallelbody{\veca_2}{r_2}\subseteq \calX
\subseteq\parallelbody{0}{R_2}$ then it holds that  
 $\exists \veca_3\in \R^d\ : \parallelbody{\veca_3}{r_1+r_2}\subseteq \calX
\subseteq\parallelbody{0}{R_1+R_2}$.
Moreover, if we have two weak optimization algorithms $\calA_{\calK_1},\calA_{\calK_2}$ for the corresponding sets $\calK_1,\calK_2$ for linear functions then the weak optimization problem for $\calK_1+\calK_2$ can be solved easily. In fact, notice that Theorem~\ref{theorem:OPT2SEP} provides a reduction from optimizing linear functions for constructing a separation oracle and then using subgradient cut and ellipsoid method (Theorem~\ref{thm:approx-minimization}) we can implement a weak optimization algorithm for any Lipschitz convex function.

Thus, let's assume that our objective is $f(\vecx)=\vecc^\top\vecx$, where $\vecc\in\mathbb{Q}^d$ and the demanded accuracy is $\epsilon>0$ (we may assume without loss of generality that $\|\vecc\|_\infty=1$. Then let's define $\epsilon_i:=\min\{r_i,\epsilon\cdot\tfrac{r_i}{8dR_i}$ for $i\in\{1,2\}$. We can call $\calA_{\calK_i}(\vecc^\top\vecx,\epsilon_i)$. 
This gives us vectors $\vecy_i$ such that $\vecy\in\parallelbody{\calK_i}{\eps_i}$ and for all $\vecx_i\in\parallelbody{\calK_i}{-\eps_i}$, it holds that 
$\vecc^\top\vecx_i\leq \vecc^\top\vecy_i+\eps_i$. Hence using similar argumentation with Lemma~\ref{lem:from-border-to-in}, we get that actually 
\[
\forall \vecx_i\in\calK_i\quad \vecc^\top\vecx_i\leq \vecc^\top\vecy_i+\eps_i +\frac{2R_i\epsilon_i\sqrt{d}}{r_i}\leq \vecc^\top\vecy_i+\frac{\eps}{2}
\]
We claim that $\vecy=\vecy_1+\vecy_2$ solves the weak optimization problem for $\calK_1+\calK_2$ for the objective $f(\vecx)=\vecc^\top\vecx$ and the accuracy parameter $\epsilon$. Indeed, trivially it holds that
\begin{enumerate}
    \item $\vecy\in\parallelbody{\calK_1}{\eps_1}+\parallelbody{\calK_2}{\eps_2}\subseteq\parallelbody{\calK}{\eps}$ 
    \item For any $\vecx\in\calK_1+\calK_2$, i.e, it can be written as $\vecx=\vecx_1+\vecx_2$, such that $\vecx_i\in\calK_i$, it holds that:
    \[\vecc^\top\vecx=\vecc^\top\vecx_1+\vecc^\top\vecx_2\leq \vecc^\top\vecy_1+ \eps/2 \vecc^\top\vecy_2 +\eps/2=\vecc^\top\vecy + \eps \]
\end{enumerate}

Similar proof holds for the intersection case. For details see \cite[Chapter 4.3]{grotschel2012geometric}.


\end{proof}

\end{document}

%% file: ec-23-headers.tex
\usepackage{subcaption}
\usepackage{tikz}
\usetikzlibrary{fadings}
\usetikzlibrary{patterns}
\usetikzlibrary{shadows.blur}
\usetikzlibrary{shapes}
\usepackage{wrapfig}

\usepackage{contour}
\usepackage{xparse}
\usepackage{mathtools}
\usepackage{mleftright}
\usepackage{enumitem}
\usepackage[%
linewidth=2pt,
linecolor=gray,
middlelinecolor= black,
middlelinewidth=0.4pt,
roundcorner=1pt,
topline = false,
rightline = false,
bottomline = false,
rightmargin=0pt,
skipabove=0pt,
skipbelow=0pt,
leftmargin=0pt,
innerleftmargin=4pt,
innerrightmargin=0pt,
innertopmargin=0pt,
innerbottommargin=0pt,
]{mdframed}

\DeclareMathOperator*{\argmax}{argmax}
\DeclareMathOperator*{\argmin}{argmin}

\DeclareMathOperator{\dist}{dist}

\newcommand{\ProbWAP}{\textsc{Weak Approximate Projection}}
\newcommand{\ProbSAP}{\textsc{Strong Approximate Projection}}
\newcommand{\ProbWSEP}{\textsc{Weak Separation Oracle} \text{ (via a circuit $\calC_{F(\vecx)}$)}}

\newcommand{\ProbSCCO}{\textsc{Strong Constrained Convex Optimization}}
\newcommand{\ProbWCCO}{\textsc{Weak Constrained Convex Optimization}}

\newcommand{\Kakutani}{\textsc{Kakutani}}
\newcommand{\ProbWKakutani}{\textsc{Kakutani with }\mathrm{WSO}_F}
\newcommand{\ProbSKakutani}{\textsc{Kakutani with }\mathrm{SO}_F}

\newcommand{\class}[1]{\ensuremath{\mathsf{#1}}\xspace}

\def\poly{\mathrm{poly}}

\newcommand{\FP}{\class{FP}}

\newcommand{\NP}{\class{NP}}
\newcommand{\FNP}{\class{FNP}}

\newcommand{\PPAD}{\class{PPAD}}

\newcommand{\FIXP}{\class{FIXP}}

\NewDocumentCommand\xDeclarePairedDelimiter{mmm}
 {%
  \NewDocumentCommand#1{som}{%
   \IfNoValueTF{##2}
    {\IfBooleanTF{##1}{#2##3#3}{\mleft#2##3\mright#3}}
    {\mathopen{##2#2}##3\mathclose{##2#3}}%
  }%
 }

\xDeclarePairedDelimiter\inner{\langle}{\rangle}
\xDeclarePairedDelimiter\ang{\langle}{\rangle}
\xDeclarePairedDelimiter\abs{\lvert}{\rvert}
\xDeclarePairedDelimiter\set{\{}{\}}
\xDeclarePairedDelimiter\p{(}{)}

\let\b=\relax
\xDeclarePairedDelimiter\b{[}{]}
\xDeclarePairedDelimiter\round{\lfloor}{\rceil}
\xDeclarePairedDelimiter\floor{\lfloor}{\rfloor}
\xDeclarePairedDelimiter\ceil{\lceil}{\rceil}
\xDeclarePairedDelimiter\norm{\lVert}{\rVert}

\newcommand{\calA}{\ensuremath{\mathcal{A}}}
\newcommand{\calB}{\ensuremath{\mathcal{B}}}
\newcommand{\calC}{\ensuremath{\mathcal{C}}}
\newcommand{\calD}{\ensuremath{\mathcal{D}}}
\newcommand{\calE}{\ensuremath{\mathcal{E}}}

\newcommand{\calH}{\ensuremath{\mathcal{H}}}
\newcommand{\calI}{\ensuremath{\mathcal{I}}}

\newcommand{\calK}{\ensuremath{\mathcal{K}}}
\newcommand{\calL}{\ensuremath{\mathcal{L}}}
\newcommand{\calM}{\ensuremath{\mathcal{M}}}
\newcommand{\calN}{\ensuremath{\mathcal{N}}}

\newcommand{\calQ}{\ensuremath{\mathcal{Q}}}
\newcommand{\calR}{\ensuremath{\mathcal{R}}}
\newcommand{\calS}{\ensuremath{\mathcal{S}}}
\newcommand{\calT}{\ensuremath{\mathcal{T}}}
\newcommand{\calU}{\ensuremath{\mathcal{U}}}
\newcommand{\calV}{\ensuremath{\mathcal{V}}}

\newcommand{\calX}{\ensuremath{\mathcal{X}}}
\newcommand{\calY}{\ensuremath{\mathcal{Y}}}

\newcommand{\Q}{\ensuremath{\mathbb{Q}}}
\newcommand{\R}{\ensuremath{\mathbb{R}}}

\newcommand{\diam}{\mathrm{diam}}
\newcommand{\metric}{\mathrm{d}}
\newcommand{\hausdorff}{\mathrm{d}_{\rmH}}
\newcommand{\size}{\mathrm{size}}
\newcommand{\nm}[1]{\b{#1} - 1}
\newcommand{\sign}{\mathrm{sign}}
\newcommand{\gridRefinement}{\xi}
\renewcommand{\bar}[1]{\overline{#1}}
\newcommand{\parallelbody}[2]{\ensuremath{\bar{\class{B}}(#1,#2)}}
\newcommand{\projectionto}[2]{\ensuremath{ \Pi_{#1}\left(#2\right)}}
\newcommand{\algoprojectionto}[3]{\ensuremath{\widehat{\Pi}_{#1}^{#3}\left(#2\right)}}
\newcommand{\strongalgoprojectionto}[3]{\ensuremath{\widetilde{\Pi}_{#1}^{#3}\left(#2\right)}}
\newcommand{\eps}{\varepsilon}

\renewcommand{\epsilon}{\varepsilon}
\newcommand{\parameters}{{\epsilon}}

\newcommand{\rmH}{\ensuremath{\mathrm{H}}}

\renewcommand{\vec}[1]{\boldsymbol{#1}}

\newcommand{\veca}{\ensuremath{\boldsymbol{a}}}
\newcommand{\vecb}{\ensuremath{\boldsymbol{b}}}
\newcommand{\vecc}{\ensuremath{\boldsymbol{c}}}

\newcommand{\vece}{\ensuremath{\boldsymbol{e}}}

\newcommand{\vecg}{\ensuremath{\boldsymbol{g}}}

\newcommand{\veck}{\ensuremath{\boldsymbol{k}}}

\newcommand{\vecp}{\ensuremath{\boldsymbol{p}}}
\newcommand{\vecq}{\ensuremath{\boldsymbol{q}}}

\newcommand{\vect}{\ensuremath{\boldsymbol{t}}}
\newcommand{\vecu}{\ensuremath{\boldsymbol{u}}}
\newcommand{\vecv}{\ensuremath{\boldsymbol{v}}}
\newcommand{\vecw}{\ensuremath{\boldsymbol{w}}}
\newcommand{\vecx}{\ensuremath{\boldsymbol{x}}}
\newcommand{\vecy}{\ensuremath{\boldsymbol{y}}}
\newcommand{\vecz}{\ensuremath{\boldsymbol{z}}}

\theoremstyle{plain}            
\newtheorem{inftheorem}{Informal Theorem}
\newtheorem{theorem}{Theorem}[section]
\newtheorem{lemma}[theorem]{Lemma}
\newtheorem{corollary}[theorem]{Corollary}

\newtheorem{claim}[theorem]{Claim}
\newtheorem{fact}[theorem]{Fact}

\theoremstyle{definition}       
\newtheorem{definition}[theorem]{Definition}

\theoremstyle{remark}           
\newtheorem{remark}[theorem]{Remark}

\numberwithin{equation}{section}

\contourlength{0.1pt}
\contournumber{10}

\newenvironment{oracle}[1][\unskip]{%
\medskip
\begin{mdframed}
\noindent
\contour{black}{\underline{$#1$}} \\
\noindent}
{\end{mdframed}}

\newenvironment{nproblem}[1][\unskip]{%
\medskip
\begin{mdframed}
\noindent
\contour{black}{\underline{$#1$ Problem.}} \\
\noindent}
{\end{mdframed}}